\documentclass[12pt]{article}
\usepackage[margin=1.0in]{geometry}
\usepackage[utf8]{inputenc}
\usepackage[intlimits]{amsmath}
\usepackage{amssymb,amsfonts,amsthm,array}
\usepackage{physics}
\usepackage[english]{babel}
\usepackage{rotating}
\usepackage{mathrsfs}
\usepackage{empheq}
\usepackage[bottom]{footmisc}

\usepackage{csquotes}

\newtheorem{theorem}{Theorem}[section]
\newtheorem{lemma}[theorem]{Lemma}

\renewcommand{\bar}[1]{\overline{#1}}

\renewcommand{\dd}{\partial}
\newcommand{\rhs}{{\it r.h.s.} }
\newcommand{\lhs}{{\it l.h.s.} }
\numberwithin{equation}{section}
\newcommand{\ie}{{\it i.e.,} }
\newcommand{\ls}{{\!\!\!\!\!\!} }

\begin{document}

\vspace{1.7 cm}

\begin{flushright}

{\small FIAN/TD/10-2021}
\end{flushright}
\vspace{1.7 cm}

\begin{center}
{\large\bf The $\sigma_-$ Cohomology Analysis  for
 Symmetric
Higher-Spin Fields}

\vspace{1 cm}

{\bf A.S.~Bychkov$^{1,2}$, K.A.~Ushakov$^{1,2}$ and  M.A.~Vasiliev$^{1,2}$}\\
\vspace{0.5 cm}
\textbf{}\textbf{}\\
 \vspace{0.5cm}
 \textit{${}^1$ I.E. Tamm Department of Theoretical Physics,
Lebedev Physical Institute,}\\
 \textit{ Leninsky prospect 53, 119991, Moscow, Russia}\\

\vspace{0.7 cm} \textit{
${}^2$ Moscow Institute of Physics and Technology,\\
Institutsky lane 9, 141700, Dolgoprudny, Moscow region, Russia
}

\end{center}

\vspace{0.4 cm}

\begin{abstract}
\noindent
     In this paper, we present a complete proof of the so-called First On-Shell Theorem that determines
     dynamical content of the unfolded equations for free symmetric  massless fields of arbitrary integer
     spin in any dimension and  arbitrary  integer or  half-integer spin in four dimensions. This is achieved
      by calculation of the respective  $\sigma_-$ cohomology both in the tensor language in  Minkowski space of any dimension and in terms of spinors in  $AdS_4$. In the $d$-dimensional case $H^p(\sigma_-)$ is computed for any $p$ and interpretation of $H^p(\sigma_-)$ is given both for the original Fronsdal system and for
      the associated systems of higher form fields.

\end{abstract}

\newpage

\vspace{-1cm}
\tableofcontents

\newpage

\section{Introduction}

Higher-spin (HS) gauge theory  is based on works of Fronsdal \cite{Fronsdal:1978rb} and Fang and Fronsdal \cite{Fang:1978wz}, where the action and  equations of motion for massless gauge fields of any spin were originally  obtained in
flat four-dimensional Minkowski space.
Even earlier, important restrictions on low-energy HS vertices were obtained
by Weinberg in \cite{Weinberg:1964ev,Weinberg:1964ew}
and so-called no-go theorems   restricting  $S$-matrix
possessing too high  symmetries in  flat space-time
were proven in  \cite{Coleman:1967ad, Haag:1974qh}. (For a review see \cite{Bekaert:2010hw}.)
 The no-go theorems implied the existence of the $s=2$ barrier suggesting that the construction of an interacting local HS theory  in  Minkowski space-time is  impossible. The proof of these theorems essentially uses the specific form of the algebra of isometries of  Minkowski space. The $s=2$ barrier in  flat space can be overcome in the space-time with non-zero sectional curvature,  for example, in the anti-de Sitter space \cite{Fradkin:1987ks}. In these spaces  it becomes possible to formulate a  consistent nonlinear   theory of fields of all spins \cite{Vasiliev:1992av,Vasiliev:2003ev}.

 The construction of a nonlinear HS theory is essentially based on the so-called unfolded approach   \cite{Vasiliev:1988xc, Vasiliev:1988sa}, which is a far-going generalization of the Cartan formulation of gravity ($s=2$) in terms of differential forms to fields of any  spin $s>2$.   Via introducing appropriate auxiliary variables, the unfolding procedure allows one to replace the system of partial differential equations of any order on a smooth manifold by a larger system of first-order equations on vector-valued differential forms. One of the essential features of this approach, which is very useful    for analysing symmetries of a given system, is that the variables in the equations are valued in one or another representation of the underlying symmetry algebra.

The dynamical content of the HS theory can be reconstructed from its unfolded formulation using the $\sigma_-$ cohomology technique \cite{Shaynkman:2000ts}. As is recalled below, the dynamical data of the theory are in one-to-one correspondence with the cohomology of certain linear nilpotent operator   $\sigma_-$ that can be read of the unfolded equations in question. The statement that unfolded equations of free HS fields are equivalent to the Fronsdal equations was made in the original papers in the spinor
\cite{Vasiliev:1986td} and   tensor  \cite{Lopatin:1987hz} formalisms. In the tensor formulation of HS theory
the idea of the proof was illustrated in \cite{Bekaert:2005vh}, where however the analysis of the
trace part of the Fronsdal equations was not completed, while  general arguments for mixed-symmetry
HS fields were given in \cite{Skvortsov:2009nv}. In \cite{Barnich:2004cr} the unfolded equations for massless fields were derived from the Fronsdal theory
by  the BRST methods. To the best of our knowledge, no detailed analysis of the problem in the spinor formalism was available in the literature.

 In this paper we present a complete proof of the so called First On-Shell Theorem by computation of the cohomology rings of $\sigma_-$ for the physically important cases of the integer-spin symmetric fields both in flat space-time of any dimension and $AdS_d$ as well as for the fields of any integer and half-integer spin in  ${AdS}_4$. The computation technique analogous  to the Hodge theory for differential forms is performed in terms of so-called $\sigma_-$ cohomology  and provides a complete analysis of the dynamical content of the free unfolded equations for symmetric massless fields of any spin.
 Giving a direct proof of the equivalence between the Fronsdal formulation of the
 HS gauge theory and its unfolded formulation this paper fills in some gaps in the literature also illustrating a general approach applicable to a broad class of
 unfolded systems. In addition, in the tensor case we compute higher $\sigma_-$ cohomology groups and interpret them in terms of higher Bianchi identities and
 more general dynamical systems. In particular, we discuss the  matching between
 the Bianchi identities in terms of one-form gauge fields and zero-form field strengths.

The rest of the paper is organized as follows. In Section 2 we  briefly recall  different approaches to the description of HS massless fields. Main idea of the
$\sigma_-$ cohomology approach is explained in
Section 3. Cohomology calculation method used in this paper is discussed in Section 4. Section 5 contains derivation of $H^p(\sigma_-)$ in Minkowski space of any dimension. In particular, the cases of $\mathrm{GL}(d)$ and $\mathrm{O}(d)$ representations are analysed here. In Section 6 calculation of the low-order cohomology groups in $AdS_4$ is performed. Obtained results are discussed in Section 7. Index conventions and normalisations of the tensor Young diagrams are presented in the Appendix A.

\section{Fronsdal theory}

\subsection{Metric formulation}
\label{sec:fronsdal_form}
According to Fronsdal \cite{Fronsdal:1978rb},  a spin-$s$  massless symmetric field can be
described in terms of two symmetric traceless tensors (for index conventions see Appendix A)
\begin{equation}
\phi^{a(s)} \equiv \phi^{a_1..a_s}, \quad \phi^{a(s-2)} \equiv \phi^{a_1..a_{s-2}}, \quad \eta_{b_1 b_2}\phi^{b_1 b_2 a_3 .. a_s} \equiv \phi^{a(s-2)b}{}_b = 0,  \quad \phi^{a(s-4)b}{}_b = 0\,.
\end{equation}

These two fields can be combined into a single rank-$s$ totally symmetric tensor
\begin{equation}\label{e:fronsdalfieldcomp}
\varphi^{a(s)} = \phi^{a(s)} + \eta^{aa}\phi^{a(s-2)}
\end{equation}
obeying the double-tracelessness condition
\begin{equation}
\varphi^{a(s-4)bc}{}_{bc} = 0\,.
\end{equation}

The field equations in Fronsdal theory have the form
\begin{equation}\label{e:fronsdaleq}
R^{a(s)}(\varphi) = \square \varphi^{a(s)} - s \partial^a \partial_k \varphi^{k a(s-1)} + \frac{s(s-1)}{2} \partial^a \partial^a \varphi^{a(s-2)k}{}_k = 0\,,
\end{equation}
where $\partial_a = \frac{\partial}{\partial x^a}$.

The tensor $R^{a(s)}(\varphi)$ is invariant under the gauge transformations with a
rank-$(s-1)$ traceless gauge parameter $\varepsilon(x)$
\begin{equation}\label{e:fronsdalgauge}
\delta \varphi^{a(s)} = \partial^a \varepsilon^{a(s-1)}, \qquad \varepsilon^{a(s-3)k}{}_k = 0\,.
\end{equation}

Fronsdal equation~\eqref{e:fronsdaleq} is a generalization of the well-known equations of fields with spins $s = 0,1,2$. For the case of $s = 1$ the last term in the Fronsdal tensor disappears and Eq.(\ref{e:fronsdaleq}) reproduces Maxwell equations for the field $A^a$. Without the last two terms at $s = 0$ it gives
Klein-Gordon equation for a massless scalar field. The case of $s = 2$  reproduces  the equations of linearized gravity \cite{Fierz:1939ix}. Gauge transformation~\eqref{e:fronsdalgauge} gives the known gauge transformations of low-spin fields and its absence for a scalar field.

\subsection{Frame-like formulation}

\subsubsection{Tensor formalism}
The unified description of massless fields of arbitrary spin can be given
in the so-called frame-like formalism that generalizes Cartan formulation of gravity, operating  in terms of differential forms \cite{Vasiliev:1986td,Lopatin:1987hz,Vasiliev:1980as}.
Frame-like formulation of the HS gauge theory in any dimension is given in terms  of the one-form fields
$ \omega^{a(s-1),b(t)} =dx^\nu \omega_\nu^{a(s-1),b(t)}$
valued in two-row Young diagrams corresponding to irreducible $\mathfrak{o}(d-1,1)$ (i.e., traceless) modules
\cite{Lopatin:1987hz},  obeying conditions
\begin{align}
& \omega^{a(s-1),a b(t-1)} = 0, \\
& \omega^{a(s-3)k}{}_{k}{}^{,b(t)} = 0\,.
\end{align}
(For index conventions see Appendix A.)

By introducing auxiliary fields it is possible to put
a system of partial differential equations into the first-order {\it unfolded} form
\cite{Vasiliev:1988xc,Vasiliev:1988sa}.
Generally, unfolded equations read as
\begin{equation}
    d W^A = \sum_{n = 1}^{\infty} G^A_{B_1,..,B_n}W^{B_1}\wedge...\wedge W^{B_n}\,,\qquad d:=dx^\nu \partial_\nu\,.
\end{equation}
Here $W^A$ is a set of differential forms over some manifold. (Indices are treated  formally and can take an infinite number of values.) The coefficients $G^A_{B_1,..,B_n}$ satisfy the (anti)symmetry condition
\begin{equation}
    G^A_{B_1,..,B_i,..,B_j,..,B_n}  =  (-1)^{|B_i||B_j|} G^A_{B_1,..,B_j,..,B_i,..,B_n}\,,
\end{equation}
where $|B_i|$ denotes the  form-degree of $W^{B_i}$. Also $G^A_{B_1,..,B_n}$ are restricted by the integrability conditions expressing that $d^2=0$.

In the tensor language the unfolded HS equations in Minkowski space proposed in \cite{Lopatin:1987hz} read as
\begin{equation}\label{unfoldedeq_1}
    D_L\omega^{a(s-1),b(t)} + h_m \wedge \omega^{a(s-1),b(t)m} = 0, \quad t \in \{0,..,s-2\},
\end{equation}
\begin{equation}\label{unfoldedeq_2}
    D_L\omega^{a(s-1),b(s-1)} = h_n \wedge h_m \wedge C^{a(s-1)n,b(s-1)m},
\end{equation}
where $h_n$ is a soldering form (vielbein, frame field, tetrad) and
$D_L = d + \varpi$ is the background Lorentz covariant derivative that satisfies  relations
\begin{equation}
D_L h^a = 0, \qquad D_L^2 = 0\,.
\end{equation}
In the Cartesian coordinate system with $\varpi=0$ the equations simplify to
\begin{equation}\label{unfoldedeq_1}
    d\omega^{a(s-1),b(t)} + h_m \wedge \omega^{a(s-1),b(t)m} = 0, \quad t \in \{0,..,s-2\},
\end{equation}
\begin{equation}\label{unfoldedeq_2}
    d\omega^{a(s-1),b(s-1)} = h_n \wedge h_m \wedge C^{a(s-1)n,b(s-1)m}\,,
\end{equation}
where $C$ satisfies the Lorentz irreducibility conditions
\begin{equation}
\label{CYT}
    C^{a(n),a b(m-1)} = 0\,, \qquad C^{a(n-2)k}{}_k{}^{,b(m)} = 0\,.
\end{equation}

The traceless tensor $C$ on the \rhs
 of (\ref{unfoldedeq_2}) is a generalized Weyl tensor.  There are also unfolded equations on $C$
  and on additional
auxiliary fields \cite{Vasiliev:2003ev} (for reviews see \cite{Bekaert:2005vh, Didenko:2014dwa}). This system
constitutes an infinite chain of zero-form equations. Zero-form sector, that contains equations on
spin-zero and spin-one fields,  will not be considered in this paper.

Equations (\ref{unfoldedeq_1}) are invariant under the gauge transformations
\begin{equation}
    \delta \omega^{a(s-1),b(t)} = d\varepsilon^{a(s-1),b(t)} + h_m  \varepsilon^{a(s-1),b(t)m}, \quad t \in \{0,..,s-2\}
\end{equation}
and eq.(\ref{unfoldedeq_2}) is invariant under
\begin{align}
&\delta \omega^{a(s-1),b(s-1)} = d \varepsilon^{a(s-1),b(s-1)},\\
&\delta C^{a(s),b(s)} = 0,
\end{align}
where $\varepsilon$ are zero-forms valued in the appropriate two-row irreducible $\mathfrak{o}(d-1,1)$-modules
obeying conditions analogous to (\ref{CYT}).

The Fronsdal field is embedded into the frame-like one-form $e^{a(s-1)}
\equiv \omega^{a(s-1)}$ in the following manner. Converting the form
index into the fiber one using vielbein $h$,
\begin{equation}\label{e:frame-fibre}
e^{a(s-1)|b} = e^{a(s-1)}_\mu h^{\mu b}\,,
\end{equation}
the resulting tensor \eqref{e:frame-fibre} can be decomposed into  irreducible $\mathfrak{o}(d-1,1)$-modules. In terms  of traceless Young diagrams this decomposition is
\begin{equation}
\begin{picture}(15,10)(0,7)
{
\put(05,10){\line(1,0){10}}%
\put(05,20){\line(1,0){10}}%
\put(05,10){\line(0,1){10}}%
\put(15,10){\line(0,1){10}}%
}
\end{picture}   \,\,\,\otimes_{so}\,\,\,
\begin{picture}(65,18)(10,7)
{\put(25,13){\scriptsize  ${s-1}$}
 \put(05,20){\line(1,0){60}}%
\put(05,10){\line(1,0){60}}%
\put(05,10){\line(0,1){10}}%
\put(65,10){\line(0,1){10}}%
}
\end{picture}
\,\ \cong\,\,\,\
\begin{picture}(65,18)(10,7)
{\put(30,13){\scriptsize  ${s}$}
 \put(05,20){\line(1,0){60}}%
\put(05,10){\line(1,0){60}}%
\put(05,10){\line(0,1){10}}%
\put(65,10){\line(0,1){10}}%
}
\end{picture}
\,\ \oplus \,\,\,\
\begin{picture}(65,18)(10,7)
{\put(25,13){\scriptsize  ${s-2}$}
 \put(05,20){\line(1,0){60}}%
\put(05,10){\line(1,0){60}}%
\put(05,10){\line(0,1){10}}%
\put(65,10){\line(0,1){10}}%
}
\end{picture}
\,\,\,\ \oplus \,\,\,\
\begin{picture}(10,18)(0,7)
{
\put(05,00){\line(1,0){10}}%
\put(05,10){\line(1,0){10}}%
\put(05,0){\line(0,1){10}}%
\put(15,0){\line(0,1){10}}%

}
\end{picture}\begin{picture}(65,18)(10,7)
{\put(25,13){\scriptsize  ${s-1}$}
 \put(05,20){\line(1,0){60}}%
\put(05,10){\line(1,0){60}}%
\put(05,10){\line(0,1){10}}%
\put(65,10){\line(0,1){10}}%
}
\end{picture}
\,.
\end{equation}

The first two components give the  Fronsdal field, while the third one is an excess of the components of the frame-like
 field in comparison with the Fronsdal field. At the tensor level, this decomposition is represented as:
\begin{equation}
e^{a(s-1)|b} = \psi_1^{a(s-1)b} + \beta_1\eta^{aa}\psi_2^{a(s-3)b} + \beta_2\eta^{ab}\psi_2^{a(s-2)} + \psi_3^{a(s-1),b},
\end{equation}
where $\psi_i$ are traceless and correspond to the $i$-th diagram. The relative coefficient $\frac{\beta_2}{\beta_1}$ is fixed by the tracelessness condition with respect to indices $a$.

This decomposition shows that the Fronsdal field identifies with the symmetric part
of the frame-like field, since the contribution of the third diagram disappears upon
symmetrization. The resulting field
\begin{equation}
\varphi^{a(s)} := e^{a(s-1)|a}\,
\end{equation}
is symmetric and double-traceless. The extra term $\psi_3^{a(s-1),b}$ is pure  gauge. Its contribution can be canceled by the gauge transformation $\delta e^{a(s-1)|b} =  \varepsilon^{a(s-1)|b} $ with suitable
gauge parameter. For detailed discussion of Fronsdal field embedding see \cite{ Bekaert:2005vh, Vasiliev:1980as, Didenko:2014dwa}.

It is not difficult to check \cite{Lopatin:1987hz, Vasiliev:1980as} (for reviews see
\cite{Bekaert:2005vh, Didenko:2014dwa}) that the Fronsdal equations and gauge transformations
follow from the unfolded system (\ref{unfoldedeq_1}), (\ref{unfoldedeq_2}). A more complicated question is whether the  Fronsdal fields and equations are the only ones
that result from (\ref{unfoldedeq_1}), (\ref{unfoldedeq_2}).
The answer can be obtained via the $\sigma_-$ cohomology technique
 \cite{Shaynkman:2000ts}.

\subsubsection{Spinor language in $AdS_4$}
\label{spinform}
The physically important  case of the unfolded system for HS connection (\ref{unfoldedeq_1}), (\ref{unfoldedeq_2}) is that of $AdS_4$ space-time in which case the language of two-component spinors is most appropriate. In this language instead of using Lorentz indices $a,b,...= 0,1,2,3$, one uses two pairs of dotted and undotted spinor indices $\alpha,\beta,...$ and $\dot\alpha,\dot\beta,...$ taking  values  $\{1,2\}$. The two languages
are related via Pauli matrices. The $AdS_4$ background geometry is described in terms of the Lorentz
connection $\varpi$ and frame field $h$, that satisfy equations
\begin{equation}
\begin{array}{l}
d h^{\alpha \dot{\beta}}+\varpi^{\alpha}{ }_{\gamma} \wedge h^{\gamma \dot{\beta}}+\bar{\varpi}^{\dot{\beta}}{}_{\dot{\gamma}} \wedge h^{\alpha \dot{\gamma}}=0\,, \\
d \varpi^{\alpha \beta}+\varpi^{\alpha}{ }_{\gamma} \wedge \varpi^{\gamma \beta}=-\lambda^{2} h^{\alpha}{}_{\dot{\gamma}} \wedge h^{\beta \dot{\gamma}}\,, \\
d \bar{\varpi}^{\dot{\alpha} \dot{\beta}}+\bar{\varpi}^{\dot{\alpha}}{}_{\dot{\gamma}} \wedge \bar{\varpi}^{\dot{\gamma} \dot{\beta}}=-\lambda^{2} h_{\gamma}{ }^{\dot{\alpha}} \wedge h^{\gamma \dot{\beta}}\,,
\end{array}
\end{equation}
where $\lambda^2$ is proportional to the curvature of $AdS_4$ and
we adopt the following rules
\begin{equation}
A_{\alpha}=A^{\beta} \epsilon_{\beta \alpha}\,, \quad A^{\alpha}=\epsilon^{\alpha \beta} A_{\beta}\,, \quad \epsilon_{\alpha \beta} \epsilon^{\gamma \beta}=\epsilon_{\alpha}{}^{\gamma}=\delta_{\alpha}^{\gamma}=-\epsilon^{\gamma}{}_{\alpha}\,,
\end{equation}
where
\begin{equation}
    \epsilon_{\alpha\beta} = -\epsilon_{\beta\alpha}\,, \quad \epsilon_{12} = 1\,.
\end{equation}

The spinor version of the unfolded system for one-form $\omega$ reads as follows. First,
the HS curvatures in the spinor language are \cite{ Vasiliev:1986td}
\begin{equation}
   R^{\alpha(n),\dot\alpha(m)}=D_L\omega^{\alpha(n),\dot\alpha(m)} + \lambda^2(n h^\alpha_{\ \dot\gamma}\wedge \omega^{\alpha(n-1),\dot\gamma\dot\alpha(m)} + m h_\gamma^{\ \dot\alpha}\wedge\omega^{\gamma\alpha(n),\dot\alpha(m-1)})
\end{equation}
 and $D_L=d+\varpi +\bar\varpi$ is a Lorentz-covariant derivative with
 Cartan's spin-connection $(\varpi\oplus\bar\varpi)$
\begin{equation*}
    D_L \omega^{\alpha(n),\dot\alpha(m)}=d\omega^{\alpha(n),\dot\alpha(m)}+n\varpi_\alpha^{\ \beta}\wedge\omega_{\beta\alpha(n-1),\dot\alpha(m)} + m\bar\varpi_{\dot\alpha}^{\ \dot\beta}\wedge\omega_{\alpha(n),\dot\beta\dot\alpha(m-1)}\,.
\end{equation*}
The $AdS_4$ deformation of the unfolded equations (\ref{unfoldedeq_1}), (\ref{unfoldedeq_2})
then takes the form \cite{Vasiliev:1988sa}
\begin{equation}
    R^{\alpha(n),\dot\alpha(m)}=\delta_{0,n}\,h_{\beta\dot\alpha}\wedge h^{\beta}_{\ \dot\alpha}\bar C^{\dot\alpha(m+2)} + \delta_{0, m}\, h_{\alpha\dot\beta}\wedge h_{\alpha}^{\ \dot\beta} C^{\alpha(n+2)}.
\end{equation}

The main advantage of the two-component spinor notation is that it makes  the representation theory of the
Lorentz group very simple. Namely, every  Lorentz irreducible multispinor
representing a traceless tensor is totally symmetric in its spinor indices.
  Since the only Lorentz invariant objects are antisymmetric bispinors  $\epsilon_{\alpha\beta}$ and $\bar\epsilon_{\dot\alpha\dot\beta}$ irreducible multispinors $X^{\alpha(n),\dot\alpha(m)}$ are necessarily symmetric with respect to the indices in the groups  $\alpha(n)$ and $\dot\alpha(m)$ separately.  Thus, working with the two-component spinor notation one can happily forget about painful calculations with the traces of Lorentz-tensors.

\section{The idea of $\sigma_-$ cohomology analysis: example of integer spin massless fields}
\label{gs}

The \lhs\!\!'s of unfolded HS equations and gauge transformations in $d$-dimensional Minkowski space
are~\cite{Lopatin:1987hz, Bekaert:2005vh}
\begin{align}
&R^{a(s-1),b(k)} = D_L \omega^{a(s-1),b(k)} + \sigma_-(\omega)^{{a(s-1),b(k)}},\\
&\delta \omega^{a(s-1),b(k)} = D_L \varepsilon^{a(s-1),b(k)} + \sigma_-(\varepsilon)^{{a(s-1),b(k)}},
\label{2}
\end{align}
where
\begin{equation}
    (\sigma_-\omega)^{{a(s-1),b(k)}} := h_c \wedge \omega^{a(s-1),b(k)c}\,,\qquad
    (\sigma_-\varepsilon)^{{a(s-1),b(k)}} := h_c \wedge \varepsilon^{a(s-1),b(k)c} .
\end{equation}
$R^{a(s-1),b(k)}$ is referred to as  (linearized) HS curvature. For simplicity we study the  Minkowski case. Since $\sigma_-$ in $AdS_d$   is defined analogously, our
 analysis applies to that case as well.

Due to their definition, HS curvatures obey  the Bianchi identities
\begin{equation}
D_L R^{a(s-1),b(k)} + \sigma_-(R)^{a(s-1),b(k)} = 0\,.
\end{equation}

The appearance of $\sigma_-$ allows one to clarify the role of the
fields $\omega^{a(s-1),b(k)}$ and gauge parameters $\varepsilon^{a(s-1),b(k)}$.  Working with the zero-forms $\varepsilon^{a(s-1),b(k)}$ and one-forms $\omega^{a(s-1),b(k)}$valued in two-row Young diagrams, we consider the space $V^p$ of $p$-forms valued  in two-row Young diagrams with any $p$. Defining $\sigma_-$ to annihilate the forms with an empty second row, we find that
$\sigma_- V^p\subset V^{p+1}$ and $\sigma_-\,\sigma_-=0$. As originally proposed in  \cite{Shaynkman:2000ts},
 the $\sigma_-$ cohomology $H(\sigma_-)=\ker({\sigma_-})/{\mathrm{im\,}(\sigma_-)}$ classifies fields,
 their equations and gauge symmetries.

Indeed, those components of the fields $\omega^{a(s-1),b(t)}$, that are not annihilated by $\sigma_-$, can be expressed
via derivatives of the fields with lower $t$  by setting suitable components of
the  HS  curvatures to zero. Such fields
are called auxiliary. Conversely, those components of the fields $\omega^{a(s-1),b(t)}$, that cannot be expressed
in terms of derivatives of lower fields via zero-curvature conditions, are in $\ker (\sigma_-)$.
By Stueckelberg fields we mean $\sigma_-$-exact fields (i.e. fields of the form $\sigma_-\chi$)
as they can be eliminated by an appropriate $\sigma_-$-exact term in the gauge transformation (\ref{2}).
Fields that are not expressed via derivatives of other fields and are not Stueckelberg are
called dynamical. These describe the physical degrees of freedom of the theory.
Thus, the dynamical HS fields are associated with $H^1(\sigma_-)$.

The classification for the gauge parameters is analogous. The parameters, that are not
annihilated by $\sigma_-$, describe algebraic Stueckelberg shifts. The leftover symmetries are described by the parameters  in $\ker (\sigma_-)$. $\sigma_-$-exact parameters correspond to the so called gauge for gauge transformations. Parameters, which are $\sigma_-$-closed and not $\sigma_-$-exact, are referred to as genuine differential gauge parameters.
Note that since in the HS example in question the gauge parameters are zero-forms there is no room for gauge for gauge symmetries in that case.

Let $V$ be a graded vector space, $\mathcal{C}$ be an element of $\Lambda^p(\mathcal{M}^d)
\otimes V$ over some smooth $d$-dimensional manifold $\mathcal{M}^d$. We demand the grading of
$V$ to be bounded from below, that is $V$ is $\mathbb{N}$-graded. Let $\sigma_ {\pm}$
be  operators that act ''vertically'', i.e. do not affect the space-time coordinates,
and shift grading by $\pm 1$, $D_L$ be the Grassmann-odd operator that does not affect
the grading  and is allowed to act non-trivially on the space-time coordinates.
Consider the covariant constancy condition of a general form along with
the zero-curvature condition
\begin{equation}\label{e:covariantconsteq}
\mathcal{D} \mathcal{C}  = (D_L + \sigma_- + \sigma_+) \mathcal{C} = 0, \quad \mathcal{D}^2 = 0.
\end{equation}
Notice that eq.\eqref{e:covariantconsteq} remains invariant under the gauge transformations
\begin{equation}\label{e:gaugetransformcov}
\delta \mathcal{C} = \mathcal{D} \varepsilon ,
\end{equation}
where  $\varepsilon \in \Lambda^{p-1}(\mathcal{M}^d) \otimes V$.

One can prove the following proposition \cite{Shaynkman:2000ts} (see also \cite{Bekaert:2005vh, Gelfond:2003vh, Vasiliev:2003ar, Vasiliev:2009ck}):
\begin{theorem}\label{Theorem_cohomology}
The following is true:

1) Differential gauge symmetry parameters $\varepsilon$ span  $H^{p-1}(\sigma_-)$

2) Nontrivial dynamical fields $\mathcal{C}$ span $H^p(\sigma_-)$

3) Physically distinguishable differential field equations on the nontrivial dynamical fields,
contained in $\mathcal{D} \mathcal{C}  = 0$, span  $H^{p+1}(\sigma_-)$
\end{theorem}

Thus, taking into account that HS gauge fields
are described by the one-forms $\omega$, to prove that the Fronsdal metric formulation is equivalent to
 the unfolded one, we have to calculate $H^{0}(\sigma_-), H^{1}(\sigma_-)$ and $H^{2}(\sigma_-)$.
More generally, higher cohomology $H^{k}(\sigma_-)$ with $k>p+1$ describes Bianchi identities for dynamical
equations at $k=p+2$ and Bianchi for Bianchi identites at $k>p+2$ \cite{Vasiliev:2009ck}. Similarly, the lower cohomology $H^{k}(\sigma_-)$ with $k<p-1$ describes gauge for gauge differential symmetries.

\section{A method for calculating cohomology}

Calculation of $\sigma_-$ cohomology is of utter importance for the analysis of unfolded
systems of the general form \eqref{e:covariantconsteq}. The straightforward calculation
of the cohomology can be quite involved.
In this paper we find cohomology using a standard homotopy approach recalled below,
that is a generalization of
the Hodge theory for de Rham cohomology extendable to a  more general class of (co)chain
complexes. Main details of the construction used in this paper follow those of \cite{Vasiliev:2009ck}, where
the $\sigma_-$ cohomology analysis was applied to the conformal HS theories of the bosonic
fields of any symmetry type. Unfortunately, some of the methods of \cite{Vasiliev:2009ck},
 based on the fact that $\sigma_-$ in conformal theories  has the clear meaning in terms of
 the conformal
algebra, are not directly applicable to the non-conformal HS theories discussed in this paper, which makes the analysis of the latter a bit more involved.

 Let $V$ be a graded vector space and $d$ be a linear operator of degree +1 on $V$
 (that is, it raises the grading of a homogeneous element by 1) such that $d^2 = 0$.
 Then  $H(d) = \ker (d)/ \mathrm{im\,}(d)$. Let $\partial$ be another operator of degree $-1$ on $V$
 (\ie it lowers the grading by 1) such that $\partial^2 = 0$. The operators
 $d$ and $\partial$ can be used to
 compose the degree 0 operator  $\Delta$
\begin{equation}
\label{De}
\Delta := \{d,\partial\} = d\,\partial + \partial\,d.
\end{equation}
It is easy to see that $\Delta$ satisfies
\begin{equation}
[d,\Delta] = [\partial,\Delta] = 0\,.
\end{equation}

\begin{lemma}\label{lem1}
If  $\Delta$ is diagonalizable on the (graded) vector space $V$,
then $H(d) \hookrightarrow \ker (\Delta)$.
\end{lemma}
\begin{proof}
First of all we should show that $\ker (d)$ is an invariant subspace of $\Delta$. Suppose $f \in \ker (d)$. Then
\begin{equation}
    \Delta f = (d \partial + \partial d) f = d \partial f \Rightarrow \Delta f \in \ker (d)\,, \forall f \in \ker (d)\,.
\end{equation}
Therefore, $\ker (d)$ is an invariant subspace of $\Delta$, because linearity is obvious.\\
Since the operator $\Delta $ is diagonalizable by assumption, we can consider eigenvectors of $\Delta$. Let $g$ be $d$-closed and $\Delta g = \lambda g, \lambda \neq 0$. Then
\begin{equation}
g = \frac{1}{\lambda}\Delta g= \frac{1}{\lambda}\,d\, \partial g\,.
\end{equation}
Hence, $g$ is also $d$-exact for $\lambda \neq 0$, representing a trivial element of $H(d)$. Thus, every $d$-closed form annihilated by $\Delta$ is not $d$-exact. In other words, every $d$-closed form $f$ can be written as $f = h + d\alpha$ with some $h\in\ker(\Delta)$.
\end{proof}

If $V$ is a Hilbert space with inner product $\langle\,,\,\rangle$, there exists such  $\partial$ that the converse inclusion $H(d) \hookleftarrow \ker (\Delta)$ takes place as well, which means that $H(d) = \ker (\Delta)$.

\begin{lemma}\label{lem2}
Let $(V,\langle\,,\,\rangle)$ be a Hilbert space, let $d^*$ be the operator conjugated
to $d$ in the usual sense $\langle\alpha, d\beta\rangle = \langle d^*\alpha, \beta\rangle$ and $\Delta = \{d\,, d^*\}$. Then $\ker (\Delta) \hookrightarrow H(d)$.
\end{lemma}
\begin{proof}
Take any $f \in \ker (\Delta)$. Then
\begin{equation}\label{e:laplaceddconj}
0 = \langle f,\Delta f\rangle = \langle df,df\rangle + \langle d^*f,d^*f\rangle \Leftrightarrow df = 0 \quad \text{and} \quad d^*f = 0\,.
\end{equation}
Hence, $f \in \ker (d)$.
To show that $f \notin \mathrm{im\,} (d)$ suppose the opposite. Let $f = dg$. Then due to \eqref{e:laplaceddconj}
\begin{equation}
d^*dg = 0 \Rightarrow 0 = \langle g, d^* dg\rangle = \langle dg,dg\rangle \Rightarrow dg = 0\,.
\end{equation}
Thus, $\ker (\Delta) \hookrightarrow H(d)$
\end{proof}

From Lemmas \ref{lem1} and \ref{lem2}, it follows that, if all the requirements are met,
\begin{equation}
\label{kD}
H(d) = \ker (\Delta)\,.
\end{equation}

Thus, in a Hilbert space with a diagonalizable  Laplace operator  $\Delta:=\{d\,,d^*\}$,
 finding the cohomology is equivalent to finding
 $\ker (\Delta)$. Further calculations of $\sigma_-$ cohomology will rely on this fact.

The following important comment  \cite{Vasiliev:2009ck} is now in order. In the case of interest,
for every unfolded subsystem associated with a fixed spin
$$
V= \oplus_n V_n
$$
with finite-dimensional grade-$n$ subspaces $V_n$. In that case
$\Delta$ leaves invariant every $V_n$ and, being self-adjoint in the finite-dimensional
Hilbert space, is diagonalizable.

It is worth noting the similarity of the above analysis with the Hodge theory mentioned at the beginning of this section.
 Indeed, consider the (finite-dimensional) vector space $V$ endowed with some nilpotent   operators $d$ and $\partial$, $d^2 = \partial^2 = 0$. The condition of disjointness is also imposed (see \cite{Kostant} for details), that is, $\mathrm{im}(d)\cap\ker(\partial) = \mathrm{im}(\partial)\cap\ker(d) = \{0\}$. In other words, it is demanded that
 \begin{subequations}
 \begin{align}\label{disjointness}
     d\partial x = 0 \quad \text{implies} \quad \partial x = 0\,,\\
     \partial d x = 0 \quad \text{implies} \quad d x = 0.
 \end{align}
 \end{subequations}
 Define the Laplacian $\Delta$ by (\ref{De}).
 Under these assumptions it can be shown that
 \begin{enumerate}
     \item $\ker (\Delta) = \ker (d) \cap \ker (\partial)$. The harmonic cocycles  annihilated by $\Delta$ are those and only those, that are $d$-closed and $\partial$-closed simultaneously;
     \item $V = \mathrm{im}(d) \oplus \mathrm{im}(\partial) \oplus \ker(\Delta)$. In other words, for any vector $x\in V$ there exists a unique Hodge  decomposition $x = d\alpha + \partial\beta + h$, where $\alpha$ and $\beta$ are some vectors in $V$, and $h$ is harmonic $\Delta h = 0$.
 \end{enumerate}
 Since by (\ref{disjointness}) $\partial\beta\neq0$ implies $d\partial\beta\neq0$, the kernel of $d$ consists of vectors of the type $d\alpha + h$, where $h$ is harmonic,
 \begin{equation}
     \ker (d) = \left\{x\in V\Big|\ x = d\alpha + h,\quad \Delta h =0\right\}.
 \end{equation}
 This implies that the harmonic cocycles and cohomology classes of $d$ are isomorphic
  as vector spaces, that is (\ref{kD}) is true.

  In the subsequent sections the operators $\sigma_-$ and $\sigma_+ := (\sigma_-)^*$ will play the roles of $d$ and $\partial$.
Moreover, in the following calculations one can spot which Young diagram or multispinor belongs to $\mathrm{im} (\sigma_-), \mathrm{im} (\sigma_+)$ or $\ker (\Delta)$ due to the equivariance of the constructed Laplace operators $\Delta$ with respect to the action of $\mathrm{GL}(d)$ or $\mathrm{O}(d)$ or $\mathrm{SL}(2;\mathbb{C})$, depending on the problem in question.

\section{$\sigma_-$ cohomology in Minkowski space of any dimension }
\label{vectorcase}
\subsection{Generating functions}
\label{gen}
The problem of finding the $\sigma_-$  cohomology in  tensor spaces of one or another type
 can be conveniently reformulated in terms of differential operators.
 To this end two-row Young diagrams in the symmetric basis can be described as a subset of polynomial ring $\mathbb{R}[Y,Z]$ generated by the set of $2d$ commuting variables $Y^a, Z^b$ (see \cite{Bekaert:2005vh} for detail). Consider the ring $\Lambda^p(\mathcal{M}^d) \otimes \mathbb{R}[Y,Z]$. Its homogeneous elements are differential $p$-forms valued in  $\mathbb{R}[Y,Z]$
\begin{equation}\label{e:formpoly}
\omega_{n,m}(x,dx,Y,Z) = \omega_{a(n),b(m)}(x,dx)Y^{a(n)}Z^{b(m)}\,.
\end{equation}
Consider the generating function
\begin{equation}
    \omega(x,dx\,|Y,Z) = \sum\limits_{n,m\geq 0}\omega_{n,m}(x,dx\,|Y,Z)=\sum\limits_{n,m\geq0} \omega_{a(n),b(m)}(x,dx)Y^{a(n)}Z^{b(m)}\,.
\end{equation}
Its expansion  in powers of $Y$ and $Z$ yields the tensor-valued forms $\omega_{a(n),b(m)}$ as the Taylor coefficients.
In this language the Young irreducibility condition reads as
\begin{equation}\label{e:youngcond}
Y^a \pdv{\omega}{Z^a} = 0\quad \Longleftrightarrow \quad\omega_{a(n),a b(m-1)} = 0\,.
\end{equation}
The tracelessness condition takes the form
\begin{equation}\label{e:youngtrace}
\eta^{ab}\partial_{Ya}\partial_{Yb} \omega = 0 \Longleftrightarrow \omega^k{}_{k a(n-2),b(m)} = 0\,.
\end{equation}
Note that all other traces are also zero  as a consequence of (\ref{e:youngcond}) and (\ref{e:youngtrace}),
\begin{equation}
\eta^{ab}\partial_{Ya}\partial_{Zb} \omega = 0 \Longleftrightarrow \omega^k{}_{a(n-1),b(m-1)k} = 0\,,
\end{equation}
\begin{equation}
\eta^{ab}\partial_{Za}\partial_{Zb} \omega = 0 \Longleftrightarrow \omega_{a(n),b(m-2)k}{}^k = 0\,.
\end{equation}

The generators of $\mathfrak{u}(d)$ and $\mathfrak{so}(d)$\footnote{ Following
 \cite{Vasiliev:2009ck}, in this section we do not
distinguish between different real forms of the same complex algebra freely going to their compact real  (Euclidean) form  since, not affecting the final results, this choice simplifies the analysis by allowing a positive-definite invariant scalar product on the space of tensors. Results of the Euclidean case coincide with those of the Lorentz one due to the equivalence of their representation theory on finite-dimensional modules. Indeed, suppose that some Lorentz-irreducible tensor $T^L$  represented $\sigma_-$ cohomology in the Lorentz case. Then analogous $\mathfrak{o}(d)$-irreducible tensor $T^E$  represents $\sigma_-$ cohomology in the compact case and vice versa. The only potential difference could be related to (anti)self-dual tensors that may exist in one signature but not in the other. However, these do not play a role in our analysis where (anti)self-dual tensors always appear in pairs or do not appear at all in sufficiently high dimensions $d>4$.
}
are now realized by the first-order differential operators
\begin{equation}\label{e:generators}
    \left(t_{\mathfrak{gl}(d)}{}\right)^a_b = Y^a \partial_{Yb} + Z^a \partial_{Zb} + \theta^a \partial_{\theta^b}, \quad \left(t^{\mathfrak{so}(d)}{}\right)_{ab} = \frac{1}{2}\bigg(\eta_{ac}t_{\mathfrak{gl}}{}^c_b - \eta_{bc}t_{\mathfrak{gl}}{}^c_a\bigg)\,,
\end{equation}
where $\theta^c$ is a Grassmann-odd element of the exterior algebra associated with the frame one-form $e^a$.

In these terms $\sigma_-$  acts as
\begin{equation}
\sigma_- \omega = \theta^a \pdv{\omega}{Z^a} = m \theta^c\omega_{a(n),c b(m-1)}(x,\theta)Y^{a(n)}Z^{b(m-1)}\,.
\end{equation}
It differs from the definition of Section 3 by an additional numerical factor introduced for future convenience. In the sequel  we sometimes do not write variables $Y,Z,\theta$ explicitly, that are
always assumed to be present implicitly. We adopt the convention that index $a$ is contracted with $Y$, $b$ with $Z$ and $c_i$ with $\theta^{c_i}$ with $\theta$s ordered as $c_1,...,c_p$.

The space $\Lambda(\mathcal{M}^d) \otimes \mathbb{R}[Y,Z]$  can be equipped with the scalar product
\begin{equation}\label{e:scalarprod}
\langle f,g\rangle = \frac{1}{\pi^{2n}}\int_{\mathbb{C}^d \times \mathbb{C}^d} d^{2d}Z\, d^{2d}Y\, d^d \theta\, d^d \bar{\theta}\, f(Z,Y,\theta)\,\overline{g(Z,Y,\theta)}\, e^{-|Z|^2-|Y|^2-\bar{\theta}\theta},
\end{equation}
where $f,g \in \Lambda(\mathcal{M}^d) \otimes \mathbb{R}[Y,Z]$ with complex $Y,Z,\theta$ and Berezin integral over anticommuting  variables.  (We work with the polynomials of complex variables with real coefficients).
The space $\Lambda(\mathcal{M}^d) \otimes \mathbb{R}[Y,Z]$ with the scalar product (\ref{e:scalarprod}) is a Hilbert space in the Euclidean metric signature  case  used in this section.
This scalar product yields the following  conjugation rules:
\begin{equation}\label{e:conjrule}
 (Z^a)^* = \partial_{Z}{}_{a}\,,\qquad (Y^a)^* = \partial_{Y}{}_a\,,\qquad
 (\theta^a)^* = \partial_{\theta}{}_a\,.
\end{equation}

\subsection{$\mathrm{GL}(d)$ example}
\label{sec:gl}
To illustrate the idea of our construction  let us first consider a simpler case where fields and gauge parameters take values in the irreps of $\mathfrak{gl}(d)$ described by two-row Young diagrams (no tracelessness conditions are imposed). Define the following operators,
that form $\mathfrak{gl}(2)$
\begin{equation}
      t_{1} = Y^{a}\frac{\partial }{\partial Z^{a}}, \quad t_{2} = Z^{a}\frac{\partial }{\partial Y^{a}}, \quad t_{0} = Y^{a}\frac{\partial }{\partial Y^{a}} - Z^{a}\frac{\partial }{\partial Z^{a}},
    \end{equation}
    \begin{equation}
      \comm{t_{1}}{t_{2}} = t_{0}, \quad \comm{t_{0}}{t_{1}} = 2 t_{1}, \quad \comm{t_{0}}{t_{2}} = -2 t_{2},
    \end{equation}
    \begin{equation}
      h_{1} = Y^{a}\frac{\partial }{\partial Y^{a}}, \quad h_2= Z^{a}\frac{\partial }{\partial Z^{a}}\,.
    \end{equation}
Namely, $t_i$ form $\mathfrak{sl}(2)$ while $h_1+h_2$ is central.

In terms of these
operators the space of $p$-forms valued in two-row Young diagrams is identified as $\ker (t_1)$
\begin{equation}
V^p = \{F \in \Lambda^p(\mathcal{M}^d) \otimes \mathbb{R}[Y,Z]| t_1 F = 0 \}\,.
\end{equation}
Here $\Lambda^p(\mathcal{M}^d)$ is generated by the Grassmann variables $\theta^a$.

Let us introduce auxiliary operators
\begin{equation}
\label{aux}
      Z_{\theta} = Z^{a}\frac{\partial }{\partial \theta^{a} }, \quad Y_{\theta} = Y^{a}\frac{\partial }{\partial \theta^{a} }, \quad D = \theta^a \frac{\partial }{\partial \theta^{a} }, \quad \theta_Y = \theta^a \frac{\partial }{\partial Y^{a} }, \quad \theta_Z = \theta^a \frac{\partial }{\partial Z^{a} }\,.
\end{equation}
Among the auxiliary operators $D$ plays the most important role as it gives differential form degree.
Now we should construct $\sigma_+: \sigma_+^2 = 0$ on $V^p$ and $\Im(\sigma_+) \subset V^p$. Consider the following operator:
\begin{equation}
      \sigma_+ = f(t_{0})Z_{\theta} + g(t_{0})Y_{\theta}t_2\,,
\end{equation}
    where $f(t_{0}) = \sum_{n=0}^{\infty}f_{n}t_{0}^{n}$ and $g(t_{0}) = \sum_{n=0}^{\infty}g_{n}t_{0}^{n}$.
Functions $f$ and $g$ have to be found from the conditions
\begin{equation}
     \sigma_+^2F=0\,,\qquad  t_{1}\sigma_+ F =0\,,  \qquad   \forall F \in V^p\,.
\end{equation}
After some re-ordering of operators this yields two equations
\begin{align}
& 0 = Y_{\theta} \bigg(f(t_{0} - 1)F + g(t_{0}-1)t_{0}F \bigg),\\
& 0 = Z_{\theta}Y_{\theta} \bigg(f(t_{0})g(t_{0}+1)t_{2}F - g(t_{0})f(t_{0}+1)t_{2}F  - g(t_{0})g(t_{0}+1)t_{2}F \bigg)\,
\end{align}
verified by
\begin{equation}
f(t_0) = -(t_0+1)g(t_0)
\end{equation}
giving
\begin{equation}
   \sigma_+ = -(t_0+1)g(t_0)Z_{\theta}+g(t_0)Y_{\theta}t_2\,.
\end{equation}
The free coefficient $g(t_0)$ is determined from the conjugacy requirement:
\begin{equation}
   (f,\sigma_+ g) = -(\theta_Z \bigg(g(t_0)(t_0+1)+g(t_0)\bigg)f, g) = (\sigma_- f,g)\,
\end{equation}
giving
\begin{equation}
g(t_0) = - \frac{1}{t_0+2}
\end{equation}
and hence
\begin{equation}\label{e:sigmaconjgl}
     \sigma_+:=(\sigma_-)^*
      = \frac{t_0+1}{t_0+2}Z_{\theta}-\frac{1}{t_0+2}Y_{\theta}t_2\,.
\end{equation}	

One can notice, that $\sigma_+$ in (\ref{e:sigmaconjgl}) differs from what one would expect from the conjugation rules (\ref{e:conjrule}). The reason is that in (\ref{e:conjrule}) we work with $C := \Lambda^p(\mathcal{M}^d) \otimes \mathbb{R}[Y,Z]$ complex. In the $\mathrm{GL}(d)$-case we deal with $\big(\Lambda^p(\mathcal{M}^d) \otimes \mathbb{R}[Y,Z]\big) \cap \ker (t_1)$ complex, therefore one should project on the highest weight vectors of the underlying $\mathfrak{sl}(2)$ in the complex C. The same procedure applies to the $\mathrm{O}(d)$-case. Though general formulae for extremal projectors are
known for any simple Lie algebra\footnote{We are grateful to the referee for bringing this fact to
our attention.} \cite{AST1,AST2,AST3, Zhelobenko}(for reviews see \cite{Tolstoy:2004tp,Tolstoy:2010kh}), to keep the paper self-contained we  derive the relevant projectors straightforwardly.

Knowing $\sigma_+$, it remains to construct the Laplace operator $\Delta = \{\sigma_-,\sigma_+\}$ and find its zeros. Elementary computation gives
\begin{equation}\label{e:laplacegl}
\Delta = \frac{t_0}{t_0+1}(D+h_2-1) + \frac{1}{t_0+1}Y_\theta \theta_Y  - \frac{1}{(t_0+1)(t_0+2)}t_2 \theta_Z Y_\theta\,.
\end{equation}

 Being built from the  manifestly $\mathfrak{gl}(d)$-invariant operators,  $\Delta$
 commutes with $\mathfrak{gl}(d)$ hence being diagonal on its irreducible submodules.
 Thus, it suffices to analyze zeros of  $\Delta$ on $\mathfrak{gl}(d)$ irreducible components of the  forms.

\subsubsection{$H^0(\sigma_-)$}

Any element of $V^0$ has the form $F = F_{a(n),b(m)}Y^{a(n)}Z^{b(m)}$. It is easy to see that
\begin{equation}
    \Delta F = h_2 F\,.
\end{equation}

Therefore
\begin{equation}\label{e:cohomgl0}
H^0(\sigma_-) = \{F = F_{a(n)}Y^{a(n)}| \forall F_{a(n)}\in \mathbb{R}\}\,.
\end{equation}

\subsubsection{$H^p(\sigma_-)$, $p>0$}

For $p>0$, a general element of $V^p$ is $F = F_{a(n),b(m)|c_1,..,c_p}Y^{a(n)}Z^{b(m)}\theta^{c_1}..\theta^{c_p}$. Generally it forms a reducible $\mathfrak{gl}(d)$-module associated with the tensor product of two diagrams. In terms of
Young diagrams it decomposes into the following  irreducible components:
\begin{equation}
\begin{picture}(15,10)(0,7)
{
\put(05,0){\line(1,0){10}}%
\put(05,20){\line(1,0){10}}%
\put(05,0){\line(0,1){20}}%
\put(15,0){\line(0,1){20}}%
\put(8,10){\scriptsize  ${p}$}%
}
\end{picture}   \,\,\,\otimes_{gl}\,\,\,
\begin{picture}(10,18)(0,7)
{
\put(05,00){\line(1,0){35}}%
\put(05,10){\line(1,0){35}}%
\put(05,0){\line(0,1){10}}%
\put(40,0){\line(0,1){10}}%
\put(17,3){\scriptsize  ${m}$}%
}
\end{picture}\begin{picture}(65,18)(10,7)
{\put(35,13){\scriptsize  ${n}$}
 \put(05,20){\line(1,0){60}}%
\put(05,10){\line(1,0){60}}%
\put(05,10){\line(0,1){10}}%
\put(65,10){\line(0,1){10}}%
}
\end{picture}
\cong
\begin{picture}(10,18)(0,7)
{
\put(05,00){\line(1,0){35}}%
\put(05,10){\line(1,0){35}}%
\put(05,0){\line(0,1){10}}%
\put(40,0){\line(0,1){10}}%
\put(17,3){\scriptsize  ${m}$}%
\put(05,-23){\line(1,0){10}}%
\put(05,-23){\line(0,1){23}}%
\put(15,-23){\line(0,1){23}}%
\put(7,-21){\begin{turn}{90}\scriptsize $p-1$\end{turn}}%
}
\end{picture}\begin{picture}(65,18)(10,7)
{\put(26,13){\scriptsize  ${n+1}$}
 \put(05,20){\line(1,0){60}}%
\put(05,10){\line(1,0){60}}%
\put(05,10){\line(0,1){10}}%
\put(65,10){\line(0,1){10}}%
}
\end{picture}
\oplus
\begin{picture}(10,18)(0,7)
{
\put(05,00){\line(1,0){35}}%
\put(05,10){\line(1,0){35}}%
\put(05,0){\line(0,1){10}}%
\put(40,0){\line(0,1){10}}%
\put(12,3){\scriptsize  ${m+1}$}%
\put(05,-23){\line(1,0){10}}%
\put(05,-23){\line(0,1){23}}%
\put(15,-23){\line(0,1){23}}%
\put(7,-21){\begin{turn}{90}\scriptsize $p-1$\end{turn}}%
}
\end{picture}\begin{picture}(65,18)(10,7)
{\put(35,13){\scriptsize  ${n}$}
 \put(05,20){\line(1,0){60}}%
\put(05,10){\line(1,0){60}}%
\put(05,10){\line(0,1){10}}%
\put(65,10){\line(0,1){10}}%
}
\end{picture}
\oplus
\begin{picture}(10,18)(0,7)
{
\put(05,00){\line(1,0){35}}%
\put(05,10){\line(1,0){35}}%
\put(05,0){\line(0,1){10}}%
\put(40,0){\line(0,1){10}}%
\put(17,3){\scriptsize  ${m}$}%
\put(05,-23){\line(1,0){10}}%
\put(05,-23){\line(0,1){23}}%
\put(15,-23){\line(0,1){23}}%
\put(7,-14){\begin{turn}{90}\scriptsize $p$\end{turn}}%
}
\end{picture}\begin{picture}(65,18)(10,7)
{\put(35,13){\scriptsize  ${n}$}
 \put(05,20){\line(1,0){60}}%
\put(05,10){\line(1,0){60}}%
\put(05,10){\line(0,1){10}}%
\put(65,10){\line(0,1){10}}%
}
\end{picture}
\oplus
\begin{picture}(10,18)(0,7)
{
\put(05,00){\line(1,0){35}}%
\put(05,10){\line(1,0){35}}%
\put(05,0){\line(0,1){10}}%
\put(40,0){\line(0,1){10}}%
\put(12,3){\scriptsize  ${m+1}$}%
\put(05,-23){\line(1,0){10}}%
\put(05,-23){\line(0,1){23}}%
\put(15,-23){\line(0,1){23}}%
\put(7,-21){\begin{turn}{90}\scriptsize $p-2$\end{turn}}%
}
\end{picture}\begin{picture}(65,18)(10,7)
{\put(26,13){\scriptsize  ${n+1}$}
 \put(05,20){\line(1,0){60}}%
\put(05,10){\line(1,0){60}}%
\put(05,10){\line(0,1){10}}%
\put(65,10){\line(0,1){10}}%
}
\end{picture}\,.\\
\end{equation}

\vspace{10mm}

At $p = 1$ the last diagram is absent. The manifest decomposition of $F_{a(n),b(m)|c_1,c_2,..,c_p}$ into
irreducible components is
\begin{multline}
F_{a(n),b(m)|c_1,c_2,..,c_p} = F_1{}_{a(n)c_1,b(m),c_2,..,c_p} + \frac{m}{n-m+2} F_1{}_{a(n)b,b(m-1)c_1,c_2,..,c_p} + F_2{}_{a(n),b(m)c_1,c_2,..,c_p} + \\ + F_3{}_{a(n),b(m),c_1,..,c_p} + F_4{}_{a(n)c_1,b(m)c_2,c_3,..,c_p},
\label{res}
\end{multline}
where $F_i$ corresponds to the $i$-th diagram. There are no restrictions on $n,m,p$ in (\ref{res}) except for $n \geq m$. If for some  $n,m$ tensor expression has a wrong Young shape, it is  zero. To simplify calculations we derive restrictions on $n,m,p$ for each diagram from the condition of being $\sigma_-$-closed. For the second and fourth diagrams we find no restrictions, but for others we have
\begin{equation}
 \sigma_- \bigg(\text{1st diagram}\bigg) = -m\bigg(1-\frac{1}{n-m+2}\bigg)F_1{}_{a(n)c_0,b(m-1)c_1,c_2,..,c_p} \Rightarrow m = 0,
\end{equation}
\begin{equation}
\label{F3}
\sigma_- F_3{}_{a(n),b(m),c_1,..,c_p} = m F_3{}_{a(n),b(m-1)c_0,c_1,..,c_p} \Rightarrow m = 0\,.
\end{equation}

Using this we obtain the action of $\Delta$ on the rest diagrams:
\begin{equation}
\Delta F_1{}_{a(n)c_1,c_2,..,c_p} = \frac{1}{n+1}\Delta \theta_Y F_1(Y,Z,\theta) = \frac{n(p-1)}{(n+1)^2}\theta_Y F_1(Y,Z,\theta),
\end{equation}
\begin{equation}
\Delta F_2{}_{a(n),b(m)c_1,c_2,..,c_p} = \frac{1}{m+1}\Delta \theta_Z F_2(Y,Z,\theta) = \frac{m+p}{m+1}\theta_Z F_2(Y,Z,\theta),
\end{equation}
\begin{multline}
\Delta F_4{}_{a(n)c_1,b(m)c_2,c_3,..,c_p} = \frac{1}{(m+1)(n+1)}\Delta \theta_Y\theta_Z F_4(Y,Z,\theta)\\ = \frac{(n-m)(p+m-1)}{(n-m+1)(m+1)(n+1)} \theta_Y\theta_Z F_4(Y,Z,\theta)\,.
\end{multline}
As a result,
\begin{equation}\label{e:cohomgl1}
H^1(\sigma_-) = \{\phi = F_{a(n)c}\theta^c Y^{a(n)}| h_2 F = 0, F\in V^1\}\,,
\end{equation}
\begin{equation}\label{e:cohomgl2}
H^p(\sigma_-) = \{W = \theta_Y\theta_Z C(Y,Z,\theta)| t_0 C = 0, C \in V^{p-2} \,,\quad p>1\},
\end{equation}
\begin{equation*}
C =C{}_{a(n),b(n),c_1,..,c_{p-2}}Y^{a(n)}Z^{b(n)}\theta^{c_1}..\theta^{c_{p-2}} \in V^{p-2}\,.
\end{equation*}

The dynamical interpretation of the obtained results is as follows. The system has one symmetric gauge field with gauge transformation described by a symmetric parameter. The second cohomology group $H^2(\sigma_-)$ is spanned by a single tensor corresponding to the generalized (traceful) Weyl tensor. If the latter
 is set to zero,
the system becomes topological with the zero-curvature field equations. Otherwise
the unfolded equations encode a set of constraints expressing all fields and Weyl tensor via
derivatives of the physical fields. Proceeding further with the equations on the Weyl tensor and its descendants results in an infinite set of constraints with no differential equations on the physical field. Such off-shell unfolded equations were considered in
\cite{Vasiliev:2005zu}. The off-shell systems are for interest in many contexts
such as, e.g., construction of actions and quantization \cite{Misuna:2020fck, Misuna:2019ijn}.
The lower cohomology groups (\ref{e:cohomgl0}), (\ref{e:cohomgl1}) and (\ref{e:cohomgl2}) match with those obtained, e.g., in \cite{Bekaert:2005vh}.

\subsection{$\mathrm{O}(d)$ case}
\label{sec:o(d)}

\subsubsection{Irreducibility conditions}
The $\mathrm{O}(d)$ case  is in many respects analogous to that  of $\mathrm{GL}(d)$. The difference is  due to the  tracelessness condition \eqref{e:youngtrace}. The algebra of the operators encoding irreducibility
conditions is extended since the  metric allows the new types of contractions between $\theta, Y, Z$ and their derivatives. From the  representation theory perspective new terms associated with traces  appear in the diagram decomposition of the form coefficients, affecting the cohomology analysis.

The following operators form the algebra $\mathfrak{sp}(4)$:
\begin{align}
      &t_{1} = Y^{a}\frac{\partial }{\partial Z^{a}}, \quad t_{2} = Z^{a}\frac{\partial }{\partial Y^{a}}, \quad h_{1} = Y^{a}\frac{\partial }{\partial Y^{a}}, \quad h_{2}= Z^{a}\frac{\partial }{\partial Z^{a}},\\
      &t_{0} = Y^{a}\frac{\partial }{\partial Y^{a}} - Z^{a}\frac{\partial }{\partial Z^{a}},\\
      &f_{1} = \partial_{Y}^a\partial_{Y a}, \quad f_{2} = \partial_{Z}^a\partial_{Z a}, \quad f_{3} = \partial_{Y}^a\partial_{Z a}, \\
      &e_{1} = Y^a Y_a, \quad e_{2} = Z^a Z_a, \quad e_{3} = Y^a Z_a\,.
\end{align}	
Evidently, these operators commute with  the $\mathfrak{so}(d)$ generators (\ref{e:generators}).
 $\mathfrak{sp}(4)$ and $\mathfrak{so}(d)$ form a  Howe-dual pair~\cite{Howe:1989}. Young condition \eqref{e:youngcond} and tracelessness condition \eqref{e:youngtrace} impose  highest weight conditions on a $\mathfrak{sp}(4)$-module.

In addition, we introduce the following  $\mathrm{O}(d)$ invariant operators:
\begin{align}
      &Z_{\theta} = Z^a \partial_{\theta a}, \quad Y_{\theta} = Y^a \partial_{\theta a},\\
      &\partial_{\theta Z} = \partial_{\theta a}\partial_{Z}^a, \quad \partial_{\theta Y} = \partial_{\theta a}\partial_{Y}^a,\\
      &\theta_Z = \theta^a \partial_{Z a}, \quad \theta_Y = \theta^a \partial_{Y a}\,,
\end{align}
which, along with $D$ (\ref{aux})  counting differential form degree, extend
$\mathfrak{sp}(4)$ to
$\mathfrak{osp}(2|4)$. The simplest way to see this is to let index $a$  take a single value, treating the operators $Z,Y,\partial_Z, \partial_Y, \theta, \partial_\theta$ as creation and annihilation operators, and apply the oscillator realization of $\mathfrak{osp}(2|4)$.

In the problem in question, the form space is
\begin{equation}\label{e:repspace_o}
V^p = \{F \in \Lambda^p(\mathcal{M}^d) \otimes \mathbb{R}[Y,Z]| t_1 F = 0, f_1 F = 0 \}\,.
\end{equation}
Note that these restrictions imply the tracelessness over indices $(Y,Z)$ and $(Z,Z)$ as a consequence of the form of commutators of $t_1$ with $f_{1,2}$.

\subsubsection{$\sigma_+$}

Let us look for $\sigma_+ = \sigma_-^* $ in the form
\begin{multline}
      \sigma_+ = g_1(h_1,h_2)Z_{\theta} + g_2(h_1,h_2)t_2Y_{\theta} + g_3(h_1,h_2)e_3 \partial_{\theta Y} + g_4(h_1,h_2)e_1t_2 \partial_{\theta Y} + g_5(h_1,h_2)e_2 \partial_{\theta Z}+\\
       +g_6(h_1,h_2)e_3 t_2\partial_{\theta Z} + g_7(h_1,h_2)e_1 t_2^2 \partial_{\theta Z}\,.
\end{multline}

The condition $\Im(\sigma_+) \subset V^p$ gives
\begin{multline}
	0 = f_1 \sigma_+ F = \bigg(2 g_2(h_1+2,h_2)+2g_3(h_1+2,h_2)+2g_4(h_1+2,h_2)(d + 2h_1)\bigg)t_2 \partial_{\theta Y} F + \\
	+ \bigg(2g_6(h_1+2,h_2)+2g_7(h_1+2,h_2)(d+2h_1)\bigg)t_2^2\partial_{\theta Z}F,
\end{multline}
\begin{multline}
	0 = t_1 \sigma_+ F = \bigg(g_1(h_1-1,h_2+1)+g_2(h_1-1,h_2+1)t_0\bigg)Y_\theta F + \bigg(-g_3(-1,1)+2g_5(-1,1)+ \\ + g_6(-1,1)t_0 \bigg)e_3 \partial_{\theta Z} + \bigg(g_3(-1,1)+g_4(-1,1)(t_0-2) \bigg)e_1 \partial_{\theta Y} +  \bigg(-g_4(-1,1)+g_6(-1,1)+ \\ + 2g_7(-1,1)(t_0-1)\bigg)e_1 t_2 \partial_{\theta Z}\,.
\end{multline}

This imposes the following  six equations on seven coefficients
\begin{align}
g_1(h_1,h_2)+g_2(h_1,h_2)(t_0+2) = 0, \\
-g_3(h_1,h_2)+2g_5(h_1,h_2)+g_6(h_1,h_2)(t_0+2) = 0, \\
g_3(h_1,h_2)+g_4(h_1,h_2) t_0 = 0, \\
-g_4(h_1,h_2)+g_6(h_1,h_2)+2g_7(h_1,h_2)(t_0+1) = 0, \\
g_2(h_1,h_2)+g_3(h_1,h_2)+g_4(h_1,h_2)(d+2h_1-4) = 0, \\
g_6(h_1,h_2)+g_7(h_1,h_2)(d+2h_1-4) = 0\,.
\end{align}
Choosing
$g_7(h_1,h_2)$ as a free parameter, we obtain
\begin{align}
g_1(h_1,h_2) = -(t_0+2)(d-4+h_1+h_2)(d-6+2h_2)g_7(h_1,h_2),
      \\
g_2(h_1,h_2) = (d-4+h_1+h_2)(d-6+2h_2)g_7(h_1,h_2),
      \\
g_3(h_1,h_2) = t_0(d-6+2h_2)g_7(h_1,h_2),
      \\
g_4(h_1,h_2) = -(d-6+2h_2)g_7(h_1,h_2),
      \\
g_5(h_1,h_2) = (t_0+1)(d-4+h_1+h_2)g_7(h_1,h_2),
      \\
g_6(h_1,h_2) = -(d-4+2h_1)g_7(h_1,h_2)\,.
\end{align}

Now using the conjugation rules (\ref{e:conjrule}) and highest weight conditions (\ref{e:repspace_o}) we get
\begin{equation}
	(F_1,\sigma_+ F_2) = \Big( -g_7(h_1,h_2)(t_0+2)(d-4+h_1+h_2)(d-6+2h_2) \Big) (\sigma_{-}F_1,F_2)\,.
\end{equation}
The condition  $\sigma_+ = \sigma^*_-$ demands
\begin{equation}
	g_7(h_1,h_2) = - \frac{1}{(t_0+2)(d-4+h_1+h_2)(d-6+2h_2)}
\end{equation}
giving
\begin{multline}
      \sigma_+ = Z_{\theta} - \frac{1}{t_0+2} t_2 Y_{\theta} - \frac{t_0}{(t_0+2)(d-4+h_1+h_2)} e_3 \partial_{\theta Y} + \frac{1}{(t_0+2)(d-4+h_1+h_2)}e_1t_2 \partial_{\theta Y} -  \\
 -\frac{t_0+1}{(t_0+2)(d-6+2h_2)} e_2 \partial_{\theta Z} +\frac{d-4+2h_1}{(t_0+2)(d-4+h_1+h_2)(d-6+2h_2)} e_3 t_2\partial_{\theta Z} - \\
      -\frac{1}{(t_0+2)(d-4+h_1+h_2)(d-6+2h_2)} e_1 t_2^2 \partial_{\theta Z}\,.
\end{multline}
This yields  operator $\sigma_+$ such that $\sigma_+^2 = 0$ on $V$, $\Im(\sigma_+) \subset V$ and $\sigma_-^* = \sigma_+$.

To calculate the
 Laplace operator $\Delta = \sigma_- \sigma_+ + \sigma_+ \sigma_-$ on $V^p$ we obtain
straightforwardly that
\begin{multline}
\sigma_- \sigma_+ = \theta_Z Z_\theta - \frac{1}{t_0+1}t_2 \theta_Z Y_\theta - \frac{1}{t_0+1}\theta_Y Y_\theta - \frac{t_0-1}{(t_0+1)(d-3+h_1+h_2)}\bigg(e_3 \theta_Z \partial_{\theta Y} + \theta^a Y_a \partial_{\theta Y}\bigg) \\
+ \frac{1}{(t_0+1)(d-3+h_1+h_2)}\bigg(e_1 t_2 \theta_Z \partial_{\theta Y} + e_1 \theta_Y \partial_{\theta Y} \bigg) - \frac{t_0}{(t_0+1)(d-4+2h_2)}\bigg(e_2 \theta_Z \partial_{\theta Z} + 2 \theta^a Z_a \partial_{\theta Z} \bigg)\\
+ \frac{d-4+2h_1}{(t_0+1)(d-3+h_1+h_2)(d-4+2h_2)}\bigg(e_3 t_2 \theta_Z \partial_{\theta Z} + e_3 \theta_Y \partial_{\theta Z} + t_2 \theta^a Y_a \partial_{\theta Z} - \theta^a Z_a \partial_{\theta Z}\bigg)\\
- \frac{1}{(t_0+1)(d-3+h_1+h_2)(d-4+2h_2)}\bigg(e_1 t_2^2\theta_Z \partial_{\theta Z} + 2e_1 t_2 \theta_Y \partial_{\theta Z} \bigg)\,.
\end{multline}

\begin{multline}
\sigma_+ \sigma_-=  Z_\theta \theta_Z + \frac{1}{t_0+2}t_2 \theta_Z Y_\theta + \frac{t_0}{(t_0+2)(d-4+h_1+h_2)}e_3 \theta_Z \partial_{\theta Y}
- \frac{1}{(t_0+2)(d-4+h_1+h_2)}\times \\ \times e_1 t_2 \theta_Z \partial_{\theta Y}
+ \frac{t_0+1}{(t_0+2)(d-6+2h_2)}e_2 \theta_Z \partial_{\theta Z}
- \frac{d-4+2h_1}{(t_0+2)(d-4+h_1+h_2)(d-6+2h_2)}e_3 t_2 \theta_Z \partial_{\theta Z}\\
+ \frac{1}{(t_0+2)(d-4+h_1+h_2)(d-6+2h_2)}e_1 t_2^2\theta_Z \partial_{\theta Z}\,.
\end{multline}

Since, by construction, both $\sigma_- \sigma_+$ and $\sigma_+ \sigma_-$ and hence
$\Delta$ are $\mathrm{O}(d)$ invariant, $\Delta$ is diagonal on
irreducible $\mathrm{O}(d)$-modules and, to compute $H(\sigma_-)$,  it suffices to find its zeros on the
irreducible components.

\subsubsection{$H^0(\sigma_-)$}
\label{H0}
In the sector of zero-forms, all  terms that contain $\frac{\partial}{\partial \theta^c}$ trivialize. Hence,
\begin{equation}
\Delta F = h_2 F = m F \,
\end{equation}
and
\begin{equation}
H^0(\sigma_-) = \{F = F_{a(n)}Y^{a(n)}| \forall F_{a(n)}\in \mathbb{R}\}\,.
\end{equation}

Comparing the resulting differential gauge parameters with \eqref{e:fronsdalgauge}, we find that, as anticipated,  differential gauge symmetries in the unfolded formulation coincide with those of the Fronsdal theory.

\subsubsection{$H^p(\sigma_-)$, $p>0$}
\label{Hp}

The main difference between $\mathfrak{o}(d)$- and $\mathfrak{gl}(d)$- cases is due to the traceful terms in the  decomposition of
the $p$-forms into the irreducible parts depicted as
\begin{equation}\label{dia}
\begin{picture}(15,10)(0,7)
{
\put(05,0){\line(1,0){10}}%
\put(05,20){\line(1,0){10}}%
\put(05,0){\line(0,1){20}}%
\put(15,0){\line(0,1){20}}%
\put(8,10){\scriptsize  ${p}$}%
}
\end{picture}   \,\,\,\otimes_{so}\,\,\,
\begin{picture}(10,18)(0,7)
{
\put(05,00){\line(1,0){35}}%
\put(05,10){\line(1,0){35}}%
\put(05,0){\line(0,1){10}}%
\put(40,0){\line(0,1){10}}%
\put(17,3){\scriptsize  ${m}$}%
}
\end{picture}\begin{picture}(65,18)(10,7)
{\put(35,13){\scriptsize  ${n}$}
 \put(05,20){\line(1,0){60}}%
\put(05,10){\line(1,0){60}}%
\put(05,10){\line(0,1){10}}%
\put(65,10){\line(0,1){10}}%
}
\end{picture}
\cong
\begin{picture}(10,18)(0,7)
{
\put(05,00){\line(1,0){35}}%
\put(05,10){\line(1,0){35}}%
\put(05,0){\line(0,1){10}}%
\put(40,0){\line(0,1){10}}%
\put(17,3){\scriptsize  ${m}$}%
\put(05,-25){\line(1,0){10}}%
\put(05,-25){\line(0,1){25}}%
\put(15,-25){\line(0,1){25}}%
\put(7,-15){\begin{turn}{90}\scriptsize $p$\end{turn}}%
}
\end{picture}\begin{picture}(65,18)(10,7)
{\put(35,13){\scriptsize  ${n}$}
 \put(05,20){\line(1,0){60}}%
\put(05,10){\line(1,0){60}}%
\put(05,10){\line(0,1){10}}%
\put(65,10){\line(0,1){10}}%
}
\end{picture}
\oplus
\begin{picture}(10,18)(0,7)
{
\put(05,00){\line(1,0){35}}%
\put(05,10){\line(1,0){35}}%
\put(05,0){\line(0,1){10}}%
\put(40,0){\line(0,1){10}}%
\put(12,3){\scriptsize  ${m+1}$}%
\put(05,-25){\line(1,0){10}}%
\put(05,-25){\line(0,1){25}}%
\put(15,-25){\line(0,1){25}}%
\put(7,-21){\begin{turn}{90}\scriptsize $p-1$\end{turn}}%
}
\end{picture}\begin{picture}(65,18)(10,7)
{\put(35,13){\scriptsize  ${n}$}
 \put(05,20){\line(1,0){60}}%
\put(05,10){\line(1,0){60}}%
\put(05,10){\line(0,1){10}}%
\put(65,10){\line(0,1){10}}%
}
\end{picture}
\oplus
\begin{picture}(10,18)(0,7)
{
\put(05,00){\line(1,0){35}}%
\put(05,10){\line(1,0){35}}%
\put(05,0){\line(0,1){10}}%
\put(40,0){\line(0,1){10}}%
\put(17,3){\scriptsize  ${m}$}%
\put(05,-25){\line(1,0){10}}%
\put(05,-25){\line(0,1){25}}%
\put(15,-25){\line(0,1){25}}%
\put(7,-21){\begin{turn}{90}\scriptsize $p-1$\end{turn}}%

}
\end{picture}\begin{picture}(65,18)(10,7)
{\put(26,13){\scriptsize  ${n+1}$}
 \put(05,20){\line(1,0){60}}%
\put(05,10){\line(1,0){60}}%
\put(05,10){\line(0,1){10}}%
\put(65,10){\line(0,1){10}}%
}
\end{picture}
\oplus
\end{equation}
\\
\begin{equation*}
\oplus
\begin{picture}(10,18)(0,7)
{
\put(05,00){\line(1,0){35}}%
\put(05,10){\line(1,0){35}}%
\put(05,0){\line(0,1){10}}%
\put(40,0){\line(0,1){10}}%
\put(12,3){\scriptsize  ${m+1}$}%
\put(05,-25){\line(1,0){10}}%
\put(05,-25){\line(0,1){25}}%
\put(15,-25){\line(0,1){25}}%
\put(7,-21){\begin{turn}{90}\scriptsize $p-2$\end{turn}}%
}
\end{picture}\begin{picture}(65,18)(10,7)
{\put(26,13){\scriptsize  ${n+1}$}
 \put(05,20){\line(1,0){60}}%
\put(05,10){\line(1,0){60}}%
\put(05,10){\line(0,1){10}}%
\put(65,10){\line(0,1){10}}%
}
\end{picture}
\oplus
\begin{picture}(10,18)(0,7)
{
\put(05,00){\line(1,0){35}}%
\put(05,10){\line(1,0){35}}%
\put(05,0){\line(0,1){10}}%
\put(40,0){\line(0,1){10}}%
\put(17,3){\scriptsize  ${m}$}%
\put(05,-25){\line(1,0){10}}%
\put(05,-25){\line(0,1){25}}%
\put(15,-25){\line(0,1){25}}%
\put(7,-21){\begin{turn}{90}\scriptsize $p-1$\end{turn}}%
}
\end{picture}\begin{picture}(65,18)(10,7)
{\put(26,13){\scriptsize  ${n-1}$}
 \put(05,20){\line(1,0){60}}%
\put(05,10){\line(1,0){60}}%
\put(05,10){\line(0,1){10}}%
\put(65,10){\line(0,1){10}}%
}
\end{picture}
\oplus
\begin{picture}(10,18)(0,7)
{
\put(05,00){\line(1,0){35}}%
\put(05,10){\line(1,0){35}}%
\put(05,0){\line(0,1){10}}%
\put(40,0){\line(0,1){10}}%
\put(12,3){\scriptsize  ${m-1}$}%
\put(05,-25){\line(1,0){10}}%
\put(05,-25){\line(0,1){25}}%
\put(15,-25){\line(0,1){25}}%
\put(7,-21){\begin{turn}{90}\scriptsize $p-1$\end{turn}}%
}
\end{picture}\begin{picture}(65,18)(10,7)
{\put(35,13){\scriptsize  ${n}$}
 \put(05,20){\line(1,0){60}}%
\put(05,10){\line(1,0){60}}%
\put(05,10){\line(0,1){10}}%
\put(65,10){\line(0,1){10}}%
}
\end{picture}
\oplus
\begin{picture}(10,18)(0,7)
{
\put(05,00){\line(1,0){35}}%
\put(05,10){\line(1,0){35}}%
\put(05,0){\line(0,1){10}}%
\put(40,0){\line(0,1){10}}%
\put(12,3){\scriptsize  ${m-1}$}%
\put(05,-25){\line(1,0){10}}%

\put(05,-25){\line(0,1){25}}%
\put(15,-25){\line(0,1){25}}%
\put(7,-21){\begin{turn}{90}\scriptsize $p-2$\end{turn}}%
}
\end{picture}\begin{picture}(65,18)(10,7)
{\put(25,13){\scriptsize  ${n+1}$}
 \put(05,20){\line(1,0){60}}%
\put(05,10){\line(1,0){60}}%
\put(05,10){\line(0,1){10}}%
\put(65,10){\line(0,1){10}}%
}
\end{picture}
\oplus
\begin{picture}(10,18)(0,7)
{
\put(05,00){\line(1,0){35}}%
\put(05,10){\line(1,0){35}}%
\put(05,0){\line(0,1){10}}%
\put(40,0){\line(0,1){10}}%
\put(14,3){\scriptsize  ${m+1}$}%
\put(05,-25){\line(1,0){10}}%
\put(05,-25){\line(0,1){25}}%
\put(15,-25){\line(0,1){25}}%
\put(7,-21){\begin{turn}{90}\scriptsize $p-2$\end{turn}}%
}
\end{picture}\begin{picture}(65,18)(10,7)
{\put(26,13){\scriptsize  ${n-1}$}
 \put(05,20){\line(1,0){60}}%
\put(05,10){\line(1,0){60}}%
\put(05,10){\line(0,1){10}}%
\put(65,10){\line(0,1){10}}%
}
\end{picture}
\oplus
\end{equation*}\\
\begin{equation*}
\oplus
\begin{picture}(10,18)(0,7)
{
\put(05,00){\line(1,0){35}}%
\put(05,10){\line(1,0){35}}%
\put(05,0){\line(0,1){10}}%
\put(40,0){\line(0,1){10}}%
\put(12,3){\scriptsize  ${m-1}$}%
\put(05,-25){\line(1,0){10}}%
\put(05,-25){\line(0,1){25}}%
\put(15,-25){\line(0,1){25}}%
\put(7,-21){\begin{turn}{90}\scriptsize $p-2$\end{turn}}%
}
\end{picture}\begin{picture}(65,18)(10,7)
{\put(26,13){\scriptsize  ${n-1}$}
 \put(05,20){\line(1,0){60}}%
\put(05,10){\line(1,0){60}}%
\put(05,10){\line(0,1){10}}%
\put(65,10){\line(0,1){10}}%
}
\end{picture}
\oplus
\begin{picture}(10,18)(0,7)
{
\put(05,00){\line(1,0){35}}%
\put(05,10){\line(1,0){35}}%
\put(05,0){\line(0,1){10}}%
\put(40,0){\line(0,1){10}}%
\put(17,3){\scriptsize  ${m}$}%
\put(05,-25){\line(1,0){10}}%
\put(05,-25){\line(0,1){25}}%
\put(15,-25){\line(0,1){25}}%
\put(7,-21){\begin{turn}{90}\scriptsize $p-2$\end{turn}}%
}
\end{picture}\begin{picture}(65,18)(10,7)
{\put(35,13){\scriptsize  ${n}$}
 \put(05,20){\line(1,0){60}}%
\put(05,10){\line(1,0){60}}%
\put(05,10){\line(0,1){10}}%
\put(65,10){\line(0,1){10}}%
}
\end{picture}
\oplus
\begin{picture}(10,18)(0,7)
{
\put(05,00){\line(1,0){35}}%
\put(05,10){\line(1,0){35}}%
\put(05,0){\line(0,1){10}}%
\put(40,0){\line(0,1){10}}%
\put(17,3){\scriptsize  ${m}$}%
\put(05,-25){\line(1,0){10}}%
\put(05,-25){\line(0,1){25}}%
\put(15,-25){\line(0,1){25}}%
\put(7,-21){\begin{turn}{90}\scriptsize $p-2$\end{turn}}%
}
\end{picture}\begin{picture}(65,18)(10,7)
{\put(35,13){\scriptsize  ${n}$}
 \put(05,20){\line(1,0){60}}%
\put(05,10){\line(1,0){60}}%
\put(05,10){\line(0,1){10}}%
\put(65,10){\line(0,1){10}}%
}
\end{picture}
\oplus
\begin{picture}(10,18)(0,7)
{
\put(05,00){\line(1,0){35}}%
\put(05,10){\line(1,0){35}}%
\put(05,0){\line(0,1){10}}%
\put(40,0){\line(0,1){10}}%
\put(17,3){\scriptsize  ${m}$}%
\put(05,-25){\line(1,0){10}}%
\put(05,-25){\line(0,1){25}}%
\put(15,-25){\line(0,1){25}}%
\put(7,-21){\begin{turn}{90}\scriptsize $p-3$\end{turn}}%
}
\end{picture}\begin{picture}(65,18)(10,7)
{\put(26,13){\scriptsize  ${n+1}$}
 \put(05,20){\line(1,0){60}}%
\put(05,10){\line(1,0){60}}%
\put(05,10){\line(0,1){10}}%
\put(65,10){\line(0,1){10}}%
}
\end{picture}
\oplus
\end{equation*}\\
\begin{equation*}
\oplus
\begin{picture}(10,18)(0,7)
{
\put(05,00){\line(1,0){35}}%
\put(05,10){\line(1,0){35}}%
\put(05,0){\line(0,1){10}}%
\put(40,0){\line(0,1){10}}%
\put(12,3){\scriptsize  ${m+1}$}%
\put(05,-25){\line(1,0){10}}%
\put(05,-25){\line(0,1){25}}%
\put(15,-25){\line(0,1){25}}%
\put(7,-21){\begin{turn}{90}\scriptsize $p-3$\end{turn}}%
}
\end{picture}\begin{picture}(65,18)(10,7)
{\put(35,13){\scriptsize  ${n}$}
 \put(05,20){\line(1,0){60}}%
\put(05,10){\line(1,0){60}}%
\put(05,10){\line(0,1){10}}%
\put(65,10){\line(0,1){10}}%
}
\end{picture}
\oplus
\begin{picture}(10,18)(0,7)
{
\put(05,00){\line(1,0){35}}%
\put(05,10){\line(1,0){35}}%
\put(05,0){\line(0,1){10}}%
\put(40,0){\line(0,1){10}}%
\put(17,3){\scriptsize  ${m}$}%
\put(05,-25){\line(1,0){10}}%
\put(05,-25){\line(0,1){25}}%
\put(15,-25){\line(0,1){25}}%
\put(7,-21){\begin{turn}{90}\scriptsize $p-3$\end{turn}}%
}
\end{picture}\begin{picture}(65,18)(10,7)
{\put(26,13){\scriptsize  ${n-1}$}
 \put(05,20){\line(1,0){60}}%
\put(05,10){\line(1,0){60}}%
\put(05,10){\line(0,1){10}}%
\put(65,10){\line(0,1){10}}%
}
\end{picture}
\oplus
\begin{picture}(10,18)(0,7)
{
\put(05,00){\line(1,0){35}}%
\put(05,10){\line(1,0){35}}%
\put(05,0){\line(0,1){10}}%
\put(40,0){\line(0,1){10}}%
\put(12,3){\scriptsize  ${m-1}$}%
\put(05,-25){\line(1,0){10}}%
\put(05,-25){\line(0,1){25}}%
\put(15,-25){\line(0,1){25}}%
\put(7,-21){\begin{turn}{90}\scriptsize $p-3$\end{turn}}%
}
\end{picture}\begin{picture}(65,18)(10,7)
{\put(35,13){\scriptsize  ${n}$}
 \put(05,20){\line(1,0){60}}%
\put(05,10){\line(1,0){60}}%
\put(05,10){\line(0,1){10}}%
\put(65,10){\line(0,1){10}}%
}
\end{picture}
\oplus
\begin{picture}(10,18)(0,7)
{
\put(05,00){\line(1,0){35}}%
\put(05,10){\line(1,0){35}}%
\put(05,0){\line(0,1){10}}%
\put(40,0){\line(0,1){10}}%
\put(17,3){\scriptsize  ${m}$}%
\put(05,-25){\line(1,0){10}}%
\put(05,-25){\line(0,1){25}}%
\put(15,-25){\line(0,1){25}}%
\put(7,-21){\begin{turn}{90}\scriptsize $p-4$\end{turn}}%
}
\end{picture}\begin{picture}(65,18)(10,7)
{\put(35,13){\scriptsize  ${n}$}
 \put(05,20){\line(1,0){60}}%
\put(05,10){\line(1,0){60}}%
\put(05,10){\line(0,1){10}}%
\put(65,10){\line(0,1){10}}%
}
\end{picture}
\end{equation*}\\

For $1 \leq p \leq 3$, the diagrams carrying negative $p$-dependent labels are absent. It is important to note that the diagram $(n,m,p-2)$ is present twice: one copy results from the contraction of one of the form
indices with the first row followed by the symmetrization of another form index with the same row.
Another one results from the application of the same procedure to the second row.
This fact leads to two different tensor implementations. Also note that some of the
diagrams vanish for special dimensions by virtue of the Two-Column Theorem:

\begin{theorem}
 $\mathfrak{so}(d)$ traceless tensors with the symmetry properties of such
Young diagrams that the sum of the heights of the first two columns exceeds $d$, are identically zero.~\cite{Hamermesh}
\end{theorem}

The cohomology $H^p(\sigma_-)$ is empty for $p>d$. Analogously,
  some potential elements of $H^p(\sigma_-)$ are  zero by  the Two-Column Theorem for large $p \leq d$.\\

Now we are in a position to consider  the action of the Laplace operator on each of the  diagrams  (\ref{dia}) separately. In the following restrictions on $n,m,p$ will be imposed: if for some  $n,m,p$ a tensor  has a wrong Young shape, it is  zero. In most cases we will simplify calculation by demanding $p$-forms be $\sigma_-$-closed. Another simplification is due to the fact that $\Delta F = 0$ is equivalent to the two equations $\sigma_- F = 0$ and $\sigma_+ F = 0$.

\textbf{Diagram (n,m;p), $n\geq0, m\geq0, p\geq0$}  has the tensor form $T_{a(n),b(m),c_1,..,c_p}$. It is $\sigma_-$-closed, if $m = 0$. Then
\begin{equation}\label{e:nmp}
    \Delta T_{a(n),c_1,..,c_p} = p T_{a(n),c_1,..,c_p}\,.
\end{equation}
This diagram is in $\ker (\Delta)$ at $p = 0$ that reproduces the already obtained result for $H^0(\sigma_-)$.

\textbf{Diagram (n+1,m+1;p-2), $n\geq0, m\geq0, p\geq2$} has the tensor form $T_{a(n)c_1,b(m)c_2,..,c_p}$. This diagram is $\sigma_-$-closed. The action of $\Delta$ is
\begin{equation}\label{e:n1m1p2}
    \Delta T_{a(n)c_1,b(m)c_2,..,c_p} = \frac{(n-m)(p+m-1)}{(n-m+1)}T_{a(n)c_1,b(m)c_2,..,c_p}\,.
\end{equation}

This diagram is in $\ker (\Delta)$, if $n=m$, thus belonging to $H^p(\sigma_-)$ with $p\geq 2$.

\textbf{Diagram (n,m+1;p-1), $n\geq1, m\geq0, p\geq1$} has the tensor form $T_{a(n),b(m)c_0,c_1,..,c_{p-1}}$.  Then
\begin{equation}
    \Delta T_{a(n),b(m)c_0,c_1,..,c_{p-1}} = (m+p) T_{a(n),b(m)c_0,c_1,..,c_{p-1}}\,.
\end{equation}
This equation admits no solutions since $p\geq 1$  in the case in question.

\textbf{Diagram (n+1,m;p-1), $n\geq0, m\geq0, p\geq1$} has the tensor form
$$T_{a(n)c_0,b(m),c_1,..,c_{p-1}} + \frac{m}{n-m+2}T_{a(n)b,b(m-1)c_0,c_1,..,c_{p-1}}\,.$$
 It is $\sigma_-$-closed iff $m = 0$. Then
\begin{equation}\label{e:n1mp1}
    \Delta T_{a(n)c_0,c_1,..,c_{p-1}} = \frac{n(p-1)}{n+1} T_{a(n)c_0,c_1,..,c_{p-1}}\,.
\end{equation}
This expression vanishes at $p = 1$. Hence,
\begin{equation}
    T_{a(n)c_0} \in H^1(\sigma_-)\,.
\end{equation}

As one can see, $n=0$ in (\ref{e:n1mp1}) also leads to zero Laplace action. However, this is not a new result, since it has been already  accounted in the diagrams $(n,m;p)$ for $n=m=p=0$ (\ref{e:nmp}), $(n+1,m+1;p-2)$ for $n=m=0,p\geq2$ (\ref{e:n1m1p2}) and the $n=m=0,p=1$ case of $(n+1,m;p-1)$ (\ref{e:n1mp1}). This fact is a simple consequence of the tensor multiplication of a column by a scalar.

\textbf{Diagram (n-1,m;p-1), $n\geq1, m\geq0, p\geq1$} has the tensor form
\begin{multline}
T = \eta_{ac_1}\rho_{a(n-1),b(m),c_2,..,c_p} - \frac{n-1}{d-4+2n}\eta_{aa}\rho_{a(n-2)c_1,b(m),c_2,..,c_p} + \frac{m(n-1)}{(d-4+m+n)(d-4+2n)}\times \\
\times \eta_{aa}\rho_{a(n-2)b,b(m-1)c_1,c_2,..,c_p} - \frac{m}{d-4+m+n}\eta_{ab}\rho_{a(n-1),b(m-1)c_1,c_2,..,c_p}\,,
\end{multline}
where $\rho$ is an arbitrary $(n-1,m;p-1)$ tensor.

One can check that $T \in \ker (\sigma_-)$ demands $m = 0$ with
\begin{equation}\label{e:divterm}
T = \eta_{ac_1}\rho_{a(n-1),c_2,..,c_p} - \frac{n-1}{d-4+2n}\eta_{aa}\rho_{a(n-2)c_1,c_2,..,c_p}.
\end{equation}

After some calculation we obtain
\begin{equation}
    \Delta T = \frac{(p-1)(d-2+n)}{d-3+n} T\,.
\end{equation}
This implies that $\Delta$ has zero at $p = 1$ contributing to $H^1(\sigma_-)$. The case of $d=2,n=1$ must be considered separately because of the divergent denominator. The seeming divergence emerges due to the second term in (\ref{e:divterm}), which is absent at $d=2,n=1$,
\begin{equation}
    T = \eta_{ac_1}\rho_{c_2,..,c_p} \Rightarrow \sigma_+ T =\frac{p-1}{2}\Big(\eta_{bc_1}\rho_{a,c_2,..,c_{p-1}} - \eta_{ac_1}\rho_{b,c_2,..,c_{p-1}}\Big)\,,
\end{equation}
leading to the same answer with $p=1$.

\textbf{Diagram (n,m-1;p-1) $n\geq1, m\geq1, p\geq1$} has the tensor form
\begin{multline}\label{e:survterms}
T = (n-1)\eta_{aa} \rho_{a(n-2)bc_1,b(m-1),c_2,..,c_p} - \frac{(n-1)(m-1)}{d-6+2m} \eta_{aa} \rho_{a(n-2)bb,b(m-2)c_1,c_2,..,c_p} - (d-4+m+n) \times \\
\times \eta_{ac_1} \rho_{a(n-1)b,b(m-1),c_2,..,c_p} - (n-m)\eta_{ab} \rho_{a(n-1)c_1,b(m-1),c_2,..,c_p} + \frac{(m-1)(d-4+2n)}{d-6+2m}\eta_{ab} \times \\
\times \rho_{a(n-1)b,b(m-2)c_1,c_2,..,c_p} + \frac{(n-m+1)(d-4+m+n)}{n} \eta_{bc_1} \rho_{a(n),b(m-1),c_2,..,c_p} - \\
-\frac{(m-1)(n-m+1)(d-4+m+n)}{(d-6+2m)n}\eta_{bb} \rho_{a(n),b(m-2)c_1,c_2,..,c_p}\,,
\end{multline}
where $\rho$ is an arbitrary $(n,m-1;p-1)$ tensor.

Let us show that this diagram can never be annihilated by $\sigma_-$. Indeed,
\begin{equation}\label{e:n,m-1}
\sigma_- T = \bigg(d-4+m+n-(n-m)\bigg)\eta_{ac_0} \rho_{a(n-1)c_1,b(m-1),c_2,..,c_p} + (m-1)(\emph{lit}),
\end{equation}
where (\emph{lit}) denotes some terms that are linearly independent from the first one. If $m=1$ the above expression reduces to
\begin{equation}
\sigma_- T = \big(d-2\big)\eta_{ac_0} \rho_{a(n-1)c_1,c_2,..,c_p}.
\end{equation}

The expression in brackets vanishes at $d = 2$. However, such diagram is zero
by virtue of the Two-column theorem 5.1, since the heights of the first two columns sum up to  $3>d$.

Thus, the nontrivial $T$ is never in  $\ker (\sigma_-)$.

\textbf{Diagram (n+1,m-1;p-2) $n\geq1, m\geq1, p\geq2$} has the tensor form
\begin{multline}\label{T}
T = \eta_{ac_1}\rho_{a(n-1)bc_2,b(m-1),c_3,..c_p}- \frac{n-m+1}{n}\eta_{bc_1}\rho_{a(n)c_2,b(m-1),c_3,..c_p}  + \\ + \frac{m-1}{n-m+3}\eta_{ac_1}\rho_{a(n-1)bb,b(m-2)c_2,c_3,..c_p}
-\frac{(m-1)(n-1)}{(d-6+2m)(n-m+3)}\eta_{aa}\rho_{a(n-2)bbc_1,b(m-2)c_2,c_3,..c_p} + \\ + \frac{2(m-1)(n-m+1)}{(d-6+2m)(n-m+3)}\eta_{ab}\rho_{a(n-1)bc_1,b(m-2)c_2,c_3,..c_p}
- \frac{(m-1)(n-m+1)}{n(n-m+3)}\times \\ \times\eta_{bc_1}\rho_{a(n)b,b(m-2)c_2,c_3,..c_p} - \frac{(m-1)(n-m+2)(n-m+1)}{(d-6+2m)(n-m+3)n} \eta_{bb}\rho_{a(n)c_1,b(m-2)c_2,c_3,..c_p}\,,
\end{multline}
where $\rho$ is an arbitrary $(n+1,m-1;p-2)$ tensor.

Though the tensor realization (\ref{T}) may look complicated, the problem is simplified by the
observation  that all
terms except for the first and second ones carry a factor of $(m-1)$. The action of $\sigma_-$ on the first and second terms produces a factor of $(m-1)$ in front of each $\eta \rho$ combination. It can be checked that $\sigma_- T$ has an overall factor of $(m-1)$ so that the only possible solution for $T \in \ker (\sigma_-)$ is at $m=1$
in which case
the tensor decomposition acquires the form
\begin{equation}
T = \eta_{ac_1}\rho_{a(n-1)bc_2,..,c_p} - \eta_{bc_1}\rho_{a(n)c_2,..,c_p}.
\end{equation}

At $p = 2$, after some calculations one can check that
\begin{equation}
    \sigma_+ T = 0\,.
\end{equation}

For $p>2$ it is not difficult to see that $\sigma_+ T \neq 0$. Indeed,
\begin{equation}
\sigma_+ T = (p-2)\bigg(1 + \frac{1}{n}\bigg)\eta_{ac_1}\rho_{a(n-1)bc_2,b,c_3,..,c_{p-1}} + (\emph{lit})\,,
\end{equation}
where (\emph{lit}) denotes other linearly independent terms. The first term is never zero.

Thus,
\begin{equation}
    T = \eta_{ac_1}\rho_{a(n-1)bc_2} - \eta_{bc_1}\rho_{a(n)c_2} \in H^2(\sigma_-)\,.
\end{equation}

\textbf{Diagram (n-1,m+1;p-2), $n\geq1, m\geq0, p\geq2$ } has the tensor form
\begin{equation}
T = \eta_{ac_1} \rho_{a(n-1),b(m)c_2,..,c_p} - \frac{n-1}{d-4+2n}\eta_{aa} \rho_{a(n-2)c_1,b(m)c_2,..,c_p}\,,
\end{equation}
where $\rho$ is an arbitrary $(n-1,m+1;p-2)$ tensor.

Though it is obviously in $\ker( \sigma_-)$ for any $m$, it is not hard to see that it is never in $\ker (\sigma_+)$.
\begin{equation}
\sigma_+ T = -\bigg(1 + \frac{p-2}{m+1} \bigg) \eta_{ac_1} \rho_{a(n-1),b(m+1),..,c_{p-1}} + (\emph{lit} )\,.
\end{equation}
Hence this diagram does not contribute to $H^p(\sigma_-)$.

\textbf{Diagram (n-1,m-1;p-2), $n\geq1, m\geq1, p\geq2$} has an involved tensor form. Since the  coefficients in the expression below are complicated, we extract the factor of $(m-1)$ once present denoting the leftover coefficients as $\alpha_i$,
\begin{spreadlines}{0.75em}
\begin{multline}
T = \eta_{ac_1}\eta_{bc_2}\rho_{a(n-1),b(m-1),..,c_p} + (m-1)\alpha_1 \eta_{aa}\eta_{aa}\rho_{a(n-4)bbc_1,b(m-2)c_2,..,c_p} + \\ + \alpha_2 \eta_{aa}\eta_{ac_1}\rho_{a(n-3)bc_2,b(m-1),..,c_p}
+ (m-1)\alpha_3\eta_{aa}\eta_{ac_1}\rho_{a(n-3)bb,b(m-2)c_2,..,c_p} + \\ + (m-1)\alpha_4 \eta_{aa}\eta_{ab}\rho_{a(n-3)bc_1,b(m-2)c_2,..,c_p}
+\alpha_5 \eta_{aa}\eta_{bc_1}\rho_{a(n-2)c_2,b(m-1),..,c_p} + \\ + (m-1)\alpha_6\eta_{aa}\eta_{bc_1}\rho_{a(n-2)b,b(m-2)c_2,..,c_p}
+ (m-1)\alpha_7 \eta_{aa}\eta_{bb}\rho_{a(n-2)c_1,b(m-2)c_2,..,c_p} + \\ + \alpha_8 \eta_{ac_1}\eta_{ab}\rho_{a(n-2)c_2,b(m-1),..,c_p}
+ (m-1)\alpha_9 \eta_{ac_1}\eta_{ab}\rho_{a(n-2)b,b(m-2)c_2,..,c_p} + \\ + (m-1)\alpha_{10} \eta_{ac_1}\eta_{bb}\rho_{a(n-1),b(m-2)c_2,..,c_p}
+ (m-1) \alpha_{11} \eta_{ab}\eta_{ab}\rho_{a(n-2)c_1,b(m-2)c_2,..,c_p} + \\ + (m-1) \alpha_{12} \eta_{ab}\eta_{bc_1}\rho_{a(n-1),b(m-2)c_2,..,c_p}\,,
\end{multline}
\end{spreadlines}
where $\rho$ is an arbitrary $(n-1,m-1;p-2)$ tensor.
The explicit form of $\alpha_i$ is given in the Appendix A.

Now we observe that the action of $\sigma_-$ on the terms free of the factor of $(m-1)$
produces such factor. Hence,  $\sigma_-(T)$ has the form of the sum of linearly independent
terms with the common factor of  $(m-1)$. Consequently,
\begin{equation}
\sigma_- T = 0, \text{ iff\quad } m = 1.
\end{equation}

At $m=1$, the only terms that remain are
\begin{multline}
 T = \eta_{ac_1}\eta_{bc_2}\rho_{a(n-1),..,c_p} + \frac{(n-2)(n-1)}{(d-3+n)(d-4+2n)} \eta_{aa}\eta_{ac_1}\rho_{a(n-3)bc_2,..,c_p} + \\ + \frac{(n-1)(d-2+n)}{(d-3+n)(d-4+2n)} \eta_{aa}\eta_{bc_1}\rho_{a(n-2)c_2,..,c_p} - \frac{n-1}{d-3+n} \eta_{ac_1}\eta_{ab}\rho_{a(n-2)c_2,..,c_p}.
\end{multline}

It can be checked that for $p = 2$
\begin{equation}
    \sigma_+ T = 0\,.
\end{equation}

For $p > 2$ it is not difficult to see that
\begin{equation}
    \sigma_+ T = \frac{(p-2)d}{(d-2)(d-3+n)(d-2+n)} e_3^2 \theta_Y \rho + (\emph{lit} ),
\end{equation}
where
$\rho = \rho_{a(n-1),c_1,..,c_{p-2}}Y^{a(n-1)}\theta^{c_1}..\theta^{c_{p-2}}$. Therefore,
\begin{multline}
   H^2(\sigma_-) \ni T = \eta_{ac_1}\eta_{bc_2}\rho{}_{a(n-1)} + \frac{(n-2)(n-1)}{(d-3+n)(d-4+2n)} \eta_{aa}\eta_{ac_1}\rho{}_{a(n-3)bc_2} + \\
     + \frac{(n-1)(d-2+n)}{(d-3+n)(d-4+2n)}\eta_{aa}\eta_{bc_1}\rho{}_{a(n-2)c_2} -\frac{n-1}{d-3+n} \eta_{ac_1}\eta_{ab}\rho{}_{a(n-2)c_2} \,.
\end{multline}

\textbf{Diagram (n+1,m;p-3), $n\geq1, m\geq1, p\geq3$} has the tensor form
\begin{equation}
    T = \eta_{bc_1}\rho_{a(n)c_2,b(m-1)c_3,..,c_p} - \frac{n}{n-m+1}\eta_{ac_1}\rho_{a(n-1)bc_2,b(m-1)c_3,..,c_p}\,,
\end{equation}
where $\rho$ is an arbitrary $(n+1,m;p-3)$ tensor.

Explicit computation gives
\begin{multline}
    \Delta T = \frac{t_0}{t_0+1}\bigg(p + h_2 - 1\bigg)T - \frac{2 t_0}{(t_0+1)(d-4+2h_2)}\theta^a Z_a \partial_{\theta Z} T + \\ + \frac{d-4+2h_1}{(t_0+1)(d-3+h_1+h_2)(d-4+2h_2)} \bigg(t_2 \theta^a Y_a \partial_{\theta Z} - \theta^a Z_a  \partial_{\theta Z}\bigg)T - \\ - \frac{t_0-1}{(t_0+1)(d-3+h_1+h_2)}\theta^a Y_a \partial_{\theta Y}T\,.
\end{multline}

This expression vanishes at $n = m$. Indeed, in this
 case $t_0 T = 0, h_1 T= h_2 T = n T, t_2 T = 0$ so that
\begin{equation}
    \Delta T = \frac{1}{d-3+2n} \bigg(t_2 \theta^a Y_a \partial_{\theta Z} - \theta^a Z_a  \partial_{\theta Z}\bigg)T + \frac{1}{d-3+2n}\theta^a Y_a \partial_{\theta Y}T\,.
\end{equation}

Since
\begin{equation}
    t_2 \theta^a Y_a \partial_{\theta Z}T = \theta^a Y_a \partial_{\theta Z}t_2 T - \theta^a Y_a \partial_{\theta Y}T + \theta^a Z_a  \partial_{\theta Z}T = - \theta^a Y_a \partial_{\theta Y}T + \theta^a Z_a  \partial_{\theta Z}T\,,
\end{equation}
it follows that $\Delta T=0$ at $n=m$. To check that $\Delta T\neq 0$
at $n\neq m$  one should substitute the expression for $T$ noticing that
different linearly independent terms have no common factor to vanish, that
implies  nontriviality of $\Delta T$.

\textbf{Diagram (n,m+1;p-3), $n\geq1, m\geq0, p\geq3$} has the tensor form  $T = \eta_{ac_1}\rho_{a(n-1)c_2,b(m)c_3,..,c_p}$. It belongs to $\ker (\sigma_-)$, but not to $\ker (\sigma_+)$.

\begin{equation}
   \sigma_+ T = \bigg(1 + \frac{p-3}{m+1}\bigg)\eta_{ac_1}\rho_{a(n-1)c_2,b(m+1),..,c_p} + (\emph{lit})\,.
\end{equation}

\textbf{Diagram (n,m;p-2) $n\geq1, m\geq0, p\geq2$} admits two tensor realizations $T_i$ due to the double presence of this diagram in the result of tensor product. The tensors
\begin{equation}
T_1 = \eta_{ac_1} \rho_1{}_{a(n-1)c_2,b(m),c_3,..,c_p} + \frac{m}{n}\eta_{bc_1} \rho_1{}_{a(n),b(m-1)c_2,c_3,..,c_p},
\end{equation}
\begin{multline}
T_2 = \eta_{aa} \rho_2{}_{a(n-2)bc_1,b(m-1)c_2,c_3,..,c_p} - \frac{n-m}{n-1}\eta_{ab} \rho_2{}_{a(n-1)c_1,b(m-1)c_2,c_3,..,c_p} - \frac{d-4+n+m}{n-1}\eta_{ac_1} \times\\
\times \rho_2{}_{a(n-1)b,b(m-1)c_2,c_3,..,c_p} + \frac{(n-m+1)(d-4+m+n)}{n(n-1)}\eta_{bc_1} \rho_2{}_{a(n),b(m-1)c_2,c_3,..,c_p}
\end{multline}
are linearly independent. That
 Laplace operator acts diagonally on $T_i$, $\Delta T_i = \lambda_i(n,m,p) T_i$, allows us to
 separately consider each  of these diagrams.
Firstly, we check if these  are in $\ker (\sigma_-)$  computing
\begin{align}
&\sigma_- T_1 = m \eta_{ac_1} \rho_1{}_{a(n-1)c_2,b(m-1)c_3,..,c_{p+1}},\\
&\sigma_- T_2 = \frac{d-4+2m}{n-1} \eta_{ac_1} \rho_2{}_{a(n-1)c_2,b(m-1)c_3,..,c_{p+1}}.
\end{align}
$T_1 \in \ker (\sigma_-)$ at $m = 0$. Formally, $T_2$ is annihilated by $\sigma_-$ at $d=2, m = 1$, but this is not allowed by the Two-column  theorem. So, the only candidate for cohomology is $T_1$. However,
\begin{equation}
  \sigma_+ T_1 =  \frac{1}{n+1}\Big(1 + \frac{p-2}{n} \Big)\eta_{bc_1}\rho_1{}_{a(n),c_2,..,c_{p-1}} + (\emph{lit})\,,
\end{equation}
which is never zero.

\textbf{Diagram (n-1,m;p-3), $n\geq2, m\geq1, p\geq3$} has the tensor form
\begin{multline}
T = \eta_{ac_1}\eta_{bc_2} \rho_{a(n-1),b(m-1)c_3,..,c_p} - \frac{(n-1)(d-3+m+n)}{(d-4+n+m)(d-4+2n)}\eta_{aa}\eta_{bc_1} \rho_{a(n-2)c_2,b(m-1)c_3,..,c_p} -\\
- \frac{(n-1)(n-2)}{(d-4+m+n)(d-4+2n)} \eta_{aa}\eta_{ac_1} \rho_{a(n-3)bc_2,b(m-1)c_3,..,c_p} + \frac{n-1}{d-4+m+n} \times \\
\times \eta_{ab}\eta_{ac_1} \rho_{a(n-2)c_2,b(m-1)c_3,..,c_p}\,,
\end{multline}
where $\rho$ is an arbitrary $(n-1,m;p-3)$ tensor.

Obviously, $T \in \ker(\sigma_-)$. However, $T \notin \ker(\sigma_+)$.
\begin{equation}
\sigma_+ T = \bigg(1 + \frac{p-3}{m}\bigg) \eta_{ac_1}\eta_{bc_2} \rho_{a(n-1),b(m),c_3,..,c_{p-1}} + (\emph{lit})\,,
\end{equation}
hence not contributing to cohomology.

\textbf{Diagram (n,m-1;p-3), $n\geq1, m\geq1, p\geq3$} has the tensor form
\begin{equation}
T = \eta_{ac_1}\eta_{bc_2}\rho_{a(n-1)c_3,b(m-1),..,c_p} + (m-1)(\emph{lit})\,
\end{equation}
with all terms except for the first one carrying a  factor of $(m-1)$. The action of $\sigma_-$ on the first term brings a  factor of $(m-1)$ in front of $\eta\rho$. Since all $\eta\rho$ terms in the decomposition are linearly independent we conclude that
\begin{equation}
\sigma_- T = 0, \text{ if } m = 1.
\end{equation}

At $p = 3$ one can check that $\sigma_+ T = 0$. However, for $p>3$ $T$ does not belong to $\ker(\sigma_+)$,
\begin{equation}
\sigma_+ T = - (p-3) \eta_{ac_1}\eta_{bc_2} \rho_{a(n-1)c_3,b,..,c_{p-1}} + (\emph{lit})\,.
\end{equation}
Consequently, the only contribution to $H^3(\sigma_-)$ is
\begin{equation}
    T = \eta_{ac_1}\eta_{bc_2}\rho_{a(n-1)c_3} \in H^3(\sigma_-)\,.
\end{equation}

\textbf{Diagram (n,m;p-4), $n\geq1, m\geq1, p\geq4$} has the tensor form
\begin{equation}
    T = \eta_{ac_1}\eta_{bc_2}\rho_{a(n-1)c_3,b(m-1)c_4,..,c_p}\,.
\end{equation}

This is obviously annihilated by $\sigma_-$, but not by $\sigma_+$.
\begin{equation}
\sigma_+ T = - \bigg(1 + \frac{p-4}{m} \bigg) \eta_{ac_1}\eta_{bc_2}\rho_{a(n-1)c_3,b(m),..,c_{p-1}} + (\emph{lit})\,.
\end{equation}
Hence it does not contribute to $H^p(\sigma_-)$.

\subsubsection{Summary}

Summarizing the results of Sections \ref{H0} and \ref{Hp} we found the following cohomology
groups:

\begin{equation}\label{e:cohology_0}
H^0(\sigma_-) = \{F = F_{a(n)}Y^{a(n)}| F \in V^0\}\,,
\end{equation}

\begin{multline}\label{e:cohology_1}
H^1(\sigma_-) = \{\phi = F_1{}_{a(n)c}Y^{a(n)}\theta^c, F_1 \in V^1;\\ \phi^{tr} = \Big[(n-1)\eta_{aa}F_2{}_{a(n-2)c} - (d-4+2n)\eta_{ac}F_2{}_{a(n-1)}\Big]Y^{a(n)}\theta^c \in V^1 \}\,,
\end{multline}

\begin{multline}\label{e:cohology_2}
H^2(\sigma_-) = \{W = \theta_Y\theta_Z C(Y,Z): t_0 C = 0, C \in V^{0}; \\ \mathcal{E}_A = \Big[\eta_{ac_1}\rho_1{}_{a(n-1)bc_2} - \eta_{bc_1}\rho_1{}_{a(n)c_2}\Big]Y^{a(n)}Z^b\theta^{c_1}\theta^{c_2} \in V^2; \\ \mathcal{E}_B = \Big[\eta_{ac_1}\eta_{bc_2}\rho_2{}_{a(n-1)} + \frac{(n-2)(n-1)}{(d-3+n)(d-4+2n)} \eta_{aa}\eta_{ac_1}\rho_2{}_{a(n-3)bc_2} +\\
+ \frac{(n-1)(d-2+n)}{(d-3+n)(d-4+2n)}\eta_{aa}\eta_{bc_1}\rho_2{}_{a(n-2)c_2} - \frac{n-1}{d-3+n} \eta_{ac_1}\eta_{ab}\rho_2{}_{a(n-2)c_2}\Big]Y^{a(n)}Z^b\theta^{c_1}\theta^{c_2} \in V^2 \}\,,
\end{multline}
\begin{multline}\label{e:cohology_3}
H^3(\sigma_-) = \{B^{fr} = \eta_{ac_1}\eta_{bc_2}\rho_{a(n-1)c_3} Y^{a(n)}Z^{b}\theta^{c_1}\theta^{c_2}\theta^{c_3} \in V^3; \\
B_1 = \theta_Y\theta_Z C(Y,Z,\theta): t_0 C = 0, C = C_{a(n),b(n),c}Y^{a(n)}Z^{b(n)}\theta^{c} \in V^1; \\
B_2 = \Big[\eta_{bc_1}\rho_{a(n)c_2,b(n-1)c_3} - n\eta_{ac_1}\rho_{a(n-1)bc_2,b(n-1)c_3}\Big]Y^{a(n)}Z^{b(n)}\theta^{c_1}\theta^{c_2}\theta^{c_3} \in V^3\}\,.
\end{multline}

At $p>3$
\begin{multline}\label{e:cohology_p}
H^p(\sigma_-) = \Big\{B_1 = \theta_Y\theta_Z C(Y,Z,\theta): t_0 C = 0, C = C_{a(n),b(n),c_1,..,c_{p-2}}Y^{a(n)}Z^{b(n)}\theta^{c_1}..\theta^{c_{p-2}} \in V^{p-2}; \\
B_2 = \Big[\eta_{bc_1}\rho_{a(n)c_2,b(n-1)c_3,..,c_p} - n\eta_{ac_1}\rho_{a(n-1)bc_2,b(n-1)c_3,..,c_p}\Big]Y^{a(n)}Z^{b(n)}\theta^{c_1}..\theta^{c_{p}}  \in V^{p} \Big\}.
\end{multline}

According to Theorem \ref{Theorem_cohomology}, the differential gauge transformation parameters are described by $H^0(\sigma_-)$ (\ref{e:cohology_0}). The gauge parameter  in the Fronsdal theory is known to be a symmetric traceless tensor. Since tensors $(n,0,0)$ constitute the cohomology group $H^0(\sigma_-)$, the unfolded differential gauge transformation  is shown to coincide with the Fronsdal one.

As recalled in  Section \ref{sec:fronsdal_form}, the Fronsdal field consists of two symmetric traceless fields \eqref{e:fronsdalfieldcomp}. These fields are represented by the cohomology groups $H^1(\sigma_-)$ (\ref{e:cohology_1}). Cohomology group $H^1(\sigma_-)$ consists of two elements $(n+1,0,0)$ and $(n-1,0,0)$ matching the components of the Fronsdal field. Thus, the physical fields in the unfolded approach indeed coincide with the Fronsdal field.

The cohomology group $H^2(\sigma_-)$ (\ref{e:cohology_2}) describes
gauge invariant combinations of derivatives of the physical fields that can be used to impose
differential equations on the latter.
The  Fronsdal cohomology classes $\mathcal{E}_A$ and $\mathcal{E}_B$ match with the Fronsdal equations:
$\mathcal{E}_A$ is associated with the traceless part of the Fronsdal equations, while  $\mathcal{E}_B$ with the  trace part, that is the equations $\mathcal{E}_A = 0, \mathcal{E}_B = 0$ just reproduce the Fronsdal equations. Note that the number of resulting equations is the same as the number of fields, as it should be in a Lagrangian system.

$W$  in (\ref{e:cohology_2}) represents "Weyl" {} cohomology. Imposing $W = 0$ in the case of gravity one gets conformally flat metrics and in the case of higher spins "conformally flat"{} fields. In
Einstein gravity and HS theory, the equation $W=0$ is not imposed. Instead, elements of $W$ are interpreted
as new fields $C$ that describe
generalized Weyl tensors by virtue of the unfolded equations (\ref{unfoldedeq_2}).
Thus, calculation of the cohomology group $H^2(\sigma_ -)$ shows that  unfolded equations (\ref{unfoldedeq_1}), (\ref{unfoldedeq_2}) contain  Fronsdal equations along with
constraints on auxiliary fields.

In accordance with the general discussion of Section \ref{gs}
 elements of  $H^3(\sigma_-)$ (\ref{e:cohology_3}) correspond to  Bianchi identities. Class $B^{fr}$ describes the Bianchi identities for the Fronsdal equations. Note that their number coincides with the number of differential gauge parameters. The remaining classes $B_1$ and $B_2$ correspond to the Bianchi identities on the Weyl tensor. It is noteworthy that the latter can be checked to  coincide with the first $\widetilde{\sigma}_-$ cohomology in the Weyl sector of zero-forms of \cite{Vasiliev:2003ev} for $s>1$. This fact exhibits the connection between the gauge and Weyl sectors.

 For $p>3$ cohomology groups $H^p(\sigma_-)$ describe the higher Bianchi identities
for Bianchi identities on the Weyl tensor also interpreted as syzygies \cite{Vasiliev:2009ck}.

Obtained lower cohomology groups match with the results of \cite{Bekaert:2005vh, Skvortsov:2009nv, Barnich:2004cr}.

In HS theory the fields are realized by one-forms. Formally, one can consider
field equations (\ref{unfoldedeq_1}), (\ref{unfoldedeq_2}) for $p$-forms
$\omega_{a(n),b(m)}$ valued
in a two-row irreducible $\mathfrak{o}(d)$-module. From our results and physical
interpretation of the $\sigma_-$ cohomology groups it follows that for $p>1$ the unfolded system in the gauge sector is off-shell. To answer the question whether the full unfolded system including both the gauge $p$-form sector and the Weyl $(p-1)$-form sector is off-shell, the analysis of $H(\tilde \sigma_{-})$ has to be performed in the Weyl sector. The case of $p>1$ may be somewhat similar to the $s=1$ case, where the equation on $A_\mu$ lies in the Weyl sector.

Finally, let us stress that the
results of this section  for HS fields in Minkowski space admit a straightforward deformation
  to $AdS_d$ with the same operator $\sigma_-$. This is
  because in that case dynamical fields
  are described by rectangular diagrams of the $AdS_d$ algebra $\mathfrak{o}(d-1,2)$
  \cite{Vasiliev:2001wa}.
  In general, in the flat limit, irreducible massless (gauge) fields in $AdS_d$  decompose into nontrivial sets of irreducible flat space massless fields \cite{Brink:2000ag,Metsaev:1995re,Metsaev:1997nj} and there is no one-to-one correspondence between massless fields in Minkowski space and $AdS_d$. Namely, a generic
  irreducible field in Minkowski space may admit no deformation to $AdS_d$ (see also \cite{Boulanger:2008up}).
\section{$\sigma_-$ cohomology in  $AdS_4$ in the spinor language}
$4d$ HS theories admit a description in terms of
 two-component spinors instead of  tensors. That is, instead of using the generating functions in the tensor form $\omega(x,dx\,|Y,Z)$, where $Y^a$ and $Z^a$ carried vector Lorentz indices $ a=\{0,1,2,3\}$, we will use
\begin{equation}
    \omega(y,\bar y\,|x,dx) = \sum\limits_{k,m}\omega^{\alpha_1\dots\alpha_k,\dot\alpha_1\dots\dot\alpha_m}(x,dx)\, y_{\alpha_1}\dots y_{\alpha_k}\,\bar y_{\dot\alpha_1}\dots\bar y_{\dot\alpha_m}\,,
\end{equation}
where the indices $\alpha$ and $\dot\alpha$  of the  two-component
commuting spinors $y^\alpha$ and $\bar y^{\dot\alpha}$ take two values  $\{1,2\}$.

Analogously to Sections 3 and 4, we have to introduce the grading on the space of $\Lambda^\bullet(M)\otimes \mathbb{C}[[y,\bar y]]$. Consider the homogeneous element of $\Lambda^\bullet(M)\otimes \mathbb{C}[[y,\bar y]]$ of degree $N$ and $\bar N$ in $y$ and $\bar y$, respectively
\begin{equation}
    \omega(\mu y, \bar{\mu}\bar y\,|x,dx) = \mu^N\,\bar{\mu}{}^{\bar N}\,\omega(y,\bar y\,|x,dx).
\end{equation}
Define the grading operator $G$ on $\Lambda^\bullet(M)\otimes \mathbb{C}[[y,\bar y]]$ as follows:
\begin{equation}
\label{Ggo}
    G\omega(y,\bar y\,|x,dx) = |N-\bar N|\,\omega(y,\bar y\,|x,dx) \equiv |\deg_y(\omega(y,\bar y\,|x,dx))-\deg_{\bar y}(\omega(y,\bar y\,|x,dx))|.
\end{equation}
Note that in the bosonic sector the frame-like fields $e^{m_1\dots m_{s-1}} \ \leftrightarrow \ e^{\alpha_1\dots \alpha_{s-1},\dot\alpha_1\dots\dot\alpha_{s-1}}$ have the lowest possible grading $G=0$. For our later computations to match with the Fronsdal theory, we define the action of $\sigma_-$ on $\omega(y,\bar y)$ to  decrease the $G$-grading:
\begin{subequations}
    \begin{align}
        \sigma_-\omega(y,\bar y) &:= i\, \bar y^{\dot\alpha}h^\alpha_{\ \dot\alpha}\partial_\alpha\ \omega(y,\bar y), \quad\quad \text{at } \deg_y(\omega)>\deg_{\bar y}(\omega),\\
        \sigma_-\omega(y,\bar y) &:=  i\, y^{\alpha}h_\alpha^{\ \dot\alpha}\bar\partial_{\dot\alpha}\ \omega(y,\bar y), \quad\quad \text{at } \deg_y(\omega)<\deg_{\bar y}(\omega),\\
        \sigma_-\omega(y,\bar y) &:= 0
           \qquad \qquad \qquad\qquad \quad\,\,\text{at } \deg_y(\omega)=\deg_{\bar y}(\omega)\,,
    \end{align}
\end{subequations}
where
\begin{align}
    \dd_\alpha:=\frac{\dd}{\dd y^\alpha}\,,\qquad
     \bar\dd_{\dot\alpha}:=\frac{\dd}{\dd \bar y^{\dot \alpha}}\,,
\end{align}
and the dependence on $x$ and $dx$ in $\omega(y,\bar y) = \omega(y,\bar y\,|x,dx)$ is always implicit.

It is easy to check that so defined $\sigma_-$ is nilpotent, $(\sigma_-)^2=0$.

Note that $\sigma_\pm$ change the grading $G$ by 2. This agrees in particular with the fact
that the bosonic and fermionic sectors, where the grading is even and odd respectively, are independent. We consider in detail the more complicated bosonic case,
observing in the end that the computation  for fermionic fields is quite similar.

Note that the analysis of $\sigma_-$ cohomology in the $4d$ conformal HS theory
was also performed in  terms of two-component spinors in \cite{Shaynkman:2014fqa, Shaynkman:2018kkc}. It is more complicated since the generating functions
of conformal HS theory depend
on twice as many independent spinors, but simpler since it is free of the module factors like
$|N- N|$ in the grading  definition  (\ref{Ggo}).

Next, we  define a scalar product\footnote{Being $\mathrm{SL}(2,\mathbb{C})$-invariant this scalar product is not positive-definite. Analogously to the tensorial case,
 without affecting the $\sigma_-$ cohomology analysis it can be made
positive-definite by going to the $\mathfrak{su}(2)\oplus \mathfrak{su}(2)$ algebra, which is the compact real form of $\mathfrak{sl}(2,\mathbb{C})\oplus \mathfrak{sl}(2,\mathbb{C})$ with altered conjugation rules  $\bar y^\alpha =y_\alpha$, $\bar y^{\dot\alpha} =y_{\dot\alpha}$.} on generating elements of $\Lambda^\bullet(M)\otimes \mathbb{C}[[y,\bar y]]$ by
\begin{align}
    \langle q^\alpha|q^\beta\rangle=\langle y^\alpha|y^\beta\rangle=i\epsilon^{\alpha\beta}\,, &&
    \langle\bar{p}_{\dot\alpha}|\bar{p}_{\dot\beta}\rangle=-\langle\bar{\partial}_{\dot\alpha}|\bar{\partial}_{\dot\beta}\rangle=i\bar{\epsilon}_{\dot\alpha\dot\beta} \,, && \langle h_{\alpha\dot\alpha}|h_{\beta\dot\beta}\rangle=\epsilon_{\alpha\beta}\epsilon_{\dot\alpha\dot\beta}\,,
\end{align}
where $q^\alpha=y^\alpha$, $p_\alpha=i\partial_\alpha$ and $h^{\alpha\dot\beta}$ is a vierbein one-form.

In some local coordinates $x^\mu$ on the base manifold (which in our case is $AdS_4$) the vierbein one-forms $h^\alpha_{\ \dot\alpha}$ can be expressed as
\begin{equation}\label{h_via_dx}
    h^\alpha_{\ \dot\alpha} = \left(h_{\mu}\right)^\alpha_{\ \dot\alpha}\,
    dx^\mu.
\end{equation}The $AdS_4$ vierbein $\left(h_{\mu}\right)^\alpha_{\ \dot\alpha}$ is
demanded to be non-degenerate at any point of
 $AdS_4$.

The next step is to obtain $\sigma_+:=\sigma_-^\dagger$ with respect to the scalar product $\langle\,,\,\rangle$, {\it i.e.},  $\langle \phi|\sigma_-\psi\rangle = \langle \sigma_+\phi|\psi\rangle$. It is not hard to check that
\begin{subequations}
    \begin{align}
    \sigma_+\omega(y,\bar y) &:= -iy^\alpha D_\alpha^{\ \dot\alpha}\bar{\partial}_{\dot\alpha}\ \omega(y,\bar y) \quad\quad\text{at } \deg_y(\omega)>\deg_{\bar y}(\omega)\,,\\
    \sigma_+\omega(y,\bar y) &:=  -i \bar y^{\dot\alpha}D^\alpha_{\ \dot\alpha}\partial_\alpha\ \omega(y,\bar y) \quad\quad \text{at } \deg_y(\omega)<\deg_{\bar y}(\omega)\,,\\
    \sigma_+ \omega(y,\bar y)&:= -i\left(y^\alpha D_\alpha^{\ \dot\alpha}\bar\partial_{\dot\alpha}+\bar y^{\dot\alpha}D^\alpha_{\ \dot\alpha}\partial_\alpha\right)\omega(y,\bar y)
    \quad\quad\text{at } \deg_y(\omega)=\deg_{\bar y}(\omega)\,,\label{3}
    \end{align}
\end{subequations}
where
\begin{align}
    D_\alpha^{\ \dot\alpha} := \frac{\dd}{\dd h^\alpha_{\ \dot\alpha}} \,,\qquad D^\alpha_{\ \dot\alpha} := \frac{\dd}{\dd h_\alpha^{\ \dot\alpha}}\,.
\end{align}
By $\omega = \omega(y,\bar y\,|x,h)$ we mean a general $p$-form  polynomial in $y$ and $\bar y$ with the coordinate one-forms $dx$  replaced  by $h$ via (\ref{h_via_dx}), that is
\begin{equation}
    \omega(y,\bar y\,|x,h) = \sum_{n,m} \omega_{\alpha_1\dots\alpha_p|\mu(n)|\dot\alpha_1\dots\dot\alpha_p|\dot\mu(m)}(x)\,h^{\alpha_1\dot\alpha_1}
    \wedge\dots\wedge h^{\alpha_p\dot\alpha_p}\,y^{\mu(n)}\,\bar y^{\dot\mu(m)}\,.
    \end{equation}
So defined  $\sigma_+$  increases the grading.
The Laplace operator
\begin{equation}
    \Delta:=\sigma_-\sigma_++\sigma_+\sigma_-
\end{equation}
is by construction self-adjoint with respect to $\langle\,|\,\rangle$ and non-negative definite
for the compact version of the space-time symmetry algebra.

\section{Bosonic case in $AdS_4$}

To calculate  cohomology of $\sigma_-$ we have to compute the action of $\Delta$. Since
$\sigma_-$ and $\sigma_+$ are defined differently in the different regions of the $(N, \bar N)$ plane, we
 compute the Laplacian action in the these regions separately. Direct computation yields:
\begin{subequations}\label{bosonicLaplaceAction}
    \begin{align}
        \Delta_{N>\bar N+2}&= N(\bar N+2) + y^\beta\dd_\alpha h^\alpha_{\ \dot\gamma} D_\beta^{\ \dot\gamma} + \bar y^{\dot\alpha}\bar\dd_{\dot\beta}h_{\gamma\dot\alpha}D^{\gamma\dot\beta}\,,\\
        \Delta_{N<\bar N-2}&=\bar N(N+2)+y^\alpha\dd_\beta h_{\alpha \dot\gamma}D^{\beta\dot\gamma}+\bar y^{\dot\alpha}\bar\dd_{\dot\beta} h_\gamma^{\ \dot\beta}D^{\gamma}_{\ \dot\alpha}\,,\\
        \Delta_{N=\bar N+2} &= \Delta_{N>\bar N+2} + \bar y^{\dot\alpha}\bar y^{\dot\beta}\dd_\alpha\dd_\beta h^\beta_{\ \dot\beta}D^\alpha_{\ \dot\alpha}\,,\\
        \Delta_{N=\bar N-2} &= \Delta_{N<\bar N-2} + y^\alpha y^\beta\bar\dd_{\dot\alpha}\bar\dd_{\dot\beta}h_\beta^{\ \dot\beta}D_\alpha^{\ \dot\alpha}\,,\\
        \Delta_{N=\bar N}&=
    \bar y^{\dot\alpha}\bar\dd_{\dot\beta} h_{\gamma\dot\alpha} D^{\gamma \dot\beta}
    +y^\alpha \dd_\beta h_{\alpha\dot\gamma}D^{\beta \dot\gamma}
    -\bar y^{\dot\alpha} y^\beta \dd_\alpha\bar\dd_{\dot\beta}h^\alpha_{\ \dot\alpha} D_\beta^{\ \dot\beta}
    - y^\alpha\bar y^{\dot\beta}\dd_\beta\bar\dd_{\dot\alpha}h_\alpha^{\ \dot\alpha}D^\beta_{\ \dot\beta}\,.
    \label{6}
    \end{align}
\end{subequations}
The computation of the cohomology $H^p(\sigma_-)$ will be performed as follows. Taking a general $p$-form $\omega(y,\bar y\,|x)$, we decompose it into Lorentz irreducible components. As we will observe, the projectors onto irreducible parts will commute with the action of the Laplacian. Thus, instead of involved calculation of the action of $\Delta$ on all of the irreducible components of $\omega(y,\bar y\,|x)$ we can first calculate its action on the general $\omega(y,\bar y\,|x)$ and then project.

\subsection{ $H^0(\sigma_-)$}
Evidently,  $\Delta_{N=\bar N}\Big|_\text{0-forms}= 0$ since all the terms in (\ref{6})
contain  derivatives in $h$'s. At the same time,  $\Delta_{N\neq\bar N}\Big|_\text{0-forms} > 0$.
Thus, we conclude
\begin{equation}
    H^0(\sigma_-)=\ker\left(\Delta\Big|_\text{0-forms}\right)= \left\{F(y,\bar y)=F_{\alpha(n),\dot\alpha(n)}y^{\alpha(n)}\bar y^{\dot\alpha(n)},\ n\in\mathbb{N}_0\right\}.
\end{equation}
By Theorem \ref{Theorem_cohomology}, elements   of this cohomology space  correspond to the parameters of differential (non-Stueckelberg) linearized  HS gauge transformations. This result fits the pattern of the
spin-$s$ Fronsdal gauge symmetry parameters with $n=s-1$.

 \subsection{ $H^1(\sigma_-)$}
The decomposition of a one-form $\Theta(y,\bar y\,|x)$ into Lorentz irreps reads as
\begin{multline}\label{omega1irrep}
    \Theta(y,\bar y\,|x)=
    \underbrace{\Theta_{\mu(n+1)|\dot\mu(m+1)}\ h^{\mu\dot\mu}y^{\mu(n)}y^{\dot\mu(m)}}_{\Theta_\mathrm{A}(y,\bar y)}
    -\frac{1}{2}\underbrace{\Theta_{\mu(n-1)|\dot\mu(m+1)}\ h_\nu^{\ \dot\mu}y^\nu y^{\mu(n-1)}\bar y^{\dot\mu(m)}}_{\Theta_\mathrm{B}(y,\bar y)} -\\
    -\frac{1}{2}\underbrace{\Theta_{\mu(n+1)|\dot\mu(m-1)}\ h^{\mu}_{\ \dot\nu}y^{\mu(n)}\bar y^{\dot\nu}\bar y^{\dot\mu(m-1)}}_{\Theta_\mathrm{C}(y,\bar y)} +
    \frac{1}{4}\underbrace{\Theta_{\mu(n-1)|\dot\mu(m-1)}\ h_{\nu\dot\nu} y^\nu y^{\mu(n-1)}\bar y^{\dot\nu}\bar y^{\dot\mu(m-1)}}_{\Theta_\mathrm{D}(y,\bar y)}.
\end{multline}
Thus, for fixed $n$ and $m$, there are four Lorentz-irreducible one-forms: $\Theta_\text{A}$, $\Theta_\mathrm{B}$, $\Theta_\mathrm{C}$, $\Theta_\mathrm{D}$.\\
For direct computations it will be convenient to separate two of the indices of the $y$ group:
\begin{equation}\label{Theta-1-form-general}
    \Theta(y,\bar y) = \Theta_{\lambda,\nu,\mu(n-1)|\dot\lambda,\dot\nu,\dot\mu(m-1)}\  h^{\lambda\dot\lambda}y^{\nu}y^{\mu(n-1)}\bar y^{\dot\nu}\bar y^{\dot\mu(m-1)}.
\end{equation}
In terms of the basis one-forms $h^{\lambda\dot\lambda}y^{\nu}y^{\mu(n-1)}\bar y^{\dot\nu}\bar y^{\dot\mu(m-1)}$ the projectors onto irreducible components are
\begin{align}\label{irreducible-THeta_ABCD}
    \mathcal{P}_\mathrm{A} &= \mathscr{S}_{(\lambda,\mu,\nu)}\mathscr{S}_{(\dot\lambda,\dot\mu,\dot\nu)}, && \mathcal{P}_\mathrm{C} = \mathscr{S}_{(\lambda,\mu,\nu)}\epsilon_{\dot\lambda\dot\nu},\\
    \mathcal{P}_\mathrm{B} &= \epsilon_{\lambda\nu}\mathscr{S}_{(\dot\lambda,\dot\mu,\dot\nu)}, &&
    \mathcal{P}_\mathrm{D} = \epsilon_{\lambda\nu}\epsilon_{\dot\lambda\dot\nu},
\end{align}
where $\mathscr{S}_{(\lambda,\mu,\nu)}$ implies
 symmetrization over indices $\lambda,\mu,\nu$ and similarly for the dotted indices.

\subsubsection{$H^1(\sigma_-)$ in the diagonal sector $N=\bar N$}
 In the diagonal sector with $n=m$ the Laplacian is a sum of the following four terms:
\begin{multline}\label{e:diag_oneform}
    \Delta_{N=\bar N}\Theta(y,\bar y) =
    \underbrace{\bar y^{\dot\alpha}\bar\dd_{\dot\beta} h_{\gamma\dot\alpha} D^{\gamma \dot\beta} \Theta(y,\bar y) }_{T_1(y,\bar y)}
    +\underbrace{y^\alpha \dd_\beta h_{\alpha\dot\gamma}D^{\beta \dot\gamma}\Theta(y,\bar y) }_{T_2(y,\bar y)}-\\
    \underbrace{-\bar y^{\dot\alpha} y^\beta \dd_\alpha\bar\dd_{\dot\beta}h^\alpha_{\ \dot\alpha} D_\beta^{\ \dot\beta}\Theta(y,\bar y) }_{T_3(y,\bar y)}
    \underbrace{-y^\alpha\bar y^{\dot\beta}\dd_\beta\bar\dd_{\dot\alpha}h_\alpha^{\ \dot\alpha}D^\beta_{\ \dot\beta}\Theta(y,\bar y) }_{T_4(y,\bar y)}.
\end{multline}
Consider the first term in (\ref{e:diag_oneform}).
\begin{equation}\label{notation-T1}
    \bar y^{\dot\alpha}\bar\dd_{\dot\beta} h_{\gamma\dot\alpha} D^{\gamma \dot\beta} \Theta(y,\bar y) = \Theta_{\lambda,\nu,\mu(n-1)|\dot\lambda,\dot\nu,\dot\mu(m-1)}\ \underbrace{\Big[\bar y^{\dot\alpha}\bar\dd_{\dot\beta} h_{\gamma\dot\alpha} D^{\gamma \dot\beta} \big(  h^{\lambda\dot\lambda}y^{\nu}y^{\mu(n-1)}\bar y^{\dot\nu}\bar y^{\dot\mu(m-1)}\big)\Big]}_{T_1^{\lambda,\nu,\mu(n-1)|
    \dot\lambda,\dot\nu,\dot\mu(n-1)}}\,.
\end{equation}
The expression in square brackets is denoted by $T_1^{\lambda,\nu,\mu(n-1)|\dot\lambda,\dot\nu,\dot\mu(n-1)}$. The notation for other irreducible components $T_2$, $T_3$ and $T_4$ is analogous. Straightforward computation yields
\begin{subequations}
    \begin{eqnarray}
    T_1^{\lambda,\nu,\mu(n-1)|\dot\lambda,\dot\nu,\dot\mu(n-1)} &=& -h^\lambda_{\ \dot\alpha}\epsilon^{\dot\nu\dot\lambda} y^\nu y^{\mu(n-1)}\bar y^{\dot\alpha}\bar y^{\dot\mu(n-1)} -{}\nonumber \\&&- (n-1) h^\lambda_{\ \dot\alpha}\epsilon^{\dot\mu\dot\lambda}y^\nu y^{\mu(n-1)}\bar y^{\dot\alpha}\bar y^{\dot\nu}\bar y^{\dot\mu(n-2)}\,,\\
    T_2^{\lambda,\nu,\mu(n-1)|\dot\lambda,\dot\nu,\dot\mu(n-1)} &=&    -h_\alpha^{\ \dot\lambda}\epsilon^{\nu\lambda}y^{\alpha}y^{\mu(n-1)}\bar y^{\dot\nu}\bar y^{\dot\mu(n-1)} -{}\nonumber \\&&- (n-1)h_\alpha^{\ \dot\lambda}\epsilon^{\mu\lambda}y^\alpha y^\nu y^{\mu(n-2)}\bar y^{\dot\nu}\bar y^{\dot\mu(n-1)}\,,\\
    T_3^{\lambda,\nu,\mu(n-1)|\dot\lambda,\dot\nu,\dot\mu(n-1)} &=& -\left(h^\nu_{\ \dot\alpha}y^\lambda y^{\mu(n-1)}+(n-1)h^\mu_{\ \dot\alpha}y^\lambda y^\nu y^{\mu(n-2)}\right){}\nonumber\times\\&&\times\left(\epsilon^{\dot\nu\dot\lambda}\bar y^{\dot\alpha}\bar y^{\dot\mu(n-1)} + (n-1)\epsilon^{\dot\mu\dot\lambda}\bar y^{\dot\alpha}\bar y^{\dot\nu}\bar y^{\dot\mu(n-2)}\right)\,,\\
    T_4^{\lambda,\nu,\mu(n-1)|\dot\lambda,\dot\nu,\dot\mu(n-1)} &=& -\left(h_\alpha^{\ \dot\nu}\bar y^{\dot\lambda} \bar y^{\dot \mu(n-1)}+(n-1)h_{\alpha}^{\ \dot\nu} \bar y^{\dot\lambda} \bar y^{\dot\nu} \bar y^{\dot\mu(n-2)}\right){}\nonumber\times\\&&\times\left(\epsilon^{\nu\lambda} y^{\alpha} y^{\mu(n-1)} + (n-1)\epsilon^{\mu\lambda} y^{\alpha} y^{\nu} y^{\mu(n-2)}\right).
\end{eqnarray}
\end{subequations}
Projecting onto the irreducible parts of $\Theta(y,\bar y)$ we find
\begin{subequations}
    \begin{align}
        \Delta_{N=\bar N}\left(\Theta_{\mathrm{A}}\right)&=0\,,\\
        \Delta_{N=\bar N}\left(\Theta_{\mathrm{B}}\right)&=(n+1)^2\Theta_{\mathrm{B}}\neq 0\,,\\
        \Delta_{N=\bar N}\left(\Theta_{\text{C}}\right)&=(n+1)^2\Theta_{\text{C}}\neq 0\,,\\
        \Delta_{N=\bar N}\left(\Theta_{\mathrm{D}}\right)&=0.
    \end{align}
\end{subequations}
Thus, the only elements of  the kernel of $\Delta^\text{1-forms}\Big|_{N=\bar N}$ are $\Theta_{\mathrm{A}}(y,\bar y)$ and $\Theta_{\mathrm{D}}(y,\bar y)$. By the Hodge theorem of Section 4 this yields that $H^1(\sigma_-) = \ker\left(\Delta^\text{1-forms}\Big|_{N=\bar N}\right)$ is
\begin{subequations}
 \begin{align}\label{bosonic-H1-explicit}
    H^1(\sigma_-) &= \bigoplus_{n\geq 0} H^1_{(n)}(\sigma_-),\\
    H^1_{(n)}(\sigma_-) &= \Big\{\phi_{(n)}(y,\bar y\,|x) + \phi_{(n)}^\text{tr}(y,\bar y\,|x),\\
    \phi_{(n)}(y,\bar y\,|x) &:= \phi_{\mu(n+1),\dot\mu(n+1)}(x)\, h^{\mu\dot\mu} \, y^{\mu(n)}\bar y^{\dot\mu(n)},\\
    \phi_{(n)}^\text{tr}(y,\bar y\,|x) &:= \phi^\mathrm{tr}_{\mu(n-1),\dot\mu(n-1)}(x)\, h_{\nu\dot\nu} \, y^\nu y^{\mu(n-1)}\bar y^{\dot\nu}\bar y^{\dot\mu(n-1)}\,, \text{ if }  n>0 \Big\},
\end{align}
\end{subequations}
where $n$ is  the number of indices of the corresponding cocycles.

Equivalently,
\begin{equation}\label{bosonic-H1-genfunc}
    H^1(\sigma_-) = \Big\{ h^{\mu\dot\mu}\,\partial_\mu\bar\partial_{\dot\mu}\,F_1(y,\bar y\,|x) + h_{\mu\dot\mu}\,y^\mu\bar y^{\dot\mu} F_2(y,\bar y\,|x)\Big\},
\end{equation}
where $F_{1,2}(y,\bar y\,|x)$ belongs to the diagonal $N=\bar N$, that is,
\begin{equation}
    \left(y^\alpha\frac{\partial}{\partial y^\alpha} - \bar y^{\dot\alpha}\frac{\partial}{\partial \bar y^{\dot\alpha}}\right)F_{1,2}(y,\bar y\,|x) = 0.
\end{equation}
The fields $\phi(y,\bar y)$ and $\phi^\text{tr}(y,\bar y)$ exactly correspond to the irreducible components of the double-traceless Fronsdal field.

It remains  to prove that there are no other nontrivial cocycles in $H^1(\sigma_- )$ at $N\neq \bar N$.

\subsubsection{$H^1(\sigma_-)$ in the far-from-diagonal sector $|N-\bar N|>2$}
Consider  the action of the  Laplace operator $\Delta_{N>\bar N+2}$ on  general one-forms at $N>\bar N+2$ \begin{multline}
    \Delta_{N>\bar N+2}\Theta(y, \bar y) =
    \left(n(m+2) + y^\beta\dd_\alpha h^\alpha_{\ \dot\gamma} D_\beta^{\ \dot\gamma} + \bar y^{\dot\alpha}\bar\dd_{\dot\beta}h_{\gamma\dot\alpha}D^{\gamma\dot\beta}\right)\Theta(y, \bar y) =\\
    =\underbrace{n(m+2)\Theta(y, \bar y)}_{T_1(y,\bar y)} +
    \underbrace{y^\beta\dd_\alpha h^\alpha_{\ \dot\gamma} D_\beta^{\ \dot\gamma}\Theta(y, \bar y)}_{T_2(y,\bar y)} +
    \underbrace{\bar y^{\dot\alpha}\bar\dd_{\dot\beta}h_{\gamma\dot\alpha}D^{\gamma\dot\beta}\Theta(y,\bar y)}_{T_3(y,\bar y)}.
\end{multline}
Analogously to (\ref{notation-T1}), we denote
\begin{equation}
    T_2^{\lambda,\nu,\mu(n-1)|\dot\lambda,\dot\nu,\dot\mu(m-1)} = y^\beta\dd_\alpha h^\alpha_{\ \dot\gamma} D_\beta^{\ \dot\gamma} \big( h^{\lambda\dot\lambda}y^{\nu}y^{\mu(n-1)}\bar y^{\dot\nu}\bar y^{\dot\mu(m-1)}\big)
\end{equation}
and similarly for $T_1$ and $T_3$. In this sector we obtain
\begin{subequations}
    \begin{align}
    T_1^{\lambda,\nu,\mu(n-1)|\dot\lambda,\dot\nu,\dot\mu(m-1)} &=n(m+2)h^{\lambda\dot\lambda}y^{\nu}y^{\mu(n-1)}\bar y^{\dot\nu}\bar y^{\dot\mu(m-1)}\,,\\
    T_2^{\lambda,\nu,\mu(n-1)|\dot\lambda,\dot\nu,\dot\mu(m-1)} &=
    -h^{\nu\dot\lambda}y^{\lambda}y^{\mu(n-1)}\bar y^{\dot\nu}\bar y^{\dot\mu(m-1)} - (n-1)h^{\mu\dot\lambda}y^{\lambda}y^{\nu}y^{\mu(n-2)}\bar y^{\dot\nu}\bar y^{\dot\mu(m-1)}\,,\\
    T_3^{\lambda,\nu,\mu(n-1)|\dot\lambda,\dot\nu,\dot\mu(m-1)} &= - h^{\lambda}_{\ \dot\alpha}\epsilon^{\dot\nu\dot\lambda}y^{\nu}y^{\mu(n-1)}\bar y^{\dot\alpha}\bar y^{\dot\mu(m-1)} - (m-1) h^{\lambda}_{\ \dot\alpha}\epsilon^{\dot\lambda\dot\mu}y^{\nu}y^{\mu(n-1)}\bar y^{\dot\alpha}\bar y^{\dot\nu}\bar y^{\dot\mu(m-2)}\,.
    \end{align}
\end{subequations}
Projection onto the irreducible components according to (\ref{irreducible-THeta_ABCD}) yields
\begin{subequations}
    \begin{align}
        \Delta_{N>\bar N+2}\left(\Theta_\mathrm{A}\right) &= n(m+1)\Theta_\mathrm{A}\neq0\,,\\
        \Delta_{N>\bar N+2}\left(\Theta_\mathrm{B}\right) &= (nm+2n+1)\Theta_\mathrm{B}\neq0\,,\\
        \Delta_{N>\bar N+2}\left(\Theta_\mathrm{C}\right) &= (nm-n+2m)\Theta_\mathrm{C}\neq0\,,\\
        \Delta_{N>\bar N+2}\left(\Theta_\mathrm{D}\right) &= (nm+2n-m+4)\Theta_\mathrm{D}\neq0\,.
    \end{align}
\end{subequations}
Thus, there are no nontrivial cocycles in this sector. For $N<\bar N-2$ the computation is analogous.  Thus, $H^1(\sigma_-) =0$ in the far-from-diagonal sector.

\subsubsection{Subtlety in the near-diagonal sector $|N-\bar N|=2$}
In this case we face  certain peculiarity. Denote the space of $p$-forms ($p=1$ for $H^{1}$ ) with $N$ chiral and $\bar{N}$ anti-chiral indices by $\mathcal{V}_{(N,\bar{N})}$. Recall that the grading operator is $G=|N-\bar{N}|$. Consider the case with $N-\bar N=2$. Namely, let $N=n+1$ and $\bar N=n-1$. At $G=2$ the operator $\sigma_-$ maps a state $X\in \mathcal V_{(n+1,n-1)}$ onto the diagonal, $\sigma_-(X)\in \mathcal V_{(n,n)}$, where in accordance with (\ref{3}),  $\sigma_+$ acts 'both up and down':
\begin{equation}
    \mathcal V_{(n+1,n-1)} \xrightarrow{\sigma_-} \mathcal V_{(n,n)} \xrightarrow{\sigma_+} \mathcal V_{(n-1,n+1)} \oplus \mathcal V_{(n+1,n-1)}\,.
\end{equation}
Thus,
\begin{equation}
    \mathcal V_{(n+1,n-1)} \xrightarrow{\sigma_+\sigma_-}\mathcal V_{(n-1,n+1)} \oplus \mathcal V_{(n+1,n-1)}\,.
\end{equation}
As a result,
\begin{align}
    \Delta_{(n+1,n-1)} : \mathcal V_{(n+1,n-1)} \longrightarrow \mathcal V_{(n-1,n+1)} \oplus \mathcal V_{(n+1,n-1)}\,,\\
    \Delta_{(n-1,n+1)} : \mathcal V_{(n-1,n+1)} \longrightarrow \mathcal V_{(n-1,n+1)} \oplus \mathcal V_{(n+1,n-1)}\,.
\end{align}
Consequently, $\ker(\Delta)$ should be searched in the form of a linear combination of the vectors both from $\mathcal{V}_{(n+1,n-1)}$ and from $\mathcal{V}_{(n-1,n+1)}$.

Indeed, let $X$ be a vector in  $\mathcal{V}_{(n+1,n-1)}$. Consider the complex conjugated  vector $\bar X\in\mathcal{V}_{(n-1,n+1)}$ and compute the action of the Laplacian on them. Let
\begin{align}
    \Delta X = \Delta_{(n+1,n-1)}X=\alpha(n) X + \beta(n) \bar X\,,\\
    \Delta \bar X = \Delta_{(n-1,n+1)}\bar X= \gamma(n) X + \delta(n) \bar X
\end{align}
with some coefficients $\alpha$, $\beta$, $\gamma$, and $\delta$. That $X$ and $\bar X$ are conjugated  and  operator $\Delta$ is self-adjoint implies the relations $\alpha = \bar\delta$ and $\beta=\bar\gamma$. Looking for $\ker(\Delta)$ in the form
\begin{equation}
    Y=F(n)X+G(n)\bar X\in\ker(\Delta)
\end{equation}
and acting on  $Y$ by the Laplace operator
we find that the condition
 $\Delta Y =0$ yields
\begin{multline}
    \Delta Y = F(n)\Delta_{(n+1,n-1)}X + G(n)\Delta_{(n-1,n+1)}\bar X =\\
    =\big(\alpha(n)F(n)+\bar{\beta(n)}G(n)\big)X + \big(\beta(n)F(n)+\bar{\alpha(n)}G(n)\big)\bar X = 0\,.
\end{multline}
Since $X$ and $\bar X$ are linearly independent, the problem of finding such  $Y=\alpha X+\beta \bar X$  that $\Delta Y=0$, amounts to the linear system
\begin{equation}
    \begin{bmatrix}
    \alpha(n) & \bar{\beta(n)} \\
    \beta(n) & \bar{\alpha(n)}
    \end{bmatrix}
    \begin{bmatrix}
    F(n) \\
    G(n)
    \end{bmatrix}
    =
    \begin{bmatrix}
    0 \\
    0
    \end{bmatrix}\,,
\end{equation}
which admits non-trivial solutions iff
\begin{equation}
    \det
    \begin{bmatrix}
    \alpha(n) & \bar{\beta(n)} \\
    \beta(n) & \bar{\alpha(n)}
    \end{bmatrix}
    =|\alpha(n)|^2 - |\beta(n)|^2 = 0.
\end{equation}
Hence, we conclude that
\begin{equation}
    \alpha(n)=\beta(n)\cdot e^{i\chi}, \quad\quad \chi\in[0,2\pi)\,.
\end{equation}
In the next section  coefficients $\alpha(n)$ and $\beta(n)$
will be shown to be
 real, \ie $e^{i\chi}=\pm 1$.

Summarizing, if  we find that the coefficients $\alpha(n)$ and $\beta(n)$ coincide up to a sign, $\alpha(n)=\pm \beta(n)$, this would imply the existence
 of a non-trivial $\sigma_-$-cocycle
\begin{equation}\label{Y=X+barX}
    Y = X \mp \bar X\,.
\end{equation}
Otherwise the cohomology is trivial.

\subsubsection{$H^1(\sigma_-)$ in the near-diagonal sector $|N-\bar N|=2$}
To compute $H^1(\sigma_-)$ in the  leftover sector of $N=\bar N+2$ (analysis at $N=\bar N-2$ is analogous) consider a general one-form $\Theta(y,\bar y)$  (\ref{Theta-1-form-general})
with $N=\bar N+2$.

In this sector, the Laplacian  differs form that at $N>\bar N+2$ by the $T_4(y,\bar y)$ term in
\begin{equation}
    \Delta_{N=\bar N+2}\Theta(y, \bar y) = \Big(\underbrace{N(\bar N+2) + y^\beta\dd_\alpha h^\alpha_{\ \dot\gamma} D_\beta^{\ \dot\gamma} + \bar y^{\dot\alpha}\bar\dd_{\dot\beta}h_{\gamma\dot\alpha}D^{\gamma\dot\beta}}_{\Delta_{N>\bar N+2}}\Big)\Theta(y, \bar y) +\underbrace{\bar y^{\dot\alpha}\bar y^{\dot\beta}\dd_\alpha\dd_\beta h^\beta_{\ \dot\beta}D^\alpha_{\ \dot\alpha}\Theta(y, \bar y)}_{T_4(y,\bar y)}.
\end{equation}
Consequently, it is essential to compute the action of this additional term. As before (cf.  (\ref{notation-T1})), denote
\begin{equation*}
    T_4^{\lambda,\nu,\mu(n-1)|\dot\lambda,\dot\nu,\dot\mu(m-1)} = \bar y^{\dot\alpha}\bar y^{\dot\beta}\dd_\alpha\dd_\beta h^\beta_{\ \dot\beta}D^\alpha_{\ \dot\alpha} \big(h^{\lambda\dot\lambda}y^{\nu}y^{\mu(n-1)}\bar y^{\dot\nu}\bar y^{\dot\mu(m-1)}\big)\,.
\end{equation*}
This yields
\begin{multline}
    T_4^{\lambda,\nu,\mu(n-1)|\dot\lambda,\dot\nu,\dot\mu(m-1)}=(n-1)\epsilon^{\mu\lambda}h^{\nu}_{\ \dot\beta}y^{\mu(n-2)}\bar y^{\dot\lambda}\bar y^{\dot\beta}\bar y^{\dot\nu}\bar y^{\dot\mu(m-1)} + (n-1)\epsilon^{\nu\lambda}h^{\mu}_{\ \dot\beta}y^{\mu(n-2)}\bar y^{\dot\lambda}\bar y^{\dot\beta}\bar y^{\dot\nu}\bar y^{\dot\mu(m-1)}+\\
    +(n-1)(n-2)\epsilon^{\mu\lambda}h^{\mu}_{\ \dot\beta}y^{\nu}y^{\mu(n-3)}\bar y^{\dot\lambda}\bar y^{\dot\beta}\bar y^{\dot\nu}\bar y^{\dot\mu(m-1)}.
\end{multline}
Projecting $T_4$ onto the irreducible parts, we find:
\begin{align}
    \text{(A)}:&& \mathscr{S}_{(\lambda,\nu,\mu)}\mathscr{S}_{(\dot\lambda,\dot\nu,\dot\mu)}T_4^{\lambda,\nu,\mu(n-1)|\dot\lambda,\dot\nu,\dot\mu(m-1)}&=0\,,\\
    \text{(B)}:&& \ \ \mathscr{S}_{(\dot\lambda,\dot\nu,\dot\mu)}\epsilon_{\lambda\nu}T_4^{\lambda,\nu,\mu(n-1)|\dot\lambda,\dot\nu,\dot\mu(m-1)} &= -(n-1)(2n-1)h^{\mu}_{\ \dot\beta} y^{\mu(n-2)}\bar y^{\dot\beta}\bar y^{\dot\mu(m+1)}\,,\\
    \text{(C)}:&& \ \ \epsilon_{\dot\lambda\dot\nu}\mathscr{S}_{(\lambda,\nu,\mu)}T_4^{\lambda,\nu,\mu(n-1)|\dot\lambda,\dot\nu,\dot\mu(m-1)}&=0\,,\\
    \text{(D)}:&& \ \ \epsilon_{\dot\lambda\dot\nu}\epsilon_{\lambda\nu}T_4^{\lambda,\nu,\mu(n-1)|\dot\lambda,\dot\nu,\dot\mu(m-1)}&=0\,.
\end{align}
We observe that the action of the Laplacian in this sector differs from the previously computed one only in the type-(B) sector, namely,
\begin{equation}
    \Delta_{N=\bar N+2}\left(\Theta_\mathrm{B}\right) = \left(n(m+3)\Theta_\mathrm{B}+(2n^2-3n+1)\Theta_\mathrm{C}\right)\Big|_{m=n-2}=
    \underbrace{(n^2+n)}_{\alpha(n)}\Theta_\mathrm{B}+\underbrace{(2n^2-3n+1)}_{\beta(n)}\Theta_\mathrm{C}\,.
\end{equation}
That $|\alpha(n)|\neq|\beta(n)|$ at integer $n$
implies triviality of  $H^1(\sigma_-)$ in the near-diagonal sector.

\subsection{$H^2(\sigma_-)$}
The calculation of $H^2(\sigma_-)$ is in main features analogous  to that of $H^1(\sigma_-)$. To decompose a general two-form $\Omega(y,\bar y\,|x)$
into  irreducible parts we use the following useful identity
\begin{equation}
    h^{\nu\dot\nu}\wedge h^{\lambda\dot\lambda} = \frac{1}{2}H^{\nu\lambda}\epsilon^{\dot\nu\dot\lambda}+\frac{1}{2}\bar H^{\dot\nu\dot\lambda}\epsilon^{\nu\lambda},
\end{equation}
where
\begin{align}
    H^{\nu\lambda}=H^{(\nu\lambda)}:=h^{\nu}_{\ \dot\gamma}\wedge h^{\lambda\dot\gamma}\,,\qquad
    \bar H^{\dot\nu\dot\lambda}=H^{(\dot\nu\dot\lambda)}:=h^{\ \dot\nu}_{\gamma}\wedge h^{\gamma\dot\lambda}\,.
\end{align}
The decomposition reads
\begin{multline}
    \Omega(y,\bar y\, |x) = \underbrace{\Omega^\mathrm{A}_{\mu(n+2)|\dot\mu(m)}\ H^{\mu\mu}y^{\mu(n)}\bar y^{\dot\mu(m)}}_{\Phi_{\mathrm{A}(n,m)}}+
    \underbrace{\bar{\Omega}^\mathrm{A}_{\mu(n)|\dot\mu(m+2)}\ \bar H^{\dot\mu\dot\mu}y^{\mu(n)}\bar y^{\dot\mu(m)}}_{\bar\Phi_{\mathrm{A}(n,m)}}+\\
    +\underbrace{\Omega^\mathrm{B}_{\mu(n-2)|\dot\mu(m)}\ H_{\nu\nu}y^\nu y^\nu y^{\mu(n-2)}\bar y^{\dot\mu(m)}}_{\Phi_{\mathrm{B}(n,m)}}+
    \underbrace{\bar\Omega^\mathrm{B}_{\mu(n)|\dot\mu(m-2)}
     \bar{H}_{\dot\nu\dot\nu} y^{\mu(n)}\bar y^{\dot\nu}\bar y^{\dot\nu}\bar y^{\dot\mu(m-2)}}_{\bar\Phi_{\mathrm{B}(n,m)}}+\\
     +\underbrace{\Omega^\mathrm{C}_{\mu(n)|\dot\mu(m)}\ H_{\nu}^{\ \mu}y^{\nu}y^{\mu(n-1)}\bar y^{\dot\mu(m)}}_{\Phi_{\mathrm{C}(n,m)}}+
     \underbrace{\bar\Omega^\mathrm{C}_{\mu(n)|\dot\mu(m)}\ \bar{H}_{\dot\nu}^{\ \dot\mu}y^{\mu(n)}\bar y^{\dot\nu}\bar y^{\dot\mu(m-1)}}_{{\bar\Phi_{\mathrm{C}(n,m)}}}.\label{dec}
\end{multline}
Consider now the reducible two-forms
\begin{subequations}
    \begin{align}
    \Phi^{\lambda(2),\nu(2),\mu(n-2)|\dot\mu(m)} = H^{\lambda\lambda}y^{\nu}y^{\nu}y^{\mu(n-2)}\bar y^{\dot\mu(m)},\label{phi-2-holomorphic}\\
    \bar\Phi^{\mu(n)|\dot\lambda(2),\dot\nu(2),\dot\mu(m-2)} = \bar H^{\dot\lambda\dot\lambda}y^{\mu(n)}\bar y^{\dot\nu}\bar y^{\dot\nu}\bar y^{\dot\mu(m-2)}. \label{phi-2-antiholomorphic}
\end{align}
\end{subequations}
In these terms, the projectors onto irreducible components are
\begin{align}
    \mathcal{P}_\mathrm{A}&=\mathscr{S}_{(\lambda,\nu,\mu)}\,, && \mathcal{P}_\mathrm{B} = \epsilon_{\lambda\nu}\epsilon_{\lambda\nu}\,, && \mathcal{P}_\mathrm{C} = \mathscr{S}_{(\lambda,\nu,\mu)}\circ\epsilon_{\lambda\nu}\,,\\
    \bar{\mathcal{P}}_\mathrm{A}&=\mathscr{S}_{(\dot\lambda,\dot\nu,\dot\mu)}\,, && \bar{\mathcal{P}}_\mathrm{B}=\epsilon_{\dot\lambda\dot\nu}\epsilon_{\dot\lambda\dot\nu}\,, && \bar{\mathcal{P}}_\mathrm{C}=\mathscr{S}_{(\dot\lambda,\dot\nu,\dot\mu)}\circ
    \epsilon_{\dot\lambda\dot\nu}\,
\end{align}
and the decomposition (\ref{dec}) reads as
\begin{align}
    \Phi_\mathrm{A}&=\mathcal{P}_\mathrm{A}\Phi\,, && \Phi_\mathrm{B}=\mathcal{P}_\mathrm{B}\Phi\,, && \Phi_\mathrm{C}=\mathcal{P}_\mathrm{C}\Phi\,, \\  \bar\Phi_\mathrm{A}&=\bar{\mathcal{P}}_\mathrm{A}\bar\Phi\,, && \bar\Phi_\mathrm{B}=\bar{\mathcal{P}}_\mathrm{B}\bar\Phi\,, && \bar\Phi_\mathrm{C}=\bar{\mathcal{P}}_\mathrm{C}\bar\Phi\,.
\end{align}

For practical calculations we have to find the result of the action of
the operator $D=\frac{\partial}{\partial h}$ on the two-form $H$. The result is
\begin{equation}
    D_{\alpha\dot\beta}H^{\nu\lambda} = D_{\alpha\dot\beta}\left(h^\nu_{\ \dot\gamma}\wedge h^{\lambda\dot\gamma}\right) = \epsilon_\alpha^{\ \nu}\epsilon_{\dot\beta\dot\gamma} h^{\lambda\dot\gamma} - h^{\nu}_{\ \dot\gamma}\,\epsilon_\alpha^{\ \lambda}\epsilon_{\dot\beta}^{\ \dot\gamma} = - \epsilon_\alpha^{\ \nu} h^\lambda_{\ \dot\beta} - \epsilon_\alpha^{\ \lambda}h^\nu_{\ \dot\beta} = -2\epsilon_\alpha^{\ (\nu} h^{\lambda)}_{\ \dot\beta}\,,
\end{equation}
or, in the condensed notation for  symmetrized indices,
\begin{equation}
    D_{\alpha\dot\beta}H^{\nu\nu}=-2\epsilon_\alpha^{\ \nu}h^{\nu}_{\ \dot\beta}.
\end{equation}

\subsubsection{$H^2(\sigma_-)$ in the far-from-diagonal sector $|N-\bar N|>2$}
Compute  $\Delta_{N>\bar N+2}$ on the general two-forms $\Phi$ and $\bar\Phi$,
\begin{equation}\label{DeltaTheta>>2}
    \Delta_{N>\bar N+2}(\Phi)=\underbrace{n(m+2)\Phi(y,\bar y)}_{T_1(y,\bar y)} + \underbrace{y^\beta\dd_\alpha h^\alpha_{\ \dot\gamma} D_\beta^{\ \dot\gamma}\Phi(y,\bar y)}_{T_2(y,\bar y)} + \underbrace{\bar y^{\dot\alpha}\bar\dd_{\dot\beta}h_{\gamma\dot\alpha}D^{\gamma\dot\beta}\Phi(y,\bar y)}_{T_3(y,\bar y)}\,.
\end{equation}
As in (\ref{notation-T1}), denote
\begin{equation}
    T_2^{\lambda(2),\nu(2),\mu(n-2)|\dot\mu(m)} = y^\beta\dd_\alpha h^\alpha_{\ \dot\gamma} D_\beta^{\ \dot\gamma} H^{\lambda\lambda}y^{\nu}y^{\nu}y^{\mu(n-2)}\bar y^{\dot\mu(m)}
\end{equation}
and similarly for $T_1$ and $T_3$. Straightforward computation yields
\begin{subequations}
    \begin{align}
    T_1^{\lambda(2),\nu(2),\mu(n-2)|\dot\mu(m)} &= n(m+2) \, H^{\lambda\lambda}y^{\nu}y^{\nu}y^{\mu(n-2)}\bar y^{\dot\mu(m)}\,,\\
    T_2^{\lambda(2),\nu(2),\mu(n-2)|\dot\mu(m)} &=-4 H^{\nu\lambda} y^{\lambda}y^{\nu}y^{\mu(n-2)}\bar y^{\dot\mu(m)} - 2(n-2)\,H^{\mu\lambda}y^{\lambda}y^{\nu(2)}y^{\mu(n-3)}\bar y^{\dot\mu(m)}\,,\\
    T_3^{\lambda(2),\nu(2),\mu(n-2)|\dot\mu(m)} &=m\, H^{\lambda\lambda}y^{\nu(2)} y^{\mu(n-2)}\bar y^{\dot\mu(m)}.
    \end{align}
\end{subequations}
Projecting onto the irreducible part $\Phi_\mathrm{A}$ we obtain
\begin{equation}
    \Delta_{N>\bar N+2}\Phi_\mathrm{A}=\mathcal{P}_\mathrm{A}\left(T_1+T_2+T_3\right)=\left[n(m+2)-2n+m\right]\Phi_\mathrm{A}=m(n+1)\Phi_\mathrm{A}.
\end{equation}
We  see that $\Phi_\mathrm{A}\in\ker(\Delta)$ whenever $m=0$. This gives a 2-cocycle of the form $H^{\mu\mu}y^{\mu(n)}$. It can be represented in terms of the generating function as follows. Contract all the indices in $H^{\mu\mu}y^{\mu(n)}$ with some symmetric coefficients $\Omega^\mathrm{A}_{\mu\mu\mu(n)}$ to obtain
\begin{equation}
    \Omega^\mathrm{A}_{\mu(n+2)}H^{\mu\mu}y^{\mu(n)}\equiv h^{(\lambda}_{\ \dot\gamma}\wedge h^{\nu\dot\gamma)}\Omega_{(\lambda\nu\mu(n))}y^{\mu(n)}\Rightarrow h^{\lambda}_{\ \dot\gamma}\wedge h^{\nu\dot\gamma}\,\dd_\lambda\dd_\nu C(y,0\, |x) \in \ker\left(\Delta_{N>\bar N+2}\Big|_\text{2-forms}\right),
\end{equation}
where  $C(y,0\, |x)=\Omega_{\mu\mu\mu(n)} y^\mu y^{\mu} y^{\mu(n)}$. Summarizing, we found a part of the kernel of $\Delta$ represented by the two-forms
\begin{equation}
    W(y,0\,|x)=H^{\mu\nu}\dd_\mu\dd_\nu C(y,0\,|x)
\end{equation}
with $C(y,0|x)$ being a general polynomial of $y$'s of degree $\geq 4$.

Let us now project (\ref{DeltaTheta>>2}) onto the second irreducible part $\Phi_\mathrm{B}$,
\begin{equation}
    \Delta_{N>\bar N+2}\Phi_\mathrm{B}=\mathcal{P}_\mathrm{B}\left(T_1+T_2+T_3\right)=\left[n(m+2)+0+m\right]\Phi_\mathrm{B}\neq 0\quad \forall n,m\in\mathbb{N}_0.
\end{equation}
Since $\Phi_{B(n,m)}$ is proportional to $H_{\nu\nu}y^\nu y^\nu y^{\mu(n-2)}\bar y^{\dot\mu(m)}$,  the case $n=m=0$ is beyond the allowed region.  Thus, $\Phi_\mathrm{B}$ does not contribute to $H^2(\sigma_-)$.

Projecting (\ref{DeltaTheta>>2}) onto $\Phi_\mathrm{C}$, we find
\begin{equation}
    \Delta_{N>\bar N+2}\Phi_\mathrm{C}=\mathcal{P}_\mathrm{C}\left(T_1+T_2+T_3\right) = [n(m+2)-2-(n-2)+m]\Phi_\mathrm{C} = (nm+n+m)\Phi_\mathrm{C}\,.
\end{equation}
Again, $\Phi_\mathrm{C}$ does not contribute to $H^2(\sigma_-)$ since $m>0$, $n\geq 0$.

Next, we  consider the anti-holomorphic two-form $\bar\Phi$. The action of the Laplacian yields
\begin{equation}
    \Delta_{N>\bar N+2}(\bar\Phi)=\underbrace{n(m+2)\bar\Phi(y,\bar y)}_{T_1(y,\bar y)} + \underbrace{y^\beta\dd_\alpha h^\alpha_{\ \dot\gamma} D_\beta^{\ \dot\gamma}\bar\Phi(y,\bar y)}_{T_2(y,\bar y)} + \underbrace{\bar y^{\dot\alpha}\bar\dd_{\dot\beta}h_{\gamma\dot\alpha}D^{\gamma\dot\beta}\bar\Phi(y,\bar y)}_{T_3(y,\bar y)}\,.
\end{equation}
As in (\ref{notation-T1}) we set
\begin{equation}
    T_2^{\mu(n)|\dot\lambda(2),\dot\nu(2),\dot\mu(m-2)} = y^\beta\dd_\alpha h^\alpha_{\ \dot\gamma} D_\beta^{\ \dot\gamma}\,\bar H^{\dot\lambda\dot\lambda}y^{\mu(n)}\bar y^{\dot\nu}\bar y^{\dot\nu}\bar y^{\dot\mu(m-2)}
\end{equation}
and analogously for $T_1$ and $T_3$. The computation in components yields
\begin{subequations}
    \begin{align}
    T_1^{\mu(n)|\dot\lambda(2),\dot\nu(2),\dot\mu(m-2)} &= n(m+2) \bar H^{\dot\lambda\dot\lambda}y^{\mu(n)}\bar y^{\dot\nu}\bar y^{\dot\nu}\bar y^{\dot\mu(m-2)},\\
    T_2^{\mu(n)|\dot\lambda(2),\dot\nu(2),\dot\mu(m-2)} &=
    -n\,\bar H^{\dot\lambda\dot\lambda} y^{\mu(n)}\bar y^{\dot\nu(2)}\bar y^{\dot\mu(m-2)},\\
    T_3^{\mu(n)|\dot\lambda(2),\dot\nu(2),\dot\mu(m-2)} &=
    -4\,\epsilon^{\dot\nu\dot\lambda}\bar H_{\dot\alpha}^{\ \dot\lambda} y^{\mu(n)}\bar y^{\dot\alpha}\bar y^{\dot\nu}\bar y^{\dot\mu(m-2)} - 2(m-2)\,\epsilon^{\dot\mu\dot\lambda}\bar H_{\dot\alpha}^{\ \dot\lambda}y^{\mu(n)}\bar y^{\dot\alpha}\bar y^{\dot\nu(2)}\bar y^{\dot\mu(m-3)}.
    \end{align}
\end{subequations}
Projecting onto the irreducible components we obtain
\begin{align}
    \Delta_{N>\bar N+2}\bar\Phi_\mathrm{A} &= \bar{\mathcal{P}}_\mathrm{A}\left(T_1+T_2+T_3\right)=[n(m+1)]\bar\Phi_\mathrm{A}\,,\\
    \Delta_{N>\bar N+2}\bar\Phi_\mathrm{B} &= \bar{\mathcal{P}}_\mathrm{B}\left(T_1+T_2+T_3\right)=(nm+n+4m)\bar\Phi_\mathrm{B}\,,\\
    \Delta_{N>\bar N+2}\bar\Phi_\mathrm{C} &= \bar{\mathcal{P}}_\mathrm{C}\left(T_1+T_2+T_3\right) =(nm+n-m)\bar\Phi_\mathrm{C}\,.
\end{align}
The condition $n>m+2$ valid in the far-from-diagonal sector does not allow  $\bar\Phi_\mathrm{A,B,C}$ to be in  the kernel of $ \Delta$.

The analysis of the opposite sector $N<\bar N-2$  is analogous via
swapping  dotted and undotted indices. As a result, the final answer
 for the under-diagonal sector is
\begin{equation}
    W(0,\bar y\, |x)=\bar H^{\dot\mu\dot\nu}\bar\dd_{\dot\mu}\bar\dd_{\dot\nu}C(0,\bar y\, |x).
\end{equation}

This completes the analysis of $H^2(\sigma_-)$ in the sector $|N-\bar N|>2$. The cohomology is represented by the two-forms
\begin{equation}
\label{weyl}
    W(y,\bar y\, |x) = h^{\mu}_{\ \dot\gamma}\wedge h^{\nu\dot\gamma}\,\dd_\mu\dd_\nu C(y,0\, |x) + h_\gamma^{\ \dot\mu}\wedge h^{\gamma\dot\nu}\,\bar\dd_{\dot\mu}\bar\dd_{\dot\nu} C(0,\bar y\, |x).
\end{equation}
These two-forms are known to represent the so-called Weyl cocycle in the HS theory. It is thus shown that there are no other non-trivial 2-cocycles
in this sector.

\subsubsection{$H^2(\sigma_-)$ on the diagonal $N=\bar N$}
Now we prove that there are no non-trivial cocycles at $N=\bar N$ except for the
Weyl cohomology (\ref{weyl}). As before, act by the operator $\Delta_{N=\bar N}$ on the two-form $\Phi^{\lambda(2)|\nu(2)|\mu(n-2)|\dot\mu(n)}(y,\bar y)=H^{\lambda\lambda}y^{\nu}y^{\nu}y^{\mu(n-2)}\bar y^{\dot\mu(n)}$
\begin{multline}
    \Delta_{N=\bar N}\Phi(y,\bar y)=
    \underbrace{\bar y^{\dot\alpha}\bar\dd_{\dot\beta} h_{\gamma\dot\alpha} D^{\gamma \dot\beta} \Phi(y,\bar y)}_{T_1(y,\bar y)}
    +\underbrace{y^\alpha \dd_\beta h_{\alpha\dot\gamma}D^{\beta \dot\gamma}\Phi(y,\bar y)}_{T_2(y,\bar y)}-\\
    \underbrace{-\bar y^{\dot\alpha} y^\beta \dd_\alpha\bar\dd_{\dot\beta}h^\alpha_{\ \dot\alpha} D_\beta^{\ \dot\beta}\Phi(y,\bar y)}_{T_3(y,\bar y)}
    \underbrace{-y^\alpha\bar y^{\dot\beta}\dd_\beta\bar\dd_{\dot\alpha}h_\alpha^{\ \dot\alpha}D^\beta_{\ \dot\beta}\Phi(y,\bar y)}_{T_4(y,\bar y)}\,.
\end{multline}
Denoting
\begin{equation}
     T_1^{\lambda(2),\nu(2),\mu(n-2)|\dot\mu(n)} = \bar y^{\dot\alpha}\bar\dd_{\dot\beta} h_{\gamma\dot\alpha} D^{\gamma \dot\beta}\, H^{\lambda\lambda}y^{\nu(2)}y^{\mu(n-2)}\bar y^{\dot\mu(n)}
\end{equation}
and analogously for $T_2$, $T_3$ and $T_4$, straightforward computation yields
\begin{subequations}
    \begin{eqnarray}
    T_1^{\lambda(2),\nu(2),\mu(n-2)|\dot\mu(n)} & = & n\, H^{\lambda\lambda}y^{\nu(2)}y^{\mu(n-2)}\bar y^{\dot\mu(n)}\,,\\
    T_2^{\lambda(2),\nu(2),\mu(n-2)|\dot\mu(n)} & = & -4 \epsilon^{\nu\lambda} H_{\alpha}^{\ \lambda} y^{\alpha}y^{\nu}y^{\mu(n-2)}\bar y^{\dot\mu(n)} -{}\nonumber\\&& - 2(n-2)\epsilon^{\mu\lambda} H_{\alpha}^{\ \lambda}y^{\alpha}y^{\nu(2)}y^{\mu(n-3)}\bar y^{\dot\mu(n)}\,,\\
    T_3^{\lambda(2),\nu(2),\mu(n-2)|\dot\mu(n)} & = & 4n\, h^{\nu}_{\ \dot\alpha}\wedge h^{\lambda\dot\mu}\, y^{\lambda} y^{\nu} y^{\mu(n-2)}\bar y^{\dot\alpha}\bar y^{\dot\mu(n-2)} +{}\nonumber\\& &+ 2n(n-2)\, h^{\mu}_{\ \dot\alpha}\wedge h^{\lambda\dot\mu} y^{\lambda}y^{\nu(2)}y^{\mu(n-3)}\bar y^{\dot\alpha}\bar y^{\dot\mu(n-1)}\,,\\
    T_4^{\lambda(2),\nu(2),\mu(n-2)|\dot\mu(n)} & = & 4n\, \epsilon^{\nu\lambda}\, h_{\alpha}^{\ \dot\mu}\wedge h^{\lambda}_{\ \dot\beta}\,y^{\alpha}y^{\nu}y^{\mu(n-2)}\bar y^{\dot\beta}\bar y^{\dot\mu(n-1)} +{}\nonumber\\& & + 2n(n-2)\, \epsilon^{\mu\lambda}\,h_{\alpha}^{\ \dot\mu}\wedge h^{\lambda}_{\ \dot\beta}\,y^{\alpha}y^{\nu(2)}y^{\mu(n-3)}\bar y^{\dot\beta}\bar y^{\dot\mu(n-1)}.
    \end{eqnarray}
\end{subequations}
and
\begin{subequations}
    \begin{align}
        \Delta_{N=\bar N}\Phi_\mathrm{A}&= \mathcal{P}_\mathrm{A}(T_1+T_2+T_3+T_4) = n^2\Phi_\mathrm{A},\\
        \Delta_{N=\bar N}\Phi_\mathrm{B}&= \mathcal{P}_\mathrm{B}(T_1+T_2+T_3+T_4) =(2n^2+5n+4)\Phi_\mathrm{B}\neq 0,\\
        \Delta_{N=\bar N}\Phi_\mathrm{C}&= \mathcal{P}_\mathrm{C}(T_1+T_2+T_3+T_4) = 2(n^2+n+1)\Phi_\mathrm{C} -2n(n+1)\bar\Phi_\mathrm{C}.
    \end{align}
\end{subequations}
We observe that the only way for some of $\Phi_\text{A,B,C}$ to be in $H^2(\sigma_-)$ is at $n=0$. But in the diagonal sector with $N=\bar N = n$ this implies $N=\bar N=0$. This case extends formula (\ref{weyl})
to the spin-one $y,\bar y$-independent sector.
The analysis of the anti-holomorphic part $\bar\Phi$ is analogous. The resulting cohomology parameterizes
the spin-one field strength, \ie Faraday field strength.

\subsubsection{$H^2(\sigma_-)$ in the near-diagonal sector $|N-\bar N|=2$}
In the near-diagonal sector a subtlety considered in Section 7.2.3  takes place. We should search for a kernel of $\Delta$ in the form of a linear combination of the two-forms lying under the diagonal and above the diagonal. Our strategy is to act separately on the general holomorphic (\ref{phi-2-holomorphic}) and anti-holomorphic (\ref{phi-2-antiholomorphic}) two-forms placed below the diagonal $N=\bar N-2$ and then  determine which two-forms are in $\ker(\Delta)$. (The computation  with $N>\bar N$ only differs  by the complex conjugation.)

We start with the general holomorphic two-form below the diagonal
\begin{equation}\label{Phi(n-1,n+1)}
    \Phi_{(n-1,n+1)}^{\lambda(2)|\nu(2)|\mu(n-3)|\dot\mu(n+1)}(y,\bar y)=H^{\lambda\lambda}y^{\nu(2)}y^{\mu(n-3)}\bar y^{\dot\mu(n+1)}.
\end{equation}
Firstly, we set $n\geq 4$ considering the cases of $n\leq 3$, that are special in
 our computation scheme, because $n-3$ is the number of indices $\mu$, later. This yields
\begin{equation}
    \Delta_{N=\bar N-2}\Phi_{(n-1,n+1)}(y,\bar y) = \underbrace{\Delta_{N<\bar N- 2}\Phi_{(n-1,n+1)}(y,\bar y)}_{T_1(y,\bar y)} +\underbrace{y^\alpha y^\beta\bar\dd_{\dot\alpha}\bar\dd_{\dot\beta}h_\beta^{\ \dot\beta}D_\alpha^{\ \dot\alpha}\Phi_{(n-1,n+1)}(y,\bar y)}_{T_2(y,\bar y)}\,.
\end{equation}
The first term is computed the same way as in (\ref{DeltaTheta>>2}) giving
\begin{multline}
    T_1^{\lambda(2),\nu(2),\mu(n-3)|\dot\mu(n+1)}= (n+1)^2 \, H^{\lambda\lambda}y^{\nu}y^{\nu}y^{\mu(n-3)}\bar y^{\dot\mu(n+1)} -4 H^{\nu\lambda} y^{\lambda}y^{\nu}y^{\mu(n-3)}\bar y^{\dot\mu(n+1)} -\\
    -2(n-3)\,H^{\mu\lambda}y^{\lambda}y^{\nu(2)}y^{\mu(n-4)}\bar y^{\dot\mu(n+1)} +(n-1)\, H^{\lambda\lambda}y^{\nu(2)} y^{\mu(n-3)}\bar y^{\dot\mu(n+1)}.
\end{multline}
The computation of the additional term $T_2(y,\bar y)$ yields
\begin{multline}
    T_2^{\lambda(2),\nu(2),\mu(n-3)|\dot\mu(n+1)} = y^\alpha y^\beta\bar\dd_{\dot\alpha}\bar\dd_{\dot\beta}h_\beta^{\ \dot\beta}D_\alpha^{\ \dot\alpha}H^{\lambda\lambda}y^{\nu(2)}y^{\mu(n-3)}\bar y^{\dot\mu(n+1)} =\\
    =-2n(n+1)\, h_{\beta}^{\ \dot\mu}\wedge h^{\lambda\dot\mu} y^{\lambda}y^{\beta}y^{\nu(2)}y^{\mu(n-3)}\bar y^{\dot\mu(n-1)}=\\
    =-n(n+1)\, \bar H^{\dot\mu\dot\mu}y^{\lambda(2)}y^{\nu(2)}y^{\mu(n-3)}\bar y^{\dot\mu(n-1)}.
\end{multline}

Projection onto the irreducible parts $A$, $B$ and $C$ yields
\begin{subequations}
\begin{align}
    \Delta \Phi_{\mathrm{A}(n-1,n+1)} &= \underbrace{n(n+1)}_{\alpha(n)}\Phi_{\mathrm{A}(n-1,n+1)} \underbrace{ - n(n+1)}_{\beta(n)}\bar\Phi_{\mathrm{A}(n+1,n-1)}\,,\label{Phi_A-1/2cocycle}\\
    \Delta \Phi_{\mathrm{B}(n-1,n+1)} &= (n^2+5n-4)\Phi_{\mathrm{B}(n-1,n+1)}\,,\\
    \Delta \Phi_{\mathrm{C}(n-1,n+1)} &= (n^2+2n-1)\Phi_{\mathrm{C}(n-1,n+1)}\,.
\end{align}
\end{subequations}Let us stress that the complex conjugation denoted by $\dagger$ swaps dotted and undotted indices
\begin{equation}
    (y^\alpha)^\dagger = \bar y^{\dot\alpha}, \quad\quad (H^{\alpha\alpha})^\dagger = \bar H^{\dot\alpha\dot\alpha}
\end{equation}
and relates $\Phi$ and $\bar \Phi$ in the following way:
\begin{equation}
    (\Phi_{\mathrm{A,B,C}(n-1,n+1)})^\dagger = \bar\Phi_{\mathrm{A,B,C}(n+1,n-1)}.
\end{equation}
The computation for the complex-conjugated objects $\bar\Phi_{\mathrm{A,B,C}(n+1,n-1)}$ is
analogous giving
\begin{subequations}
\begin{align}
    \Delta \bar\Phi_{\mathrm{A}(n+1,n-1)} &= n(n+1)\,\bar\Phi_{\mathrm{A}(n+1,n-1)} - n(n+1)\,\Phi_{\mathrm{A}(n-1,n+1)}\,,\label{barPhi_A-1/2cocycle}\\
    \Delta \bar\Phi_{\mathrm{B}(n+1,n-1)} &= (n^2+5n-4)\bar\Phi_{\mathrm{B}(n+1,n-1)}\,,\\
    \Delta \bar\Phi_{\mathrm{C}(n+1,n-1)} &= (n^2+2n-1)\bar\Phi_{\mathrm{C}(n+1,n-1)}\,.
\end{align}
\end{subequations}
From (\ref{Phi_A-1/2cocycle}) and (\ref{barPhi_A-1/2cocycle}) we observe that there is a non-trivial 2-cocycle
\begin{equation}\label{E_A-cocycle}
    \mathcal{E}_\mathrm{A}=\Phi_{\mathrm{A}(n-1,n+1)}+\bar\Phi_{\mathrm{A}(n+1,n-1)}
    = \mathcal{E}^\mathrm{A}_{\mu(n+1),\dot\mu(n+1)}\,\Big(H^{\mu\mu}y^{\mu(n-1)}\bar y^{\dot\mu(n+1)} + \bar H^{\dot\mu\dot\mu}y^{\mu(n+1)}\bar y^{\dot\mu(n-1)}\Big)
\end{equation}
with arbitrary coefficients $\mathcal{E}^\mathrm{A}_{\mu(n+1),\dot\mu(n+1)}(x)$. This answer
agrees with the analysis of Section 7.2.3. Indeed, the coefficients on the \rhs of (\ref{Phi_A-1/2cocycle})
coincide up to a sign $\alpha(n) = -\beta(n)$, and by (\ref{Y=X+barX}) of Section 7.2.3 this implies a non-trivial 2-cocycle (\ref{E_A-cocycle}).

This cocycle represents the traceless part of the free Fronsdal HS equations.

The irreducible representations of types $(B)$ and $(C)$ do not contribute to cohomology since they are not in $\ker(\Delta)$ (recall that we are assuming $n\geq 4$).

Now consider the cases of $n=1,2,3$. Computing the action of the Laplace operator on the following
objects:
\begin{subequations}
\begin{align}
    \Phi_{(0,2)}^{\lambda\lambda|\dot\mu\dot\mu}(y,\bar y) &= H^{\lambda\lambda}\bar y^{\dot\mu}\bar y^{\dot\mu},\\
    \Phi_{(1,3)}^{\lambda\lambda|\mu|\dot\mu(3)}(y,\bar y) &= H^{\lambda\lambda}y^\mu\bar y^{\dot\mu(3)},\\
    \Phi_{(2,4)}^{\lambda\lambda|\nu\nu|\dot\mu(4)}(y,\bar y) &= H^{\lambda\lambda}y^\nu y^\nu\bar y^{\dot\mu(4)}\,,
\end{align}
\end{subequations}
it is not difficult to obtain
\begin{subequations}\label{n=1,2,3first}
\begin{align}
    (\Delta\Phi_{(0,2)})^{\lambda\lambda|\dot\mu\dot\mu}(y,\bar y) &= 2 H^{\lambda\lambda}\bar y^{\dot\mu(2)} - 2\bar H^{\dot\mu\dot\mu}y^{\lambda(2)},\\
    (\Delta\Phi_{(1,3)})^{\lambda\lambda|\mu|\dot\mu(3)}(y,\bar y) &= 8H^{\lambda\lambda}y^\mu\bar y^{\dot\mu(3)} - 2H^{\mu\lambda}y^\lambda\bar y^{\dot\mu(3)} - 6\bar H^{\dot\mu\dot\mu}y^{\lambda(2)}y^\mu\bar y^{\dot\mu},\\
    (\Delta\Phi_{(2,4)})^{\lambda\lambda|\nu\nu|\dot\mu(4)}(y,\bar y) &= 16H^{\lambda\lambda}y^{\nu\nu}\bar y^{\dot\mu(4)} - 4H^{\nu\lambda}y^\lambda y^\nu\bar y^{\dot\mu(4)} - 12\bar H^{\dot\mu\dot\mu}y^{\lambda(2)}y^{\nu(2)}\bar y^{\dot\mu(2)}.
\end{align}
\end{subequations}
We see that these results for $n=1,2,3$ extend the traceless part of the Fronsdal
cohomology (\ref{E_A-cocycle}) to spins $s=2,3,4$.

It remains to analyze the case of anti-holomorphic two-form below the diagonal $N=\bar N-2$
\begin{equation}\label{phi_(n+1,n-1)}
    \bar\Phi_{(n-1,n+1)}^{\dot{\nu}(2)|\mu(n-1)|\dot\lambda(2)|\dot\mu(n-1)}(y,\bar y) = \bar H^{\dot\lambda\dot\lambda}y^{\mu(n-1)}\bar y^{\dot\nu(2)}\bar y^{\dot\mu(n-1)}.
\end{equation}
Unlike Eq.~(\ref{Phi(n-1,n+1)}), the number of indices $\mu$ and $\dot\mu$  in  (\ref{phi_(n+1,n-1)}) is $n-1$, not $n-3$. Hence, there is no need to consider  separately the
 cases of $n\geq 4$ and $n\leq 3$. Instead, we set $n\geq 2$ and then analyze the $n=1$ case separately.

Let $n\geq2$. The action of the corresponding Laplace operator on (\ref{phi_(n+1,n-1)}) yields
\begin{equation}
    \Delta_{N=\bar N-2}\bar\Phi_{(n-1,n+1)}(y,\bar y) = \underbrace{\Delta_{N<\bar N+2}\bar\Phi_{(n-1,n+1)}(y,\bar y)}_{T_3(y,\bar y)} +\underbrace{y^\alpha y^\beta\bar\dd_{\dot\alpha}\bar\dd_{\dot\beta}h_\beta^{\ \dot\beta}D_\alpha^{\ \dot\alpha}\bar\Phi_{(n-1,n+1)}(y,\bar y)}_{T_4(y,\bar y)}.
\end{equation}
The computation is completely analogous to that for the holomorphic two-form.
After projecting onto the irreducible components it gives
\begin{subequations}
\begin{align}
    \Delta \bar\Phi_{\mathrm{A}(n-1,n+1)} &= (n^2+n-4)\bar\Phi_{\mathrm{A}(n-1,n+1)}\,,\\
    \Delta \bar\Phi_{\mathrm{B}(n-1,n+1)} &= \underbrace{(n^2+4n-1)}_{\alpha(n)}\bar\Phi_{\mathrm{B}(n-1,n+1)} \underbrace{ - (n^2+4n-1)}_{\beta(n)}\Phi_{\mathrm{B}(n+1,n-1)}\label{barPhi_B-1/2cocycle}\,,\\
    \Delta \bar\Phi_{\mathrm{C}(n-1,n+1)} &= (n^2+2n+1)\bar\Phi_{\mathrm{C}(n-1,n+1)}\,.
\end{align}
\end{subequations}
Applying  once again the result (\ref{Y=X+barX}) of Section 7.2.3 to (\ref{barPhi_B-1/2cocycle}), on the \rhs of
 which the coefficients coincide up to a sign, $\alpha(n) = -\beta(n)$, we obtain the 2-cocycle of the form
\begin{multline}\label{E_B-cocycle}
    \mathcal{E}_\mathrm{B}=\Phi_{\mathrm{B}(n+1,n-1)}+\bar\Phi_{\mathrm{B}(n-1,n+1)} =\\
    = \mathcal{E}^\mathrm{B}_{\mu(n-1)\dot\mu(n-1)}\,\Big(H_{\nu\nu}y^{\nu(2)} y^{\mu(n-1)}\bar y^{\dot\mu(n-1)} + \bar H_{\dot\nu\dot\nu}y^{\mu(n-1)}\bar y^{\dot\nu(2)}\bar y^{\dot\mu(n-1)}\big)\,,
\end{multline}
that represents the trace part of the Fronsdal equations.

Having considered
$n\geq 2$, now  consider the case of $n=1$. Computation of the action of the Laplace operator on the following two-form:
\begin{equation}
    \bar\Phi_{(0,2)}^{\dot\lambda\dot\lambda|\dot\mu\dot\mu}(y,\bar y) = \bar H^{\dot\lambda\dot\lambda}\bar y^{\dot\mu}\bar y^{\dot\mu}
\end{equation}
yields
\begin{equation}\label{barPhi_(0,2)Weylcontribution}
    (\Delta\bar\Phi_{(0,2)})^{\dot\lambda\dot\lambda|\dot\mu\dot\mu}(y,\bar y) = 4 \bar H^{\dot\lambda\dot\lambda}\bar y^{\dot\mu(2)}-4\bar H^{\dot\mu\dot\lambda}\bar y^{\dot\lambda}\bar y^{\dot\mu} - 2 H_{\alpha\alpha}y^{\alpha}y^{\alpha}\epsilon^{\dot\mu\dot\lambda}\epsilon^{\dot\mu\dot\lambda}.
\end{equation}
After projecting onto the irreducible components, we find that the case of $n=1$ extends the trace part of the Fronsdal cocycle $\mathcal{E}_\mathrm{B}$ (\ref{E_B-cocycle}) to spin $s=2$. In addition, (\ref{barPhi_(0,2)Weylcontribution}) also contributes to the antiholomorphic part of the Weyl cocycle represented by the second term on the \rhs of (\ref{weyl}). The holomorphic part of the latter lies in the opposite (complex-conjugated) region, in which the analysis is completely analogous. This 2-cocycle represents the Weyl tensor for the linearized gravity ($s=2$) in $AdS_4$.

This completes the analysis of $H^2(\sigma_-)$ in the near-diagonal sector $N=\bar N \pm 2$.

\subsection{Summary for bosonic $H^{0,1,2}(\sigma_-)$}
Here we collect the  final results for the cocycles associated with the bosonic HS gauge parameters, fields and
field equations in $AdS_4$.

Recall that $H^0(\sigma_-)$ represents  parameters of the differential HS gauge symmetries. It is spanned by the zero-forms
\begin{equation}
    F(y,\bar y|\,x) = F_{\alpha(n)\,\dot\alpha(n)}(x)\,y^{\alpha(n)}\bar y^{\dot\alpha(n)}\,,\qquad
    n\in\mathbb{N}_0\,.
\end{equation}

$H^1(\sigma_-)$ represents  the dynamical HS fields. For the bosonic HS fields in $AdS_4$ it is spanned by the two 1-cocycles $\phi(y,\bar y\,|x)$ and $\phi^\text{tr}(y,\bar y\,|x)$ corresponding,
respectively, to the traceless and trace components of the original Fronsdal field in the metric formalism:
\begin{subequations}
\begin{align}
    \phi(y,\bar y\,|x) = h^{\mu\dot\mu}\,\partial_\mu\bar\partial_{\dot\mu}\, F_1(y,\bar y\,|x),\\
    \phi^\text{tr}(y,\bar y\,|x) = h_{\mu\dot\mu}\,y^\mu \bar y^{\dot\mu}\, F_2(y,\bar y\,|x),
\end{align}
\end{subequations}
where $F_{1,2}(y,\bar y\,|x)$ are $(N,\bar N)$-diagonal, that is
\begin{equation}
\label{diag}
    \left(y^\alpha\frac{\partial}{\partial y^\alpha} - \bar y^{\dot\alpha}\frac{\partial}{\partial \bar y^{\dot\alpha}}\right)F_{1,2}(y,\bar y\,|x) = 0.
\end{equation}

Finally, $H^2(\sigma_-)$, which represents gauge invariant differential operators
 on the bosonic HS fields,
are spanned by three different 2-cocycles: the so-called Weyl cocycle $W(y,\bar y\,|x)$ and two irreducible
components of the Fronsdal cocycle $\mathcal{E}_\mathrm{A}(y,\bar y\,|x)$ (\ref{E_A-cocycle})
and $\mathcal{E}_\mathrm{B}(y,\bar y\,|x)$ (\ref{E_B-cocycle}). The latter correspond to the
$\lhs$'s of the dynamical equations for the fields of spin $s>1$ (spin $s\leq1$ field equations are in the zero-form
sector of unfolded equations \cite{Vasiliev:1988sa}). Note that these cocycles are real since they contain
 equal numbers of dotted and undotted indices.
\begin{subequations}\label{bosonic-H2-genfunc}
\begin{align}
    W(y,\bar y\, |x) &= H^{\mu\nu} \dd_\mu\dd_\nu C(y,0\, |x) + \bar H^{\dot\mu\dot\nu} \bar\dd_{\dot\mu}\bar\dd_{\dot\nu} C(0,\bar y\, |x)\,,\\
    \mathcal{E}_\mathrm{A}(y,\bar y\,|x) &= \Big(H^{\mu\nu}\partial_\mu\partial_\nu + \bar H^{\dot\mu\dot\nu}\bar\partial_{\dot\mu}\bar\partial_{\dot\nu}\Big) C^\text{diag}(y,\bar y\,|x)\,,\\
    \mathcal{E}_\mathrm{B}(y,\bar y\,|x) &= \Big(H^{\mu\nu} y_\mu y_\nu + \bar H^{\dot\mu\dot\nu}\bar y_{\dot\mu}\bar y_{\dot\nu}\Big) C^\text{diag}(y,\bar y\,|x),
\end{align}
\end{subequations}
where $C^\text{diag}(y,\bar y)$  obey (\ref{diag}).

\section{Fermionic HS fields in $AdS_4$}\label{Fermionic HS fields}
So far, we considered the bosonic case with even grading $G=|N-\bar N|$.
By (\ref{Ggo}) odd $G$ corresponds to fields of half-integer spins, \ie
 oddness of $G$ determines the field statistics.

To extend the results for $H^p(\sigma_-)$ to fermionic fields, we first define the operator $\sigma_-$ on multispinors of odd ranks. In the fermionic case, the lowest possible odd grading is $G=|N-\bar N|=1$. This means that previously unique lowest grading line on the $(N,\bar{N})$-plane splits into two separate lines $N - \bar{N} = \pm 1$. Therefore, the definition of $\sigma_{-}$ and its conjugated $\sigma_{+}$ depends on the lowest grading line. We define the action of $\sigma_-$ to vanish on the both lines. In all other gradings, $\sigma_\pm$ is defined analogously to the bosonic case. Namely,
\begin{subequations}
    \begin{align}
        \sigma_-\omega(y,\bar y) &:= i\, \bar y^{\dot\alpha}h^\alpha_{\ \dot\alpha}\partial_\alpha\ \omega(y,\bar y), \quad\quad \text{at } N\geq\bar N+3\,,\\
        \sigma_-\omega(y,\bar y) &:=  i\, y^{\alpha}h_\alpha^{\ \dot\alpha}\bar\partial_{\dot\alpha}\ \omega(y,\bar y), \quad\quad \text{at } N\leq\bar N-3\,.
    \end{align}
\end{subequations}
Analogously, the operator $\sigma_+$ is defined as
\begin{subequations}
    \begin{align}
        \sigma_+\omega(y,\bar y) &:= -i\, y^\alpha D_\alpha^{\ \dot\alpha}\bar{\partial}_{\dot\alpha}\ \omega(y,\bar y), \quad\quad\text{at } N\geq\bar N+3\,,\\
        \sigma_+\omega(y,\bar y) &:=  -i\, \bar y^{\dot\alpha}D^\alpha_{\ \dot\alpha}\partial_\alpha\ \omega(y,\bar y), \quad\quad \text{at } N\leq\bar N-3\,.
    \end{align}
\end{subequations}
For the lowest grading  lines $N-\bar N=1$ and $N-\bar N=-1$,  $\sigma_+$ is defined as in the sectors $N\geq\bar N+3$ and $N\leq\bar N-3$, respectively.

Notice that the action of the fermionic Laplace operator is analogous to that of the bosonic one
(\ref{bosonicLaplaceAction}) with the
grading shifted by one, $\Delta^\text{fermionic}_G=\Delta^\text{bosonic}_{G-1}$, except for the lowest grading. The final result is
\begin{subequations}
   \ls  \begin{align}
       \bullet \quad \Delta^\text{fermionic}_{N>\bar N+3}&= \Delta^\text{bosonic}_{N>\bar N+2} = N(\bar N+2) + y^\beta\dd_\alpha h^\alpha_{\ \dot\gamma} D_\beta^{\ \dot\gamma} + \bar y^{\dot\alpha}\bar\dd_{\dot\beta}h_{\gamma\dot\alpha}D^{\gamma\dot\beta}\,,\\
        \bullet \quad \Delta^\text{fermionic}_{N=\bar N+3}&= \Delta^\text{bosonic}_{N=\bar N+2} = \Delta_{N>\bar N+2} + \bar y^{\dot\alpha}\bar y^{\dot\beta}\dd_\alpha\dd_\beta h^\beta_{\ \dot\beta}D^\alpha_{\ \dot\alpha}\,,\\
        \bullet \quad \Delta^\text{fermionic}_{N=\bar N+1}&= \bar y^{\dot\alpha}\bar\dd_{\dot\beta}h_{\gamma\dot\alpha} D^{\gamma \dot\beta}-\bar y^{\dot\alpha} y^\beta \dd_\alpha\bar\dd_{\dot\beta}h^\alpha_{\ \dot\alpha} D_\beta^{\ \dot\beta}\,, \\
    \bullet \quad \Delta^\text{fermionic}_{N=\bar N-1}&= y^\alpha \dd_\beta h_{\alpha\dot\gamma}D^{\beta \dot\gamma} - y^\alpha\bar y^{\dot\beta}\dd_\beta\bar\dd_{\dot\alpha}h_\alpha^{\ \dot\alpha}D^\beta_{\ \dot\beta}\,. 
    \end{align}
\end{subequations}

This allows us do deduce the fermionic cohomology from the bosonic one arriving at the following
final results.

\subsection{Fermionic $H^0(\sigma_-)$}
The space $H^0(\sigma_-)$ for fermionic HS fields is spanned by two independent zero-forms with $N-\bar N = \pm1:$ \begin{equation}
    H^0(\sigma_-) = \Big\{F(y,\bar y\,|x) + \bar{F}(y,\bar y\,|x) = F_{\alpha(n+1),\dot\alpha(n)}(x)\,y^{\alpha(n+1)}\bar y^{\dot\alpha(n)}
    +\bar{F}_{\alpha(n),\dot\alpha(n+1)}(x)\,y^{\alpha(n)}\bar y^{\dot\alpha(n+1)}\Big\}\,.
\end{equation}
Recall that, by Theorem 3.1, $H^0(\sigma_-)$ represents parameters of differential HS gauge transformations.

\subsection{Fermionic $H^1(\sigma_-)$}
In the bosonic case, we had two physically different cocycles in $H^1$ (\ref{bosonic-H1-explicit})
 corresponding to traceless $\phi(y,\bar y\,|x)$ and trace $\phi^\text{tr}(y,\bar y\,|x)$ parts of
  the Fronsdal field.
These belong to the diagonal $N=\bar N$.

For the fermionic case, the situation is almost analogous. The lowest grading is now $G = |N-\bar N| = 1$. So, in this sector, there are four (not two) different 1-cocycles inherited  from the bosonic case: $\psi$, $\psi^\text{tr}$, $\bar \psi$ and $\bar \psi^\text{tr}$ and two additional cocycles, $\psi^{ext}$, $\bar{\psi}^{ext}$, given by
\begin{subequations}
    \begin{align}
    \psi(y,\bar y\,|x) &= \psi_{\mu(n+2),\dot\mu(n+1)}(x)\,h^{\mu\dot\mu}\,y^{\mu(n+1)}\bar y^{\dot\mu(n)},\\
    \bar \psi(y,\bar y\,|x) &= \bar\psi_{\mu(n+1),\dot\mu(n+2)}(x)\,h^{\mu\dot\mu}\,y^{\mu(n)}\bar y^{\dot\mu(n+1)},\\
    \psi^\text{tr}(y,\bar y\,|x) &= \psi^\text{tr}_{\mu(n),\dot\mu(n-1)}(x)\,h_{\nu\dot\nu}\,y^\nu y^{\mu(n)}\bar y^{\dot\nu}\bar y^{\dot\mu(n-1)},\\
    \bar\psi^\text{tr}(y,\bar y\,|x) &= \bar\psi^\text{tr}_{\mu(n-1),\dot\mu(n)}(x)\,h_{\nu\dot\nu}\,y^\nu y^{\mu(n-1)}
    \bar y^{\dot\nu}\bar y^{\dot\mu(n)},\\
    \psi^{ext}(y,\bar y\,|x) &= \psi^{ext}_{\mu(n),\dot\mu(n+1)}(x)\,h_{\nu}{}^{\dot\mu}\,y^\nu y^{\mu(n)}\bar y^{\dot\mu(n)},\\
    \bar{\psi}^{ext}(y,\bar y\,|x) &= \bar{\psi}^{ext}_{\mu(n+1),\dot\mu(n)}(x)\,h^{\mu}{}_{\dot\nu}\, y^{\mu(n)}\bar y^{\dot\nu}\bar y^{\dot\mu(n)} 
\end{align}
\end{subequations}
with a non-negative integer $n$ (positive for $\psi^\text{tr}$ and $\bar \psi^\text{tr}$).
Cocycles $\psi$, $\psi^\text{tr}$ and $\psi^{ext}$ belong to the upper near-diagonal line $N=\bar N+1$,
whereas $\bar\psi$, $\bar\psi^\text{tr}$ and $\bar{\psi}^{ext}$ belong to the lower near-diagonal line $N=\bar N -1$. All of them have grading $G=1$. $\psi $ and $\bar \psi$ are mutually conjugated.

These results can be put into the following concise  form
\begin{subequations}
    \begin{align}
    \psi(y,\bar y\,|x) &= h^{\mu\dot\mu}\,\partial_\mu\bar\partial_{\dot\mu}\,F_1(y,\bar y\,|x),\\
    \psi^\text{tr}(y,\bar y\,|x) &= h_{\mu\dot\mu}\,y^\mu\bar y^{\dot\mu}\,F_2(y,\bar y\,|x),\\
    \psi^{ext}(y,\bar y\,|x) &= h_{\nu}{}^{\dot\mu}\, y^\nu \bar{\partial}_{\dot\mu}\,F_3(y,\bar y\,|x), 
\end{align}
\end{subequations}
where $F_{1,2}(y,\bar y\,|x)$  are of the homogeneity degree $N-\bar N=1$, \ie
\begin{equation}
    \left(y^\alpha\frac{\partial}{\partial y^\alpha} - \bar y^{\dot\alpha}\frac{\partial}{\partial \bar y^{\dot\alpha}}\right) F_{1,2}(y,\bar y\,|x) =  F_{1,2}(y,\bar y\,|x)\,,
\end{equation}
and $F_{3}(y,\bar y\,|x)$ is of the homogeneity degree $N-\bar N=-1$, \ie
\begin{equation}
    \left(y^\alpha\frac{\partial}{\partial y^\alpha} - \bar y^{\dot\alpha}\frac{\partial}{\partial \bar y^{\dot\alpha}}\right) F_{3}(y,\bar y\,|x) =  -F_{3}(y,\bar y\,|x)\,.
\end{equation}
For $\bar\psi$, $\bar\psi^\text{tr}$ and $\bar{\psi}^{ext}$ the results are analogous except that $\bar F_{1,2}(y,\bar y\,|x)$ have degree $N -\bar N = -1$ and $\bar F_{3}(y,\bar y\,|x)$ has degree $N -\bar N = 1$.

\subsection{Fermionic $H^2(\sigma_-)$}
According to the same arguments, the fermionic $H^2$ is almost analogous to the bosonic one. Recall that the bosonic 2-cocycles are represented by three different two-forms: Weyl tensor, traceless and traceful parts of the generalized Einstein tensors (near diagonal, $G=3$).

The fermionic Weyl cohomology is given by the same formula as the bosonic one:
\begin{equation}
    W^\text{ferm}(y,\bar y\,|x) = H^{\mu\nu}\,\partial_\mu\partial_\nu C(y,0\,|x) + \bar H^{\dot\mu\dot\nu}\,\bar\partial_{\dot\mu}\bar\partial_{\dot\nu}C(0,\bar y\,|x),
\end{equation}
where
$C(y,0\,|x)$ and $C(0,\bar y\,|x)$ are polynomials of $y$ and $\bar y$, respectively.

The two bosonic Fronsdal cocycles (\ref{bosonic-H2-genfunc}) were represented by the two zero-forms
  $C^\text{diag}(y,\bar y)$ with the support on the diagonal $N=\bar N$. In the fermionic case
  the two Fronsdal
   cocycles split into four.
   The bosonic diagonal polynomial $C^\text{diag}(y,\bar y)$ is replaced by a pair of near-diagonal $C^\text{near-diag}(y,\bar y)$ and $\bar C^\text{near-diag}(y,\bar y)$ satisfying the relations
\begin{subequations}
    \begin{align}
    \left(y^\alpha\frac{\partial}{\partial y^\alpha} - \bar y^{\dot\alpha}\frac{\partial}{\partial \bar y^{\dot\alpha}}\right) C^\text{near-diag}(y,\bar y\,|x) &=  C^\text{near-diag}(y,\bar y\,|x),\\
    \left(y^\alpha\frac{\partial}{\partial y^\alpha} - \bar y^{\dot\alpha}\frac{\partial}{\partial \bar y^{\dot\alpha}}\right) \bar C^\text{near-diag}(y,\bar y\,|x) &= -\bar C^\text{near-diag}(y,\bar y\,|x)\,
\end{align}
\end{subequations}
These support the  fermionic 2-cocycles associated with the $\lhs$'s of the fermionic field equations
for spin $s\geq 3/2$ massless fields as follows
\begin{subequations}
    \begin{align}
    \mathcal{E}^\text{ferm}_\mathrm{A}(y,\bar y\,|x) &= \Big(H^{\mu\nu}\partial_\mu\partial_\nu + \bar H^{\dot\mu\dot\nu}\bar\partial_{\dot\mu}\bar\partial_{\dot\nu}\Big) C^\text{near-diag}(y,\bar y\,|x),\\
    \bar{\mathcal{E}}^\text{ferm}_\mathrm{A}(y,\bar y\,|x) &= \Big(H^{\mu\nu}\partial_\mu\partial_\nu + \bar H^{\dot\mu\dot\nu}\bar\partial_{\dot\mu}\bar\partial_{\dot\nu}\Big) \bar C^\text{near-diag}(y,\bar y\,|x),\\
    \mathcal{E}^\text{ferm}_\mathrm{B}(y,\bar y\,|x) &= \Big(H^{\mu\nu} y_\mu y_\nu + \bar H^{\dot\mu\dot\nu}\bar y_{\dot\mu}\bar y_{\dot\nu}\Big) C^\text{near-diag}(y,\bar y\,|x),\\
    \bar{\mathcal{E}}^\text{ferm}_\mathrm{B}(y,\bar y\,|x) &= \Big(H^{\mu\nu} y_\mu y_\nu + \bar H^{\dot\mu\dot\nu}\bar y_{\dot\mu}\bar y_{\dot\nu}\Big) \bar C^\text{near-diag}(y,\bar y\,|x)\,.
\end{align}
\end{subequations}

As in the case of $1$-forms, additional cocycles appear in the cohomology
\begin{subequations}
    \begin{align}
    \mathcal{E}^\text{ferm}_\mathrm{C}(y,\bar y\,|x) &= H^{\mu\nu}y_\mu\partial_\nu C^\text{near-diag}(y,\bar y\,|x), \\
    \bar{\mathcal{E}}^\text{ferm}_\mathrm{C}(y,\bar y\,|x) &= \bar H^{\dot\mu\dot\nu}\bar{y}_{\dot\mu}\bar\partial_{\dot\nu} \bar C^\text{near-diag}(y,\bar y\,|x)\,. 
\end{align}
\end{subequations}

\section{Conclusion}

In this paper, free  unfolded equations for massless HS fields are analyzed in detail in terms of $\sigma_-$ cohomology. This is done both in  flat space of arbitrary dimension in the tensor formalism for bosonic fields and in $AdS_4$ in the spinor formalism for both  bosonic and fermionic fields. Not surprisingly,  the final
results agree with those stated long ago in the original papers \cite{Vasiliev:1988sa, Lopatin:1987hz}.
Our aim is to present the detailed analysis of the $\sigma_-$ cohomology providing an exhaustive
proof of the so-called First On-Shell Theorem of the form of free unfolded HS equations, allowing the interested reader to check every step.

In the tensor case the full set of cohomology groups $H^p(\sigma_-)$  was found  both for the  groups $\mathrm{GL}(d)$ and $\mathrm{O}(d)$.  Our results for $\mathrm{GL}(d)$ and $p<3$ coincide with those found in \cite{Bekaert:2005vh}. For the $\mathrm{O}(d)$ case of traceless fields
lower cohomology groups matched against those in  \cite{Bekaert:2005vh, Skvortsov:2009nv, Barnich:2004cr}.
In $AdS_4$  we used spinor formalism  to analyze $H^{0,1,2}(\sigma_-)$ for both bosonic and fermionic HS fields. To the best of our knowledge such analysis was not available in the literature.

Practically, to compute $H^k(\sigma_-)$ in both $\mathrm{Mink}^d$ and $AdS_4$ cases we used the analogue of the Hodge theorem. Namely, the problem of finding the cohomology
$H^k(\sigma_-) = \mathrm{ker}(\sigma^{k}_-)/\mathrm{im}(\sigma^{k-1}_-)$ was reduced to
the calculation of
 the kernel of an appropriate positive-definite Laplace-Hodge operator $\Delta$
 invariant under the action of compact version of the space-time symmetry algebra.
 This technique was shown
 to be lucid and efficient. Having found the cohomology groups $H^k(\sigma_-)$ for $k=0,1,2$,
 in accordance with \cite{Shaynkman:2000ts} we obtained
  the exhaustive information about the differential HS  gauge parameters, dynamical HS gauge fields and their field equations. Thus, we have explicitly proven the so-called First On-Shell Theorem for bosonic HS fields in $\mathrm{Mink}^d$ (which case is straightforwardly extendable to
$AdS_d$)
and all massless fields in $AdS_4$.  The technique used in this paper can be further applied to the calculation of $H^p(\sigma_-)$ in the zero-form sector of HS fields studied in \cite{Shaynkman:2000ts, Gelfond:2003vh} that describes
dynamics of a scalar field and $s=1$ particle as well as to more general systems
considered in \cite{Gelfond:2013lba, Vasiliev:2015mka,Vasiliev:2018zer}.
One of the byproduct results of this paper is the interpretation of the matching between
$\sigma_-$ cohomology of the one-form sector against zero-form sector expressing the
matching between Bianchi identities in the two sectors.

\section*{Acknowledgement}
We are grateful to  Vyacheslav Didenko,  Anatoly Korybut, and Alexander Tarusov for helpful and stimulating discussions and Maxim Grigoriev for a useful comment. We are particularly grateful to
the referee for  useful suggestions and the correspondence, and to Yuri Tatarenko for 
pointing out missed fermionic cohomology in Section \ref{Fermionic HS fields}.
We acknowledge a partial support from the Russian Basic Research Foundation Grant No 20-02-00208.

\section*{Appendix A}
\subsection*{Index conventions}
\label{indices}

Since most of the tensors encountered in the course of this paper
 are Young tensors in the symmetric basis, it is convenient to accept the following notation.

A tensor without a certain type of index symmetry will be denoted as $T^{a|b|c|..}$, where the vertical line $|$ separates groups of indices  not related by any symmetries to each other.

A tensor that has a symmetric set of $n$ indices, say, $(a_1,a_2,\dots,a_n)$ will be denoted $T^{a(n)|...}\equiv T^{(a_1a_2\dots a_n)|...}$. A tensor corresponding to a certain Young diagram in the symmetric basis then has the form: $T^{a(n), b(m), c(k),...}$.

Symmetrization over $n$ indices is performed by the formula $\mathrm{Sym} = \frac{1}{n!}\sum_{\text{all permutations}}$.

We will denote all symmetrized  tensor indices by the same letter. For example,
\begin{equation}
T^{a(n)|a} \equiv \frac{1}{(n+1)!} \sum_{\sigma\in S_{n+1}} T^{\sigma(a_1..a_n|a_{n+1})}\,.
\end{equation}

In Section \ref{vectorcase} we omit $Y,Z,\theta$ and assume that all indices are contracted with the corresponding variables.

The rules for raising and lowering $\mathfrak{sl}(2)$-indices are
\begin{equation*}
    A_\alpha = A^\beta\epsilon_{\beta\alpha}, \quad\quad A^\alpha = \epsilon^{\alpha\beta}A_\beta, \quad\quad \epsilon_{\alpha\beta}\epsilon^{\gamma\beta} = \epsilon_\alpha^{\ \gamma} = \delta_\alpha^{\ \gamma} = -\epsilon^\gamma_{\ \alpha}
\end{equation*}
with
\begin{equation*}
    \epsilon_{\alpha\beta} = -\epsilon_{\beta\alpha},\quad\quad\quad \epsilon_{12} = 1.
\end{equation*}

\subsection*{Coefficients in the tensor form of the diagram $(n-1, m-1; p-2)$}
\label{coefficients_in_tensor_form}
{\large
\begin{spreadlines}{0.75em}
\begin{equation}
\alpha_1 = \tfrac{(-3 + n) (-2 + n) (-1 + n) (-3 + d + 2 n)}{(-5 + d +
   m + n) (-4 + d + m + n) (-4 + d + 2 n) (d^2 + d (-8 + m + 3 n) -
   2 (-8 + m + m^2 + 7 n - 3 m n))},
 \end{equation}
\end{spreadlines}
}%
{\large
\begin{spreadlines}{0.75em}
\begin{equation}
\alpha_2 = \tfrac{(-2 + n) (-1 + n)}{(-4 + d + m + n) (-4 + d + 2 n)},
\end{equation}
\end{spreadlines}
}%

{\large
\begin{spreadlines}{0.75em}
\begin{equation}
\alpha_3 = -\tfrac{2 (-2 + n) (-1 + n) (-3 + d + 2 n)}{(-4 + d + m +
    n) (-4 + d + 2 n) (d^2 + d (-8 + m + 3 n) -
    2 (-8 + m + m^2 + 7 n - 3 m n))},
\end{equation}
\end{spreadlines}
}%

{\large
\begin{spreadlines}{0.75em}
\begin{equation}
\alpha_4 = \tfrac{2(1 + m - n) (-2 + n) (-1 + n) (-3 + d + 2 n)}{(-5 +
   d + m + n) (-4 + d + m + n) (-4 + d + 2 n) (d^2 +
   d (-8 + m + 3 n) - 2 (-8 + m + m^2 + 7 n - 3 m n))},
\end{equation}
\end{spreadlines}
}%

{\large
\begin{spreadlines}{0.75em}
\begin{equation}
\alpha_5 = \tfrac{ (-1 + n) (-3 + d + m + n)}{(-4 + d + m + n) (-4 + d + 2 n)},
\end{equation}
\end{spreadlines}
}%

{\large
\begin{spreadlines}{0.75em}
\begin{equation}
\alpha_6 = -\tfrac{(-2 + d + 2 m) (-1 + n)}{(-4 + d + 2 n) (d^2 +
    d (-8 + m + 3 n) - 2 (-8 + m + m^2 + 7 n - 3 m n))},
\end{equation}
\end{spreadlines}
}%

{\large
\begin{spreadlines}{0.75em}
\begin{equation}
\alpha_7 = \tfrac{(-1 + n) (-3 + d + 2 n) (16 + d^2 - 7 m - 9 n +
   2 d (-4 + m + n) + (m + n)^2)}{(-5 + d + m + n) (-4 + d + m +
   n) (-4 + d + 2 n) (d^2 + d (-8 + m + 3 n) -
   2 (-8 + m + m^2 + 7 n - 3 m n))},
\end{equation}
\end{spreadlines}
}%
{\large
\begin{spreadlines}{0.75em}
\begin{equation}
\alpha_8 = -\tfrac{(-1 + n)}{(-4 + d + m + n)},
\end{equation}
\end{spreadlines}
}%
{\large
\begin{spreadlines}{0.75em}
\begin{equation}
\alpha_9 = \tfrac{2 (-1 + n) (-3 + d + 2 n)}{(-4 + d + m + n) (d^2 +
   d (-8 + m + 3 n) - 2 (-8 + m + m^2 + 7 n - 3 m n))},
\end{equation}
\end{spreadlines}
}%
{\large
\begin{spreadlines}{0.75em}
\begin{equation}
\alpha_{10} = -\tfrac{(-3 + d + 2 n)}{
 d^2 + d (-8 + m + 3 n) - 2 (-8 + m + m^2 + 7 n - 3 m n)},
\end{equation}
\end{spreadlines}
}%
{\large
\begin{spreadlines}{0.75em}
\begin{equation}
\alpha_{11} = -\tfrac{(-4 + d + 2 m) (-1 + n) (-3 + d + 2 n)}{(-5 + d + m + n) (-4 + d + m + n) (d^2 + d (-8 + m + 3 n) - 2 (-8 + m + m^2 + 7 n - 3 m n))},
\end{equation}
\end{spreadlines}
}%
{\large
\begin{spreadlines}{0.75em}
\begin{equation}
\alpha_{12} = \tfrac{(14 + d^2 - 6 n + 2 m (-5 + m + n) +
   d (-8 + 3 m + n))}{(-4 + d + m + n) (d^2 + d (-8 + m + 3 n) -
   2 (-8 + m + m^2 + 7 n - 3 m n))}\,.
\end{equation}
\end{spreadlines}
}%


\end{document}